\pdfoutput=1

\newif\ifdraft\drafttrue
\newif\iffull\fulltrue

\documentclass[acmsmall,nonacm]{acmart}

\acmJournal{PACMPL}
\acmVolume{1}
\acmNumber{CONF} 
\acmArticle{1}
\acmYear{2018}
\acmMonth{1}
\acmDOI{} 
\startPage{1}

\setcopyright{none}

\usepackage{microtype}
\usepackage{mathpartir}
\usepackage{graphicx}
\usepackage{wrapfig}
\usepackage[all]{xy}
\usepackage{bussproofs}
\usepackage{amssymb}
\usepackage{amsmath}
\usepackage{etex}
\usepackage{comment}
\usepackage{braket}
\usepackage{algorithm}
\usepackage[noend]{algpseudocode}
\usepackage{color}
\usepackage{booktabs}
\usepackage{stmaryrd}

\newcommand{\tetsuya}[1]{}
\newcommand{\mg}[1]{}
\newcommand{\gb}[1]{}
\newcommand{\jh}[1]{}
\newcommand{\sk}[1]{}

\newcommand{\ol}{\overline}

\newcommand{\FinSet}{\mathbf{Fin}}

\newcommand{\Meas}{\mathbf{Meas}}
\newcommand{\Span}{\mathbf{Span}}

\newcommand{\wCPO}{\omega\mathbf{CPO}}

\newcommand{\BB}{\mathbb B}

\newcommand{\EE}{\mathbb E}

\newcommand{\NN}{\mathbb N}

\newcommand{\RR}{\mathbb R}

\newcommand{\ZZ}{\mathbb Z}

\newcommand{\labeleq}[1]{\mathbin{\ooalign{$=$\crcr\hss\raisebox{1.0ex}[0pt][0pt]{\scalebox{0.7}[0.7]{#1}}\hss}}}

\newcommand{\inverse}[1]{{#1}^{\hspace{-0.15em} - \hspace{-0.1em}1}\hspace{-0.1em}}

\newcommand{\lrangle}[1]{\langle #1\rangle}
\newcommand{\interpret}[1]{{\llbracket {#1} \rrbracket}}
\newcommand{\Rinterpret}[1]{{\llparenthesis {#1} \rrparenthesis}}

\newcommand{\defeq}{\stackrel{\mathsf{def}}{=}}

\newcommand{\PARAGRAPH}[1]{\iffull\subsection{#1}\else\vspace*{0.15em}\noindent\textbf{#1.}\fi}

\newtheorem{proposition}{Proposition}[section]
\newtheorem{remark}{Remark}[section]

\begin{document}

\title{Approximate Span Liftings}
\subtitle{Compositional Semantics for Relaxations of Differential Privacy}

\author{Tetsuya Sato}
\affiliation{%
  \institution{University at Buffalo, SUNY}
  \city{Buffalo}
  \state{New York}
  \country{USA}
}

\author{Gilles Barthe}
\affiliation{%
  \institution{IMDEA Software Institute}
  \city{Madrid}
  \country{Spain}
}

\author{Marco Gaboardi}
\affiliation{%
  \institution{University at Buffalo, SUNY}
  \city{Buffalo}
  \state{New York}
  \country{USA}
}

\author{Justin Hsu}
\affiliation{%
  \institution{Cornell University}
  \city{Ithaca}
  \state{New York}
  \country{USA}
}

\author{Shin-ya Katsumata}
\affiliation{%
  \institution{National Institute of Informatics}
  \city{Tokyo}
  \country{USA}
}

\begin{abstract}
We develop new abstractions for reasoning about relaxations of differential
privacy: \emph{R\'enyi differential privacy}, \emph{zero-concentrated
differential privacy}, and \emph{truncated concentrated differential privacy},
which express different bounds on statistical divergences between two output
probability distributions. In order to reason about such properties
compositionally, we introduce \emph{approximate span-lifting}, a novel
construction extending the
approximate relational lifting approaches previously developed for standard
differential privacy to a more general class of divergences, and also to
continuous distributions. As an application, we develop a program logic
based on approximate span-liftings capable of proving relaxations of
differential privacy and other statistical divergence properties.
\end{abstract}



\maketitle

\section{Introduction}
\label{sec:introduction}
Differential privacy~\citep{DworkMcSherryNissimSmith2006} is a strong,
statistical notion of data privacy that has attracted the attention of
theoreticians and practitioners alike.  One reason for its success is that
differential privacy can often be proved \emph{compositionally}, enabling easy
construction of new private algorithms and making formal verification practical.
By now, researchers have developed a wide variety of programming languages and
program analysis tools to prove differential
privacy~\citep{AH17,Barthe:2015:HAR:2676726.2677000,Barthe:2012:PRR:2103656.2103670,GHHNP13,McSherry:2009:PIQ:1559845.1559850,ReedPierce10,Winograd-Cort:2017:FAD:3136534.3110254,zhang2016autopriv}
(\citet{Murawski:2016:2893582} provide a recent survey).

Seeking more refined composition properties, researchers have
recently proposed new relaxations of differential privacy: \emph{R\'enyi
  differential privacy} (RDP)~\citep{Mironov17}, \emph{zero-concentrated
  differential privacy} (zCDP)~\citep{BunS16}, and \emph{truncated concentrated
differential privacy} (tCDP)~\citep{BDRS18}. Roughly speaking, standard
differential privacy requires a bound on the magnitude of a random variable
measuring the privacy loss, while RDP, zCDP, and tCDP model finer bounds on the
\emph{moments} of this random variable. (Recall that the first moment of a
random variable is its average value, and the second moment of a random variable
is its variance.) These relaxations capture fine-grained aspects of the privacy
loss, enabling more precise privacy analyses and allowing algorithms to add less
random noise to achieve the same privacy level.

Each of RDP, zCDP, and tCDP is defined in terms \emph{R\'enyi
divergences}~\citep{Renyi1961}, sophisticated distances on distributions
originating from information theory. Inspiring our work, Barthe and Olmedo
previously developed abstractions for reasoning about a family of divergences
called $f$-\emph{divergences} as part of their work on the program logic
$f$pRHL~\citep{BartheOlmedo2013,olmedo2014approximate}. In particular, the
semantic foundation of $f$pRHL is a \emph{2-witness relational lifting} for
$f$-divergences, which tracks the $f$-divergence between relates pairs of
distributions.
However, this framework is not sufficient to establish about our target
properties for two reasons. First, R\'enyi divergences are not
$f$-divergences,\footnote{%
  For instance, all $f$-divergences are jointly convex while R\'enyi divergences
are only quasi-convex~\citep{6832827}.}
while zCDP and tCDP are properly described as \emph{supremums} of \emph{R\'enyi
divergences}, rather than single divergences. As a result, these relaxations of
differential privacy cannot be described in terms of $f$-divergences, nor
captured in $f$pRHL. Accordingly, we develop new relational liftings
supporting significantly more general divergences, allowing direct reasoning
about RDP, zCDP, and tCDP.

A further challenge is that 2-witness relational liftings to date have only been
proposed for discrete distributions, while many algorithms satisfying
relaxations of differential privacy---indeed, the motivating examples of such
algorithms---sample from continuous distributions, such as the Gaussian
distribution.
Handling these distributions requires a careful treatment of measure theory.
\citet{Sato2016MFPS} has previously considered a different semantic
model for standard differential privacy over continuous distributions using
\emph{witness-free} relational lifting based on a categorical construction
called \emph{codensity lifting}~\citep{katsumata_et_al:LIPIcs:2015:5532}, but it
is not clear how to handle more general divergences with this method.


To overcome these difficulties, we generalize $2$-witness liftings in two
directions. First, we replace the notion of $f$-divergence with a more general
class of divergences, identifying the basic properties needed for compositional
reasoning. Second, we generalize these liftings to about continuous
probability measures.  The main challenge is establishing a sequential
composition principle---the continuous case introduces further measurability
requirements for composition.
Accordingly, we extend the structure of 2-witness liftings to a new notion
called \emph{approximate span-liftings}, which have the necessary data to ensure
closure under sequential composition. Finally, we specialize our general model
to R\'enyi divergence, divergences for zCDP, and divergences for tCDP,
establishing categorical properties needed to build approximate span-liftings.
As an extended application, we develop a relational program logic that can
verify differential privacy, RDP, zCDP, and tCDP within a single logic for
programs using discrete or continuous sampling, and interpret the logic via
approximate span-liftings.

After motivating the various relaxations of differential privacy and presenting
the key technical challenges (Section~\ref{sec:motivation}), and introducing
mathematical preliminaries (Section~\ref{sec:prelims}), we present our main
contributions.
\begin{itemize}
\item We identify a general class of divergences supporting basic properties
  composition properties, and we show that our class can model RDP, zCDP and
  tCDP (Section~\ref{sec:divergences}).
\item We extend $2$-witness relational liftings to the continuous case by
  introducing a novel notion of approximate span-lifting and showing how to
  translate composition properties of specific divergences  to their
  corresponding approximate span-liftings (Section~\ref{sec:span-liftings}).
\item We develop a program logic supporting four flavors of differential
  privacy---standard DP, RDP, zCDP, and tCDP---where programs may use both
  discrete and continuous random sampling, and show soundness
  (Section~\ref{sec:span-apRHL}). We demonstrate our logic on three
  examples (Section~\ref{sec:examples}).
\end{itemize}
We survey related work (Section~\ref{sec:rw}) and then conclude with promising
future directions (Section~\ref{sec:conclusion}).

\section{Background: Motivation and Technical Challenges}
\label{sec:motivation}
To better understand the key technical challenges, we first introduce relevant
background on privacy, divergences, and existing relational verification
techniques. 
For simplicity, in this section we consider probability distributions
which have associated density functions.
\subsection{Differential Privacy and its Relaxations}

We first introduce differential
privacy.  A \emph{randomized algorithm} is a
measurable function $\mathcal{A} \colon X \to \mathrm{Prob}(Y)$ from a set $X$
of inputs to the set $\mathrm{Prob}(Y)$ of \emph{probability distributions} on a
set $Y$ of outputs.
\begin{definition}[Differential Privacy (DP)~\citep{DworkMcSherryNissimSmith2006}]
  A randomized algorithm $\mathcal{A} \colon X \to \mathrm{Prob}(Y)$
  is \emph{$(\varepsilon,\delta)$-differentially private} w.r.t an adjacency
  relation $\Phi\subseteq X\times X$, if for any pairs of inputs $(x,x')\in
  \Phi$, and any measurable subset $S\subseteq Y$, we have 
  $
   \Pr[\mathcal{A}(x) \in S] \leq e^\varepsilon \Pr[\mathcal{A}(x') \in S] + \delta.
  $
\end{definition}
\begin{definition}[R\'enyi divergence~\citep{Renyi1961}]
  Let $\alpha > 1$. The \emph{R\'enyi divergence} of order $\alpha$
  between two probability distributions $\mu_1$ and $\mu_2$ on a measurable space $X$ is defined by:
\begin{equation}\label{eq:Renyi_divergences}
  D^{\alpha}_X(\mu_1||\mu_2) \defeq \frac{1}{\alpha - 1} \log \int_{X} \mu_2(x) \left( \frac{\mu_1(x)}{\mu_2(x)} \right)^\alpha~dx.
\end{equation}
\end{definition}
\begin{definition}[R\'enyi Differential Privacy (RDP)~\citep{Mironov17}]
A randomized algorithm $\mathcal{A}:X\to\mathrm{Prob}(Y)$ is 
\emph{$(\alpha,\rho)$-R\'enyi differentially private} w.r.t an adjacency relation $\Phi\subseteq X\times X$, if for any pairs of inputs $(x,x')\in \Phi$, we have 
$
D^{\alpha}_X(\mathcal{A}(x)||\mathcal{A}(y)) \leq \rho.
$
\end{definition}
\begin{definition}[zero-Concentrated Differential Privacy (zCDP)~\citep{BunS16}]
A randomized algorithm $\mathcal{A}:X\to\mathrm{Prob}(Y)$ is 
\emph{$(\xi,\rho)$-zero concentrated differentially private} w.r.t an adjacency relation $\Phi\subseteq X\times X$, if for any pairs of inputs $(x,x')\in \Phi$, we have
\begin{equation}\label{eq:zCDP_definition}
  \forall{\alpha > 1}.~ D^\alpha_Y(\mathcal{A}(x)||\mathcal{A}(x')) \leq \xi + \alpha\rho.
\end{equation}
\end{definition}

\begin{definition}[Truncated Concentrated Differential Privacy (tCDP)~\citep{BDRS18}]
A randomized algorithm $\mathcal{A}:X\to\mathrm{Prob}(Y)$ is 
\emph{$(\rho,\omega)$-truncated concentrated differentially private} w.r.t an adjacency relation $\Phi\subseteq X\times X$, if for any input pairs $(x,x')\in \Phi$, we have
\begin{equation}\label{eq:tCDP_definition}
  \forall{1 < \alpha < \omega}.~ D^\alpha_Y(\mathcal{A}(x)||\mathcal{A}(x')) \leq \alpha\rho.
\end{equation}
\end{definition}

While these notions may seem cryptic at first sight, they can all be understood
as bounds on the \emph{privacy loss}, defined for any two private inputs $x, x'$ by
\[
  \mathcal{L}^{x \to x'}(y) = \frac{\Pr[\mathcal{A}(x) = y]}{\Pr[\mathcal{A}(x') = y]} .
\]
Intuitively, the privacy loss measures how much information is revealed when the
output of a private algorithm is seen to be $y$.  While output values with a
high value of privacy loss are highly revealing---since they are far more likely
to result from a private input $x$ rather than a different private input
$x'$---if these outputs are only seen with very small probability, then their
influence can be discounted. Accordingly, the different privacy definitions
bound different functions of the privacy loss function, evaluated at some output
$y$ drawn from the output distribution of the private algorithm. The following
table summarizes these bounds. 

\begin{center}
  \begin{tabular}{ll}
    \toprule
    Privacy notion of $\mathcal{A}$
    &
    Bound on privacy loss $\mathcal{L}$
    \\
    \midrule
    $(\varepsilon,\delta)$-DP
    &
    $\Pr_{y \sim \mathcal{A}(x)} [ \mathcal{L}^{x \to x'}(y) \leq e^\varepsilon ] \geq 1 - \delta$
    \\
    $(\alpha,\rho)$-RDP
    &
    $\EE_{y \sim \mathcal{A}(x)} [\mathcal{L}^{x \to x'}(y)^\alpha] \leq e^{(\alpha - 1)\rho}$
    \\
    $(\xi,\rho)$-zCDP
    &
    $\forall{\alpha > 1}.~ \EE_{y \sim \mathcal{A}(x)} [\mathcal{L}^{x \to x'}(y)^\alpha] \leq e^{(\alpha - 1)(\xi + \alpha\rho)}$
    \\
    $(\omega,\rho)$-tCDP
    &
    $\forall{1 < \alpha < \omega}.~ \EE_{y \sim \mathcal{A}(x)} [\mathcal{L}^{x \to x'}(y)^\alpha] \leq e^{(\alpha - 1)\alpha\rho}$
    \\
    \bottomrule
  \end{tabular}
\end{center}

In particular, DP bounds the maximum value of the privacy loss,
$(\alpha,\cdot)$-RDP bounds the $\alpha$-moment, zCDP bounds all moments, and
$(\cdot, \omega)$-tCDP bounds the moments up to some cutoff $\omega$.  Many
conversions are known between these definitions; for instance, the relaxations
of RDP, zCDP, and tCDP are known to sit between $(\varepsilon, 0)$ and
$(\varepsilon, \delta)$-differential privacy in terms of expressivity, up to
some modification in the parameters.  While this means that RDP, zCDP, and tCDP
can sometimes be analyzed by reduction to standard differential privacy,
converting between the different notions requires weakening the parameters and
often the privacy analysis is simplest and most precise by working with RDP,
zCDP, or tCDP directly. For further details, the interested reader can refer to
the original papers~\citep{BunS16,Mironov17}.
%
%

A motivating example of a mechanism fitting these three definitions is the
\emph{Gaussian mechanism} and \emph{Sinh Normal mechanism}, which add noise
according to a Gaussian distribution and sinh-normal distribution over the
real numbers respectively. The distributions are generated by continuous
density functions.
\mg{I think we should give here the Gaussian mechanism as well as the sinhGauss
  mechanism and try to stress that these mechanisms are naturally
``continuous'' and it is unclear how to best approximate them.}
\tetsuya{I tried to revise the above, please check.}

%
\subsection{2-witness Relational Liftings for $f$-divergences in Discrete Case}

\citet{BartheOlmedo2013} observed that standard differential privacy can be
phrased in terms of a general class of divergences, called
$f$-\emph{divergences}.

\begin{definition}
A \emph{weight function} is a convex function $f \colon \RR_{\geq 0} \to
\RR$ continuous at $0$.\footnote{%
As is conventional~\citep{1705001_2006}, we exclude the condition $f(1) = 0$ from the
definition of weight function to support the exponential of R\'enyi divergence
of order $\alpha$.  We also assume $0f(a/0) = \lim_{t \to 0+} tf(a/t)$ for $a >
0$ and $0f(0/0) = 0$.}
\end{definition}

\begin{definition}[$f$-divergence]
For a weight function $f$, the \emph{$f$-divergence $\Delta^{f}$} between two
distributions $\mu_1, \mu_2$ over a measurable space $X$ is defined as
\begin{equation} \label{eq:f-divergences}
  \Delta^{f}_X(\mu_1,\mu_2) = \int_{X} \mu_2(x) f\left( \frac{\mu_1(x)}{\mu_2(x)} \right)~dx.
\end{equation}
\end{definition}

In particular, differential privacy can be modeled by the $f$-divergence
$\Delta^{\mathtt{DP}(\varepsilon)}$ with weight function
$\mathtt{DP}(\varepsilon)(t) = \max(0,1-e^\varepsilon
t)$~\citep{BartheOlmedo2013,olmedo2014approximate}.
For any randomized algorithm $\mathcal{A}:X\to\mathrm{Prob}(Y)$
and adjacency relation $\Phi\subseteq X\times X$, we have
\[
\mathcal{A} \text{ is } (\varepsilon,\delta)\text{-DP }
\mathbin\text{ iff }
(\text{for all } (x,x') \in \Phi,~ \Delta^{\mathtt{DP}(\varepsilon)}_Y(\mathcal{A}(x),\mathcal{A}(x')) \leq \delta ).
\]

To verify $f$-divergence properties of probabilistic programs, Barthe and Olmedo
introduced \emph{2-witness relational lifting} for $f$-divergences as a key
abstraction. This construction lifts a relation $R \subseteq
X \times Y$ over discrete sets $X,Y$ to a relation $R^{\sharp(f,\delta)}
\subseteq \mathrm{Dist}(X) \times \mathrm{Dist}(Y)$ over subprobability
distributions:\footnote{%
  In order to reason about possibly non-terminating programs, they work with an
  extension of $f$-divergence to \emph{subprobability distributions}.}
\begin{equation} \label{eq:2-witness_lifting}
R^{\sharp(f,\delta)}
=
\Set{(\mu_1,\mu_2) | \exists{\mu_L,\mu_R \in \mathrm{Dist}(R)}.~
\pi_1(\mu_L) = \mu_1,~ \pi_2(\mu_R) = \mu_2,~
\Delta^f_R(\mu_L,\mu_R) \leq \delta}.
\end{equation}
\iffull
Above,
$\pi_i(\mu)$ is the $i$-th marginal of $\mu$, that is, 
$(\pi_1(\mu))(x) = \sum_{y \in Y}\mu(x,y)$ and $(\pi_2(\mu))(y) = \sum_{x \in X}\mu(x,y)$.
\fi
The distributions $\mu_L$ and $\mu_R$ are called \emph{witness distributions},
since to show that two distributions are related by a lifting, one must show the
existence of two appropriate witnesses.

Barthe and Olmedo used these relational liftings as the foundation of their
relational program logic $f$pRHL. These liftings have several attractive
features. First, they reflect $f$-divergences:
\[
\mathrm{Eq}_X^{\sharp(f,\delta)} = \Set{(x,x) \mid x \in X}^{\sharp(f,\delta)} = \Set{(\mu_1,\mu_2)| \Delta^f_X(\mu_1,\mu_2) \leq \delta}.
\]
\iffull
So, they can be used to characterize differential privacy: a program
$\mathcal{A} \colon X \to \mathrm{Dist}(Y)$ is $(\varepsilon,\delta)$-differentially
private w.r.t. an adjacency relation $\Phi$, if
$(\mathcal{A}(x),\mathcal{A}(x'))\in
\mathrm{Eq}_Y^{\sharp(\mathtt{DP}({\varepsilon}),\delta)}$, for every $(x,x')
\in \Phi$.
\fi
Second, 2-witness liftings satisfy various composition properties, enabling
clean verification of probabilistic programs.  However, this construction works
only in the discrete case---all subprobability distributions are over countable
discrete sets---and the logic $f$pRHL cannot reason about programs that sample
from continuous distributions, like the Gaussian distribution.

\subsection{Challenge 1: Handling Richer Divergences}

Much like standard differential privacy can be viewed in terms of
$f$-divergences, we would like to view RDP, zCDP, and tCDP as bounds on more
general divergences. A natural candidate for R\'enyi differential privacy is
R\'enyi divergence $D^\alpha$, as in its original definition. Indeed, we have:
\[
\mathcal{A} \text{ is } (\alpha,\rho)\text{-RDP }
\quad\text{iff}\quad
(\text{for all } (x,x') \in \Phi,~ D^\alpha_Y(\mathcal{A}(x)||\mathcal{A}(x')) \leq \rho ).
\]
However, the R\'enyi divergence $D^\alpha(\mu_1||\mu_2)$ of order $\alpha$
is not an $f$-divergence, and so it does not fit in the 2-witness
lifting framework.
Likewise, zCDP~\citep{BunS16} and tCDP~\citep{BDRS18} can be defined via uniform
bounds on families of R\'enyi divergence:
\begin{equation}
\label{eq:divergence_zCDP}
\Delta^{\mathtt{zCDP}(\xi)}_X(\mu_1,\mu_2) = \sup_{1<  \alpha}
\frac{1}{\alpha}\left(D^\alpha_X(\mu_1||\mu_2) - \xi\right)
\quad\text{for $0 \leq \xi$},
\end{equation}
\begin{equation}
\label{eq:divergence_tCDP}
\Delta^{\omega-\mathtt{tCDP}}_X(\mu_1,\mu_2) = \sup_{1 < \alpha < \omega}
\frac{1}{\alpha}\left(D^\alpha_X(\mu_1||\mu_2)\right)
\quad\text{for $1 < \omega$},
\end{equation}
letting us reformulate zCDP and tCDP as
\begin{align*}
\mathcal{A} \text{ is } (\xi,\rho)\text{-zCDP }
\mathbin\text{ iff }&
(\text{for all } (x,x') \in \Phi,~ \Delta^{\mathtt{zCDP}(\xi)}_Y(\mathcal{A}(x), \mathcal{A}(x')) \leq \rho )\\
\mathcal{A} \text{ is } (\rho,\omega)\text{-tCDP }
\mathbin\text{ iff }&
(\text{for all } (x,x') \in \Phi,~ \Delta^{\omega-\mathtt{tCDP}}_Y(\mathcal{A}(x), \mathcal{A}(x')) \leq \rho ).
\end{align*}
These divergences are also not $f$-divergences. Furthermore, the RDP, zCDP and
tCDP divergences may take negative values when applied to sub-probability
distributions, which can arise from probabilistic computations that may not
terminate with probability $1$.  Accordingly, we generalize the notion of
divergence to go beyond $f$-divergences and also to handle sub-probability
distributions.
Starting from families of real valued functions from pairs of distributions, we
introduce basic properties needed to give good composition properties for their
corresponding liftings.

\subsection{Challenge 2: Extending 2-witness Liftings to the Continuous Case}
\label{sec:cont-case-motiv}
In order to support natural examples for RDP, zCDP, and tCDP, we need a framework
supporting continuous distributions, such as Gaussian, Laplace, and sinh-normal
distributions.
\mg{Also here I would mention Gauss and sinhGauss}
\tetsuya{I tried to revise the above, please check.}
Unfortunately, extending 2-witness relational liftings to the continuous case
presents further technical challenges related to composition. The relational
lifting $(-)^{\sharp(\mathtt{DP}({\varepsilon}),\delta)}$ for standard
differential privacy satisfies a sequential composition principle:
\[
	\AxiomC{$(f,g) \colon R \to S^{\sharp(\mathtt{DP}({\varepsilon_1}),\delta_1)}$ is a relation-preserving map.}
	\UnaryInfC{$(f^\sharp,g^\sharp) \colon R^{\sharp(\mathtt{DP}({\varepsilon_2}),\delta_2)} \to S^{\sharp(\mathtt{DP}({\varepsilon_1+\varepsilon_2}),\delta_1+\delta_2)}$ is a relation-preserving map.}
	\DisplayProof
\]
Here, $f^\sharp$ and $g^\sharp$ are the Kleisli liftings of $f$ and $g$ with
respect to the monad $\mathrm{Dist}$ of (discrete) subprobability distributions;
this composition property gives 2-witness relational liftings a
\emph{graded monad} structure~\citep{Katsumata2014PEM,FujiiKM16FoSSaCS}, highly
useful for compositional reasoning.
Since 2-witness lifting is defined through the existence of
witness distributions, for any $(d_1,d_2) \in
R^{\sharp(\mathtt{DP}({\varepsilon_2}),\delta_2)}$, we then need witness
distributions showing $(f^\sharp(d_1),g^\sharp(d_2)) \in
S^{\sharp(\mathtt{DP}({\varepsilon_1+\varepsilon_2}),\delta_1+\delta_2)}$.
In the discrete case, these witnesses can be constructed in two steps:
\begin{enumerate}
  \item
    For any $(x,y) \in R$, there exist witnesses $d'_L,d'_R \in
    \mathrm{Dist}(S)$ proving $(f(x),g(y)) \in
    S^{\sharp(\mathtt{DP}({\varepsilon_1}),\delta_1)}$.
    By applying the axiom of choice, we obtain a selection function
    \[
      \langle l_1, l_2 \rangle \colon R \to \Set{(d'_L,d'_R) \mid \Delta_S^{\mathtt{DP}({\varepsilon_1})}(d'_L,d'_R) \leq \delta_1}
    \]
  \item
    For any witnesses $d_L,d_R \in \mathrm{Dist}(R)$ proving $(d_1,d_2) \in
    R^{\sharp(\mathtt{DP}({\varepsilon_2}),\delta_2)}$,
    $(l_1^\sharp(d_L),
    l_2^\sharp(d_R))$ is a pair of witness distributions proving
    $(f^\sharp(d_1),g^\sharp(d_2)) \in
    S^{\sharp(\mathtt{DP}({\varepsilon_1+\varepsilon_2}),\delta_1+\delta_2)}$ by
    composability of $\Delta^{\mathtt{DP}(\varepsilon)}$.
\end{enumerate}
The first step is problematic to extend to the continuous case because the
witness-selecting functions $l_1$ and $l_2$ obtained by the axiom of choice may
not be measurable---the Kleisli extensions $l_1^\sharp$ and $l_2^\sharp$ in the
second step may not be well-defined in the continuous case.  

To resolve this difficulty, we introduce a novel notion of \emph{approximate
span-liftings}. The key idea is that morphisms between span-liftings carry a
built-in measurable witness selection function, making it unnecessary to use
the axiom of choice when proving sequential composition.

\mg{Here I would add a section, saying a bit more in details how we
  address these challenges. I would mention some of the properties we
  have and why they are important. I think that we should stress how
  having the right definition is important to 1) internalize the
  different primitives, 2) achieve composition and
  3) achieve group privacy. In particular, we should mention that we
  get some of them for free, once we have the right definition of
  lifting --- we should be careful because thse indeed follow from the
  same properties that we have in the different notions of privacy. We
  should also mention the fact that we will use a graded monad as a
  main reasoning tool. 
We should also talk about the properties
  we have: additivity, composability, approximability, etc. and what
  is their role, and how they could be of more interest in general. 
}

\section{Mathematical Preliminaries}
\label{sec:prelims}

\PARAGRAPH{Measure Theory}
We briefly review some definitions from measure theory;
readers should consult a textbook for more details~\citep{rudin-real-complex}.
Given a set $X$, a $\sigma$-\emph{algebra} on $X$ is a collection $\Sigma$ of
subsets of $X$ including the empty set, closed under complements, countable
unions, and countable intersections; a \emph{measurable space} $X$ is a
set $|X|$ with a $\sigma$-algebra $\Sigma_X$, called the measurable sets.
A countable set $X$ yields the \emph{discrete} measurable space where all subsets are measurable: $\Sigma_X = 2^X$.

A map $f \colon X \to Y$ between measurable spaces is
\emph{measurable} if $\inverse{f}(A) \in \Sigma_X$ for all $A \in \Sigma_Y$.
Any subset $S$ of measurable space $X$ forms a \emph{subspace} where the
$\sigma$-algebra is given by $\Sigma_S =\Set{A \cap S \mid A \in \Sigma_X}$.
$\Sigma_S$ is given as the coarsest one making the inclusion map
$S \hookrightarrow X$ measurable.

A \emph{measure} on a measurable space is a map $\mu \colon\!\Sigma_X \to \RR_{\geq 0} \cup \{\infty\}$
such that $\mu(\varnothing) = 0$ and $\mu(\cup_i X_i) = \sum_i \mu(X_i)$ for any countable family of
disjoint measurable sets $X_i$.
Measures with $\mu(X) = 1$ are called \emph{probability measures}, and measures with $\mu(X) \leq 1$ are called \emph{subprobability measures}.

For any pair of subprobability measures $\mu_1$ on $X$ and $\mu_2$ on $Y$,
the \emph{product measure} $\mu_1 \otimes \mu_2$ of $\mu_1$ and $\mu_2$ is
the unique measure on $X \times Y$ satisfying $(\mu_1 \otimes \mu_2)(A\times B)  = \mu_1(A) \cdot \mu_2(B)$.

For any measurable space $X$ and element $x \in X$, we write $\mathbf{d}_x$ for the Dirac measure on $X$ centered at $x$, defined as $\mathbf{d}_x(A) = 1$ if $x \in A$, and $\mathbf{d}_x(A) = 0$ otherwise.

Measurable spaces and measurable functions form a category $\Meas$; this
category has all limits and colimits, and finite products distribute over finite
coproducts.  We denote by $\FinSet$ the full subcategory of $\Meas$ consisting
of all finite discrete spaces.

\PARAGRAPH{The Sub-Giry Monad}
The \emph{sub-Giry monad} $\mathcal{G}$ is the subprobabilistic variant of the
Giry monad~\citep{Giry1982}.

\iffull
\begin{definition}
  The sub-Giry monad $(\mathcal{G},\eta,(-)^\sharp)$ over $\Meas$ is defined as
  follows:
  \begin{itemize}
  \item For any $X\in\Meas$, the measurable space
    $\mathcal{G}X$ is the set of subprobability measures ( measures whose mass is equal or less than $1$ ) on $X$
    equipped with the coarsest $\sigma$-algebra induced by the
    evaluation functions $\mathrm{ev}_A \colon \mathcal{G}X \to [0,1]$
    defined by $\nu\mapsto\nu(A)$ ($A \in \Sigma_X$).
  \item For each $f \colon X \to Y$ in $\Meas$,
    $\mathcal{G}f \colon \mathcal{G}X \to \mathcal{G}Y$ is defined by
    $(\mathcal{G}f)(\mu) = \mu(\inverse{f}(-))$.
  \item The unit $\eta$ is defined by the Dirac distributions
    $\eta_X (x) = \mathbf{d}_{x}$.
  \item The Kleisli extension $f^\sharp \colon \mathcal{G}X \to \mathcal{G}Y$ of $f \colon X \to \mathcal{G}Y$ is given
    by for any $\mu \in \mathcal{G}X$ and $A \in \Sigma_Y$, $f^\sharp(\mu)(A) = \int_{X} f(x)(A)~d\mu(x)$.
  \end{itemize}
\end{definition}

The sub-Giry monad satisfies useful properties for interpreting
probabilistic programs. It is commutative and strong with respect to
the Cartesian products of $\Meas$, where the double strength
$\mathrm{dst}_{X,Y} \colon \mathcal{G}({X})\times \mathcal{G}({Y})
\Rightarrow \mathcal{G}({X} \times {Y})$ is given by the product
measures $\mathrm{dst}_{X,Y}(\nu_1,\nu_2) = \nu_1 \otimes \nu_2$. The
double strength is used to define semantics for composition and to
interpret typing contexts. Additionally, the sub-Giry monad provides
a structure to interpret loops.
Namely, we can introduce an $\wCPO_\bot$ structure over measurable functions of
type $X \to \mathcal{G}(Y)$ with the following order:\footnote{%
	This ordering gives an $\wCPO_\bot$-enrichment of the Kleisli category $\Meas_\mathcal{G}$, 
	which is equivalent to the partial additivity of stochastic relations \citep{Panangaden1999171}.}
\[
f \sqsubseteq g \iff \forall{x \in X, B \in \Sigma_Y}. f(x)(B) \leq g(x)(B)
\quad (f,g \colon X \to \mathcal{G}(Y) \text{ in } \Meas).
\]
\fi

\PARAGRAPH{Graded Monads}
A \emph{graded monad}~\citep{Katsumata2014PEM,FujiiKM16FoSSaCS} is a monad refined by indices from
a monoid.  Let $A = (A,\cdot,1_A,{\preceq})$ be a preordered monoid.
An $A$-graded monad on a category $\mathbb{C}$ consists of 
\begin{itemize}
\item a family $\{T_{e}\}_{e \in M}$ of endofunctors $T_{e}$ on $\mathbb{C}$,
\item a morphism $\eta_X \colon X \to T_{1_A} X$ for $X \in \mathbb{C}$ (unit),
\item a morphism $({-})^{e_1 \sharp e_2} \colon
\mathbb{C}(X,T_{e_2} Y) \to \mathbb{C}(T_{e_1}X,T_{e_1e_2} Y)$ for $X, Y \in
\mathbb{C}$ and $e_1, e_2 \in A$ (Kleisli lifting),
\item a family $\{{\sqsubseteq}^{e_1, e_2} \}_{e_1 \preceq  e_2}$ of natural transformations
$\sqsubseteq^{e_1, e_2} \colon T_{e_1} \Rightarrow T_{e_2}$ (inclusion)
\end{itemize}
satisfying the following compatibility condition: for any $f \colon X \to T_{e_1} Y$ and $g \colon Y \to T_{e_2} Z$,
\begin{gather*}
{\sqsubseteq}^{(e_2 e_1), (e_2 e_3)}_Z \circ f^{e_2 \sharp e_1} = ({\sqsubseteq}^{e_1, e_2}_Y \circ f)^{e_2 \sharp e_3}, \quad
f^{e_3 \sharp e_1} \circ {\sqsubseteq}^{e_2, e_3}_X = {\sqsubseteq}^{(e_2 e_1), (e_3 e_1)}_Y \circ f^{e_2 \sharp e_1},\\
f^{1\sharp e_1} \circ \eta_X = f, \quad
\eta_X^{1\sharp e} = \mathrm{id}_{T_eX}, \quad
(g^{e_1 \sharp e_2} \circ f)^{e_0 \sharp e_1 e_2} =  g^{e_0 e_1 \sharp e_2} \circ f^{e_0 \sharp e_1}.
\end{gather*}

A typical way of constructing a graded monad is by refining a plain monad with
indices. An \emph{$A$-graded lifting} of a monad $(T,\eta^T,({-})^\sharp)$ on
$\mathbb{D}$ along a functor $U \colon \mathbb{C} \to \mathbb{D}$ is an
$A$-graded monad $\{T_{e}\}_{e \in A}$ on $\mathbb{C}$ satisfying $U\circ
{T_{e}} = T\circ U$, $U(f^{e_2\sharp e_1}) = (Uf)^\sharp$, $U(\eta_D) =
\eta^T_{UD}$, and $U({\sqsubseteq}^{e_1, e_2}_D) = \mathrm{id}_{TUD}$.
The functor $U$ erases the grading of $T_{e}$, yielding the original
(plain) monad $T$.

\PARAGRAPH{The Category of Spans on Measurable Spaces}
To extend the relational lifting approach to the continuous setting, we work
with the category of \emph{spans}, whose objects generalize relations by taking
arbitrary functions in place of projections. Morphisms between spans will encode
the information needed to ensure good compositional behavior.
\begin{definition}
  The category $\Span(\Meas)$ of spans in $\Meas$ consists of:
  \begin{itemize}
  \item Objects $(X,Y,\Phi,\rho_1,\rho_2)$ given by span $X
    \xleftarrow[]{\rho_1} \Phi \xrightarrow[]{\rho_1} Y$ in $\Meas$.
  \item Morphisms
    $(X,Y,\Phi,\rho_1,\rho_2) \to (Z,W,\Psi,\rho'_1,\rho'_2)$ given by
    triples $(h,k,l)$ of 
    morphisms $h \colon X \to Z$, $k \colon Y \to W$, and
    $l \colon \Phi \to \Psi$ in $\Meas$ satisfying
    $h \circ \rho_1 = \rho'_1 \circ l$ and
    $k \circ \rho_2 = \rho'_2 \circ l$.
  \end{itemize}
\end{definition}
For simplicity, we often denote a $\Span(\Meas)$-object $(X,Y,\Phi,\rho_1,\rho_2)$ by $\Phi$.
The category $\Span(\Meas)$ has several useful properties. First, the category
has binary products:
\[
 (X,Y,\Phi,\rho_1,\rho_2) \mathbin{\dot\times} (Z,W,\Psi,\rho'_1,\rho'_2)
=(X\times Z,Y \times W,\Phi\times\Psi,\rho_1\times\rho_1',\rho_2\times\rho_2').
\]
We will frequently use two notions of pairing on functions.  Let $f_1 \colon X
\to Y$, $f_2 \colon X \to W$, we have $\langle f_1, f_2 \rangle \colon X \to Y
\times W$ and $f_1 \times f_2 \colon X \times X \to Y \times W$.  As functions,
$\langle f_1, f_2 \rangle$ takes a single input $x$ and returns a pair
$(f_1(x),f_2(x))$. On the other hand, $f_1 \times f_2$ take a pair of inputs
$(x,y)$ and returns $(f_1(x),f_2(y))$.

The category $\Span(\Meas)$ also has coproducts:
\[
 (X,Y,\Phi,\rho_1,\rho_2) \mathbin{\dot+} (X',Y',\Phi',\rho_1',\rho_2')
=(X + X',Y + Y',\Phi + \Phi',\rho_1 + \rho_1',\rho_2 + \rho_2'). 
\]
Standard binary relations can be interpreted as spans. For $X,Y \in
\Meas$, any binary relation $\Phi \subseteq |X| \times |Y|$ determines a span $X
\xleftarrow[]{\pi_1} \Phi \xrightarrow[]{\pi_2} Y$ in $\Meas$,
where $\pi_1$ and $\pi_2$ are projections, and $\Phi$ is regarded as a subspace of $X \times Y$.

Finally, relation-preserving maps can be interpreted as morphisms of spans.
Consider two binary relations $\Phi \subseteq |X| \times |Y|$ and $\Psi \subseteq |Z| \times |W|$,
and suppose that they are interpreted as spans $(X,Y,\Phi,\pi_1,\pi_2)$ and $(Z,W,\Psi,\pi_1,\pi_2)$
as above.
If $f \colon X \to Z$ and $g \colon Y \to W$ in $\Meas$
satisfy $(f(x),g(y)) \in \Psi$ for any $(x,y) \in \Phi$, then we have the following morphism
\[
(f,g,f\times g|_\Phi) \colon (X,Y,\Phi,\pi_1,\pi_2) \to (Z,W,\Psi,\pi_1,\pi_2) \quad \text{ in } \Span(\Meas)
\]
where $f\times g|_\Phi$ is the restriction of $f\times g$ on $\Phi$ (we often write just $f \times g$).
These features are crucial to interpret probabilistic program logics,
as we will see in Section \ref{sec:span-apRHL}.

\section{General Statistical Divergences}
\label{sec:divergences}
Now that we have covered the preliminaries, our goal is
to build a suitable graded monad on $\Span(\Meas)$---this will be our
abstraction for relational reasoning about divergences. We
proceed in two stages.  In this section, we introduce a general class
of \emph{divergences}, real-valued functions on two measures
over the same space. Then, we identify important composition
properties inspired from analogous properties of
$f$-divergences~\citep{BartheOlmedo2013,1705001_2006}. We will leverage
these properties to give a graded monad structure on $\Span(\Meas)$
capturing these divergences in the next section.  We write $\ol\RR$
for the set $\RR\cup\{-\infty,+\infty\}$ of extended reals.  We regard
both $\ol\RR$ and $\RR_{\geq 0}$ as partially ordered additive
monoids. For the former one, the addition is extended by $\infty + (-\infty) = -\infty$.
\begin{definition}
A \emph{divergence} is a family $\Delta = \{\Delta_X\}_{X \in \Meas}$ of functions
\[
\Delta_X \colon |\mathcal{G}X| \times |\mathcal{G}X| \to \ol{\RR}.
\]
\end{definition}
To describe composition of divergences, it is useful to work with indexed
families of divergences; often, two divergences can be combined to give a new
divergence with different indices. For instance, the notion of zCDP can be
characterized by the family $\{\Delta^{\mathrm{zCDP}(\xi)}\}_{0 \leq \xi}$ of
divergences $\Delta^{\mathrm{zCDP}(\xi)}$ introduced in Section
\ref{sec:motivation} (Equation \ref{eq:divergence_zCDP}).
For this reason, we introduce the notion of \emph{graded families of divergences}.
\begin{definition}
Let $(A,\cdot,1_A,\preceq)$ be a preordered monoid.  An $A$-\emph{graded
family of divergences} is a family $\mathbf{\Delta} = \{\Delta^\alpha \}_{\alpha
\in A}$ such that
\[
	\alpha \preceq \beta \implies (\forall{X \in \Meas}.~\forall{\mu_1,\mu_2 \in \mathcal{G}X}.~ \Delta^\beta_X(\mu_1,\mu_2) \leq \Delta^\alpha_X(\mu_1,\mu_2)).
\]
\end{definition}
Note that the preorder on the grading is contravariant.
We will regard a divergence $\Delta$ as a
singleton-graded family $\{ \Delta \}$.

\subsection{Basic Properties of Divergences}
We define basic properties of graded families of divergences 
for given $(A,\cdot,1_A,\preceq)$.
\begin{definition}\label{def:div:properties}
  An $A$-graded family $\mathbf{\Delta} = \{\Delta^\alpha \}_{\alpha \in A}$ of divergences is:
  \begin{description}
    \item[reflexive:] if
      $\Delta^\alpha_X(\mu,\mu) \leq 0$.
    \item[functorial:] if 
      $\Delta^\alpha_Y(\mathcal{G}k(\mu_1),\mathcal{G}k(\mu_2)) \leq \Delta^\alpha_X(\mu_1,\mu_2)$
      for any $k\colon X \to Y$.
    \item[substitutive:] if 
      $\Delta^\alpha_Y(f^\sharp\mu_1,f^\sharp\mu_2) \leq \Delta^\alpha_X(\mu_1,\mu_2)$
      for any $f\colon X \to \mathcal{G}Y$.
    \item[additive:] if
      $\Delta^{\alpha\cdot\beta}_{X \times Y}(\mu_1 \otimes \mu_3, \mu_2 \otimes \mu_4) \leq \Delta^{\alpha}_X(\mu_1,\mu_2) + \Delta^{\beta}_Y(\mu_3,\mu_4)$.
    \item[continuous:] if
      $\Delta^\alpha_X(\mu_1,\mu_2)
      = \sup\Set{
      \Delta^\alpha_I (\mathcal{G}k(\mu_1),\mathcal{G}k(\mu_2)) \mid I \in \FinSet, k \colon X \to I}$.
    \item[composable:] if 
      $\Delta^{\alpha\cdot\beta}_Y(f^\sharp \mu_1, g^\sharp \mu_2) \leq
      \Delta^{\alpha}_X(\mu_1, \mu_2)  + \sup_{x \in X}\Delta^{\beta}_Y(f(x),
      g(x))$ for any $f,g\colon X \to \mathcal{G}Y$.
    \end{description}
  All functions are assumed to be measurable.
\end{definition}

These properties are inspired by properties from the literature on
$f$-divergences and differential privacy.
\iffull
For instance, substitutivity is the generalization of the usual notion of
\emph{data-processing inequality} for $f$-divergences~\citep[Chapter
2]{PardoVajda1997}, while functoriality is the special case where the
data-processing function is deterministic. These two properties are also known
in the differential privacy literature as \emph{resilience to
post-processing}~\citep[Proposition 2.1]{DworkRothTCS-042}, in the randomized and
deterministic case.
\fi
Composability corresponds to composition in differential privacy, which
states that we can adaptively compose two differentially private mechanisms.
Additivity corresponds to a simple instance of composition where the second
mechanism does not depend on the result of the first. 
Continuity is the generalization of the continuity of $f$-divergences~\citep[Theorem 16]{PardoVajda1997},
which approximates divergences of continuous distributions by divergences of discrete distributions.

Reflexivity and composability are key properties to give a structure of graded
monad.  Intuitively, reflexivity gives a unit, and composability gives a
(graded) Kleisli lifting. We also need additivity to give a \emph{strength} of
the graded monad, allowing a lifting on real-valued distributions---often
available from known results in probability theory---to be converted into a
lifting on distributions over larger spaces (e.g., program memories). In some
ways, composability is the key property: reflexivity is usually immediate, and
additivity is a consequence.

\begin{theorem}\label{thm:div:cont+comp->mon}
  An $A$-graded family $\mathbf{\Delta}$ is additive if it is continuous and composable.
\end{theorem}

Although these properties have been studied before in the discrete case, there
are subtleties when passing to our continuous ones. For example, in the case of
discrete distributions, additivity is an instance of
composability~\citep[Proposition 4]{BartheOlmedo2013}.  In the case of continuous
distributions, this may no longer hold. However, one can recover additivity from
composability by using a continuity property.

%

To prove composability, it is often easier to establish two other properties of
families of divergences first: approximability and finite-composability.  These
properties describe divergences that are well-behaved with respect to
discretization, in order to smoothly extend properties in the discrete case to
the continuous case.

\begin{definition}
  An $A$-graded family $\mathbf{\Delta} = \{\Delta^\alpha \}_{\alpha \in A}$ of divergences is:
  \begin{description}
    \item[approximable:]
      if for any $X \in \Meas$ and $I \in \FinSet$, $f,g \colon X \to \mathcal{G}I$, and
      $\mu_1,\mu_2 \in \mathcal{G}X$, there are $J_n \in \FinSet$ and
      $m^\ast_n \colon X \to J_n$ and $m_n \colon J_n \to X$ in $\Meas$ such that
      \[
        \Delta^\alpha_I(f^\sharp(\mu_1),g^\sharp(\mu_2))
        = \lim_{n \to \infty}
        \Delta^\alpha_{I}({(f \circ m_n \circ m^\ast_n)}^\sharp(\mu_1), {(g \circ m_n \circ m^\ast_n )}^\sharp(\mu_2)) .
      \]
    \item[finite-composable:]
      if for any $I,J \in \FinSet$, $f,g \colon I \to \mathcal{G}J$, and $d_1,d_2 \in \mathcal{G}I$,
      \[
        \Delta^{\alpha\cdot\beta}_J(f^\sharp d_1, g^\sharp d_2)
        \leq \Delta^{\alpha}_I(d_1, d_2)  + \sup_{i \in I}\Delta^{\beta}_J(f(i), g(i)) .
      \]
  \end{description}
\iffull
The function $m^\ast_n$ in the definition of the approximability of $\mathbf{\Delta}$ discretizes points in $X$ to $J_n$, and $m_n$ reconstructs points in $X$ from $J_n$.
Finite-composability of $\mathbf{\Delta}$ means the composability of $\mathbf{\Delta}$ in the discrete case.
\fi
\end{definition}

These properties allow us to extend composability of divergences in the discrete
case, witnessed by finite-composability, to the continuous case.
Finite-composability  is often known for standard divergences, or can be
established by direct calculations. If $\mathbf{\Delta}$ is approximable and
continuous, finite-composability implies composability. Formally, we have the
following theorem.

\begin{theorem}\label{thm:div:fin-comp->comp}
A continuous approximable $A$-graded family $\mathbf{\Delta}$ is composable if
finite-composable.
\end{theorem}

\subsection{Basic Properties of $f$-divergences}

To discuss basic properties of divergences for DP, RDP, zCDP, and tCDP,
we begin with basic properties of $f$-divergences since DP can be formulated by
a graded family $\mathbf{\Delta^{\mathtt{DP}}} = \{\Delta^{\mathtt{DP}(\varepsilon)}\}_{0 \leq \varepsilon}$ of $f$-divergences, and R\'enyi divergences are logarithms of $f$-divergences.
An $f$-divergence $\Delta^{f}$ of subprobability measures is defined in the same way as 
$f$-divergence of probability measures (\ref{eq:f-divergences}).
\iffull
The $f$-divergences are not necessarily positive for subprobability
measures, though they are positive for proper probability measures.
We can extend the continuity of $f$-divergences \citep[Theorem
16]{1705001_2006} to support subprobability measures.

\begin{theorem}[Cf. {\citet[Theorem 16]{1705001_2006}}]\label{f-divergence:continuity}
For any weight function $f$, the $f$-divergence $\Delta^f$ is
continuous:\footnote{%
  Note that a measurable finite partition $\{A_i\}_{i = 0}^n$ on $X$ is
  equivalent to a measurable function $k \colon X \to I$ where $I =
\{0,1,\ldots,n\}$.}
for any \emph{subprobability measures} $\mu_1, \mu_2 \in \mathcal{G}X$ on $X$,
we have
\[
\Delta^f_X(\mu_1,\mu_2)
= \sup \Set{\sum_{i = 0}^n \mu_2(A_i)f\left(\frac{\mu_1(A_i)}{\mu_2(A_i)}\right) \mid \{A_i\}_{i = 0}^n
\text{ is a measurable finite partition of } X}.
\]
\end{theorem}


As we have seen, DP can be formulated by the $\RR_{\geq 0}$-graded family
$\mathbf{\Delta^{\mathtt{DP}}} = \{\Delta^{\mathtt{DP}(\varepsilon)}\}_{0 \leq
\varepsilon}$ of $f$-divergences, while the R\'enyi divergences supporting RDP,
zCDP, and tCDP are logarithms of $f$-divergences.
Before proving basic properties of divergences for DP, RDP, zCDP, and tCDP,
we first need two important basic properties of $f$-divergences, continuity and approximability,
and we show that finite-composability of $f$-divergences are extended to (proper) composability.

\begin{theorem}\label{f-divergence:approximability}
The $f$-divergence $\Delta^f$ is approximable for any weight function $f$. 
\end{theorem}
Therefore, any finite-composable family of $f$-divergences is composable.
\begin{theorem}\label{f-divergence:composability}
An $A$-graded family $\mathbf{\Delta} = \{\Delta^{f_\alpha}
\}_{\alpha \in A}$ of the $f_\alpha$-divergences is composable if it is finite-composable.
\end{theorem}

We remark here that any composable family of $f$-divergences is also additive by
applying Theorem \ref{thm:div:cont+comp->mon}, since $f$-divergences are always
continuous (Theorem \ref{f-divergence:continuity}).
%
\subsection{Properties of Divergences for DP, RDP, zCDP, and tCDP}

As we have seen, DP can be formulated by the $\mathbb{R}_{\geq 0}$-graded family $\mathbf{\Delta^{\mathtt{DP}}}$ of $f$-divergences.
By Theorem \ref{thm:div:cont+comp->mon} and \ref{f-divergence:composability} and \citet[Theorem 1]{BartheOlmedo2013},
we obtain the basic properties of the divergences $\mathbf{\Delta^{\mathtt{DP}}}$ for DP as follows:
\begin{theorem}[Cf. {\citet[Theorem 1]{BartheOlmedo2013}}]
\label{properties:DP}
The $\mathbb{R}_{\geq 0}$-graded family $\mathbf{\Delta^{\mathtt{DP}}}=\{\Delta^{\mathtt{DP}(\varepsilon)}\}_{0 \leq \varepsilon}$
is reflexive, continuous, approximable, composable, and additive.
\end{theorem}
Similarly, we can obtain basic properties for RDP, zCDP, and tCDP.
First, by Theorem~\ref{f-divergence:continuity} and Theorem~\ref{f-divergence:approximability},
the exponential $\exp(D^\alpha)$ of R\'enyi divergence of order $\alpha$ is
continuous and approximable because is exactly the $f$-divergence 
with weight function $t \mapsto \exp(\alpha/(1-\alpha))t^\alpha$.

Since the logarithm function is monotone and continuous except at $0$, R\'enyi divergence is continuous and approximable too.
Reflexivity and finite-composability of R\'enyi divergences follow by direct calculations.
Theorem \ref{f-divergence:composability} yields:
\begin{theorem}\label{properties:Renyi}
  For any $\alpha > 1$, the R\'enyi divergence $D^\alpha$ of order $\alpha$ is reflexive,
  continuous, approximable, composable, and additive (as a singleton-graded
  family).  
\end{theorem}

\iffull
We extend the following properties of R\'enyi divergences which give the
transitive laws of RDP and zCDP to support subprobability measures. (An known
analogous law for tCDP is not known.)
%
%
\begin{proposition}[Cf. {\citet[Theorem 3]{6832827}}]\label{lem:Renyi:monotonicity}
We have
\[
1 < \alpha \leq \beta
\implies
D^\alpha_X(\mu_1||\mu_2) \leq D^\beta_X(\mu_1||\mu_2).
\]
\end{proposition}
%
%
\begin{proposition}[Cf. {\citet[Lemma 4.1]{Langlois2014}}]\label{lem:Renyi:weak-triangle}
For any $\alpha > 1$, $\mu_1,\mu_2, \mu_3 \in \mathcal{G}X$, and  $p,q > 1$ satisfying $\frac{1}{p}+\frac{1}{q} = 1$, we have
\[
D^\alpha_X(\mu_1||\mu_3) \leq \frac{p\alpha -1 }{p(\alpha - 1)}D^{p\alpha}_X(\mu_1||\mu_2) + D^{\frac{q}{p}(p\alpha - 1)}_X(\mu_1||\mu_2)
.\]
\end{proposition}
\fi

As we have seen in Section \ref{sec:cont-case-motiv}, we can define
divergences for zCDP and tCDP by Equation (\ref{eq:divergence_zCDP}) and
Equation (\ref{eq:divergence_tCDP}).
Explicitly, we introduce the divergences for zCDP and tCDP by 
$\Delta^{\mathtt{zCDP}(\xi,\rho)} = \sup_{1 < \alpha} \frac{1}{\alpha}( D^\alpha-\xi)$ and 
$\Delta^{\omega-\mathtt{tCDP}(\rho)} = \sup_{1 < \alpha < \omega} \frac{1}{\alpha} D^\alpha$ respectively.
Since two supremums are commutative ($\sup_x \sup_y A(x,y) = \sup_y \sup_x
A(x,y)$) in general, the following basic properties of the graded family of zCDP
and the divergence of tCDP are obtained from Theorem \ref{properties:Renyi}.

\begin{theorem}\label{properties:zCDP}
The $\mathbb{R}_{\geq 0}$-graded family $\mathbf{\Delta^{\mathtt{zCDP}}}=\{\Delta^{\mathtt{zCDP}(\xi)}\}_{0 \leq \xi}$ for zCDP 
is reflexive, continuous, composable, and additive.
\end{theorem}
\begin{theorem}\label{properties:tCDP}
For each $1 < \omega$, the divergence $\Delta^{\omega-\mathtt{tCDP}}$ for $\omega$-tCDP is reflexive, continuous, composable, and additive.
\end{theorem}
Note that we may not have approximability, but the family is still composable.
These results also hold for subprobability measures where R\'enyi divergence and
divergences for zCDP and tCDP are defined in a way similar to Equation
(\ref{eq:Renyi_divergences}) and Equation (\ref{eq:zCDP_definition})
respectively.

\section{Approximate Span-Lifting}
\label{sec:span-liftings}

We are now ready to combine graded divergences with spans, leading to our
new relational liftings.  Given an $A$-graded family $\mathbf{\Delta} =
\{\Delta^\alpha\}_{\alpha \in A}$ of divergences, we introduce a graded monad on
$\Span(\Meas)$ called the \emph{approximate span-lifting}
$(-)^{\sharp(\mathbf{\Delta},\alpha,\delta)}$ for the family $\mathbf{\Delta}$,
where $\alpha \in A$ and $\delta \in \ol{\RR}$.  We first
define its action on objects.

\begin{definition}\label{def:span}
We define the span-constructor $(-)^{\sharp(\mathbf{\Delta},\alpha,\delta)}$ as follows:
for any $(X,Y,\Phi,\rho_1,\rho_2)$ in $\Span(\Meas)$, we define the $\Span(\Meas)$-object
\begin{align*}
(X,Y,\Phi,\rho_1,\rho_2)^{\sharp(\mathbf{\Delta},\alpha,\delta)}
&= (\mathcal{G}X,\mathcal{G}Y,W(\Phi,\mathbf{\Delta},\alpha,\delta),~\mathcal{G}\rho_1 \circ \pi_1,~\mathcal{G}\rho_1 \circ \pi_2)\\
\text{ where }&
W(\Phi,\mathbf{\Delta},\alpha,\delta)
= \Set{(\nu_1,\nu_2) \in \mathcal{G}\Phi \times \mathcal{G}\Phi \mid \Delta^\alpha_{\Phi}(\nu_1,\nu_2) \leq \delta}.
\end{align*}
We view $W(\Phi,\mathbf{\Delta},\alpha,\delta)$ as a subspace of the measurable space
$\mathcal{G}\Phi \times \mathcal{G}\Phi$.
\end{definition}
Intuitively, $(X,Y,\Phi,\rho_1,\rho_2)^{\sharp(\mathbf{\Delta},\alpha,\delta)}$
relates subprobability measures with $\Delta^\alpha$-distance at most $\delta$.
The set $W(\Phi,\mathbf{\Delta},\alpha,\delta)$ contains all possible witness distributions,
and $\pi_1$ and $\pi_2$ are canonical projections from
$W(\Phi,\mathbf{\Delta},\alpha,\delta)$ to $\mathcal{G}\Phi$. As a special case,
the approximate span-lifting $(-)^{\sharp(\mathbf{\Delta},\alpha,\delta)}$
recovers the divergence $\Delta^\alpha$ by applying the equality relation
$(X,X,\mathtt{Eq}_X,\pi_1,\pi_1)^{\sharp(\mathbf{\Delta},\alpha,\delta)}$.
\begin{theorem}
\label{prop:span-lift:describe}
For any $A$-graded family $\mathbf{\Delta}$, $\alpha \in A$, and $\delta \in \ol{\RR}$, we have
\[
(X,X,X,\mathtt{id}_X,\mathtt{id}_X)^{\sharp(\mathbf{\Delta},\alpha,\delta)}
=  (\mathcal{G}X,\mathcal{G}X,\Set{(\mu_1,\mu_2)\mid\Delta^\alpha_X (\mu_1,\mu_2) \leq \delta},\pi_1,\pi_2).
\]
Here, $(X,X,X,\mathrm{id}_X,\mathrm{id}_X)$ is isomorphic to the equality
relation $(X,X,\mathrm{Eq}_X,\pi_1|_{\mathrm{Eq}_X},\pi_1|_{\mathrm{Eq}_X})$.
\end{theorem}

Next, we give approximate span-liftings the structure of a graded monad with
double strength.  We consider the important case where $\mathbf{\Delta}$ is a
reflexive, composable, and additive $A$-graded family of divergences; in some
cases, we can recover more limited versions of approximate span-liftings by
dropping or weakening these properties.
\begin{theorem}\label{thm:span-lifting:conclusion}
If an $A$-graded family $\mathbf{\Delta}$ is reflexive, composable, and additive, 
then the approximate span-lifting $(-)^{\sharp(\mathbf{\Delta},\alpha,\delta)}$
form an $A \times\ol\RR$-graded monad with double
strength. Namely, there are maps
\vspace*{0.2em}
\begin{description}
\item[Functor:]
For any morphism $(h,k,l) \colon {(X,Y,\Phi,\rho_1,\rho_2)} \to {(Z,W,\Psi,\rho'_1,\rho'_2)}$ in the category $\Span(\Meas)$ and any $(\alpha,\delta) \in A \times\ol\RR$,
\[
  (\mathcal{G}h,\mathcal{G}k,\mathcal{G}l \times \mathcal{G}l)
  \colon (X,Y,\Phi,\rho_1,\rho_2)^{\sharp(\mathbf{\Delta},\alpha,\delta)}
  \to (Z,W,\Psi,\rho'_1,\rho'_2)^{\sharp(\mathbf{\Delta},\alpha,\delta)}.
\]
\item[Unit:] For any morphism $(X,Y,\Phi,\rho_1,\rho_2)$ in $\Span(\Meas)$,
\[
{(\eta_X,\eta_Y,\langle \eta_\Phi,\eta_\Phi \rangle)} \colon
{(X,Y,\Phi,\rho_1,\rho_2)} \to {(X,Y,\Phi,\rho_1,\rho_2)}^{\sharp(\mathbf{\Delta},1_A,0)}.
\]

\item[Kleisli lifting:] For any morphism $(h,k,l)\colon(X,Y,\Phi,\rho_1,\rho_2) \to(Z,W,\Psi,\rho'_1,\rho'_2)^{\sharp(\mathbf{\Delta},\alpha,\delta)}$ in $\Span(\Meas)$ and $(\beta, \gamma) \in A \times\ol\RR$,
\[
(h^\sharp, k^\sharp, (\pi_1\circ l)^\sharp \times (\pi_2\circ l)^\sharp )\colon
(X,Y,\Phi,\rho_1,\rho_2)^{\sharp(\mathbf{\Delta},\beta,\gamma)} \to (Z,W,\Psi,\rho'_1,\rho'_2)^{\sharp(\mathbf{\Delta},\alpha\beta,\delta+\gamma)}
\]
\item[Inclusions:] For any $(X,Y,\Phi,\rho_1,\rho_2)$ in $\Span(\Meas)$, and any $\alpha \preceq \beta$ and $\delta \leq \gamma$,
\[
(\mathtt{id}_{\mathcal{G}X},\mathtt{id}_{\mathcal{G}Y}, \mathtt{id}_{\mathcal{G}\Phi} \times \mathtt{id}_{\mathcal{G}\Phi})
\colon
(X,Y,\Phi,\rho_1,\rho_2)^{\sharp(\mathbf{\Delta},\alpha,\delta)}\to (X,Y,\Phi,\rho_1,\rho_2)^{\sharp(\mathbf{\Delta},\beta,\gamma)}.
\]

\item[Double strength:]
For any $(X,Y,\Phi,\rho_1,\rho_2)$ and $(Z,W,\Psi,\rho'_1,\rho'_2)$ in $\Span(\Meas)$,
and parameters $(\alpha,\delta)$ and $(\beta,\gamma)$ in $A \times\ol\RR$, by letting
$\theta_i = \mathtt{dst}_{\Phi,\Psi}\circ(\pi_i \times \pi_i)$ where $i = 1,2$,
\[
\begin{aligned}
&(\mathtt{dst}_{X,Z},\mathtt{dst}_{Y,W},\langle \theta_1, \theta_2\rangle)\\
&\qquad \colon
(X,Y,\Phi,\rho_1,\rho_2)^{\sharp(\mathbf{\Delta},\alpha,\delta)} \mathbin{\dot\times} (Z,W,\Psi,\rho'_1,\rho'_2)^{\sharp(\mathbf{\Delta},\beta,\gamma)} \to (\Phi \mathbin{\dot\times} \Psi)^{\sharp(\mathbf{\Delta},\alpha\beta,\delta+\gamma)}
\end{aligned}.
\]
\end{description}
\end{theorem}
\begin{proof}[Proof Sketch]
Checking of the axioms of graded monad is straightforward since all structures are inherited from the sub-Giry monad $\mathcal{G}$.
It suffices to prove the well-definedness of the above maps.
For example, we check the well-definedness of the Kleisli lifting of a morphism 
$(h,k,l)\colon(X,Y,\Phi,\rho_1,\rho_2) \to(Z,W,\Psi,\rho'_1,\rho'_2)^{\sharp(\mathbf{\Delta},\alpha,\delta)}$ in $\Span(\Meas)$.
To prove this, we first show that the third component $(\pi_1\circ l)^\sharp \times (\pi_2\circ l)^\sharp$ of the Kleisli lifting forms a measurable function from $W(\Phi,\mathbf{\Delta},\beta,\gamma)$ to $W(\Psi,\mathbf{\Delta},\alpha\beta,\delta+\gamma)$ by using the composability of $\mathbf{\Delta}$ where measurability is obvious since $W(\Phi,\mathbf{\Delta},\beta,\gamma)$ and $W(\Psi,\mathbf{\Delta},\alpha\beta,\delta+\gamma)$ are the subspaces of $\mathcal{G}\Phi \times \mathcal{G}\Phi$ and $\mathcal{G}\Psi \times \mathcal{G}\Psi$.
Next, we show $\mathcal{G}\rho'_1\circ \pi_1 \circ ((\pi_1\circ l)^\sharp \times (\pi_2\circ l)^\sharp) = h^\sharp\circ \rho_1$ and $\mathcal{G}\rho'_2\circ \pi_2 \circ ((\pi_1\circ l)^\sharp \times (\pi_2\circ l)^\sharp) = k^\sharp\circ \rho_2$, but this is given from the assumption $\rho'_1 \circ l = h \circ \rho_1$ and $\rho'_2 \circ l = k \circ \rho_2$.

Similary, the well-definedness of functor part and unit are proved by using the composability and reflexivity of $\mathbf{\Delta}$; the inclusion is obtained from the definition of $A$-graded family of divergences; the double strength is obtained from the additivity of $\mathbf{\Delta}$.
\end{proof}

%

\PARAGRAPH{Remark: Adaptive Compositions}
Many composition theorems of differential privacy are based on the notion of $k$-fold adaptive composition~\citep[Definition 2.3]{Winograd-Cort:2017:FAD:3136534.3110254} and \citep[Section A]{5670947}.
Roughly speaking, for $k$ programs $q_1,\ldots,q_k$ their $k$-fold adaptive composition $q_1 \rhd q_2 \rhd \cdots \rhd q_k$ calculates in the following way:
\begin{enumerate}
\item  The first program $q_1$ takes an input $x$ in $X$, and returns an output $y_1$ in $Y_1$.
\item  The second program $q_2$ takes an input $x \in X$ and the output $y_1 \in Y_1$ of the previous program $q_1$, and returns an output $y_2 \in Y_2$.
\item[\ldots]
\item[($k$)] The $k$-th program $q_k$ takes an input $x \in X$ and the outputs $y_1,\dots,y_{k-1}$ of previous programs $q_1,\dots,q_{k-1}$, and returns an output $y_k \in Y_k$.
\end{enumerate}

We observe that our definition of composability of divergences covers the
\emph{standard} composability with respect to $k$-fold adaptive
composition.\footnote{%
  For differential privacy, there are \emph{advanced composition} theorems such as
  \citet[Theorem 3.3]{5670947}, \citet[Theorem 3.20]{DworkRothTCS-042}, which give
stronger privacy guarantees.}
For example, adaptive composition of two randomized programs can be formulated categorically as follows: let $f \colon X \to \mathcal{G}Y$ and $f \colon Y \times X \to \mathcal{G}X$ be two randomized programs.
The adaptive composition $f \rhd g \colon X \to \mathcal{G}(Y \times Z)$ is defined by
 \[
 f \rhd g
 =
 (\mathrm{st}_{Y,Z} \circ (\mathrm{id}_Y \times g) \circ \alpha_{Y,Y,Z})^\sharp \circ 
 \mathrm{st}'_{Y \times Y,X} \circ (\mathcal{G}\mathrm{copy}_Y \times \mathrm{id}_X) \circ (f \times \mathrm{id}_X) \circ \mathrm{copy}_X.
 \]
 Here, $\mathrm{st}'_{Y \times Y,X}$ is the costrength $\mathcal{G}(Y \times Y) \times X \to \mathcal{G}((Y \times Y)\times X)$;
 $\mathrm{copy}_X$ is the diagonal map $X \to X \times X$ ($x \mapsto (x,x)$) on $X$; 
 $\alpha_{Y,Y,Z}$ is the associativity $(Y \times Y) \times Z \to Y \times (Y \times Z)$ of cartesian product of $\Meas$.
 We show that the composability of $\mathbf{\Delta}$ is stronger than the adaptive composability.
 Suppose that $\mathbf{\Delta}$ reflexive, continuous and composable. 
 Since $(-)^{\sharp(\mathbf{\Delta},\alpha,\delta)}$ is a graded span-lifting
 with a double strength, the adaptive composition of the following two morphisms
 $(f_1,f_2,f_3) \colon \Phi \to \Psi^{\sharp(\mathbf{\Delta},\alpha,\delta)}$ and
 $(g_1,g_2,g_3) \colon \Psi \mathbin{\dot\times} \Phi \to \Omega^{\sharp(\mathbf{\Delta},\beta,\gamma)}$
 of spans is given by $(f_1 \rhd g_1,f_2 \rhd g_2,l) \colon \Phi \to (\Psi
 \mathbin{\dot\times}
 \Omega)^{\sharp(\mathbf{\Delta},\alpha\beta,\delta+\gamma)}$ (we omit details
 of $l$).
\fi
\PARAGRAPH{Approximate Span-liftings for DP, RDP, and zCDP}
Finally, we build approximate span-liftings for DP, RDP, zCDP, and tCDP by combining
Theorems \ref{properties:DP}, \ref{properties:Renyi}, \ref{properties:zCDP}, and
\ref{properties:tCDP} 
with the construction of categorical structures of approximate span-liftings
(Theorem \ref{thm:span-lifting:conclusion}).
\begin{theorem}[Approximate span-lifting for DP, RDP, zCDP, tCDP]
The following approximate span-liftings are graded liftings
with a double strength of $\mathcal{G} \times \mathcal{G}$ along $U \colon \Span(\Meas) \to \Meas\times\Meas$.
\begin{center}
\begin{tabular}{llll}
\toprule
Privacy &
(Graded family of )Divergence &
Approximate span-lifting &
Grading Monoid\\
\midrule
DP &
$\mathbf{\Delta}^\mathtt{DP} = \{\Delta^\mathtt{DP(\varepsilon)}\}_{0 \leq \varepsilon}$
&
$\{ (-)^{\sharp (\mathbf{\Delta}^\mathtt{DP},\varepsilon,\delta)} \}_{0 \leq \varepsilon,0 \leq \delta}$ &
$\RR_{\geq 0} \times \RR_{\geq 0}$
\\
RDP &
$D^\alpha$
\quad (R\'enyi divergence; see (\ref{eq:Renyi_divergences})) &
$\{ (-)^{\sharp (D^\alpha,\ast,\rho)} \}_{\ast \in \{\ast\},\rho \in \ol\RR}$ &
$\ol\RR$
\\
zCDP &
$\mathbf{\Delta}^\mathtt{zCDP} = \{\Delta^\mathtt{zCDP(\xi)}\}_{0 \leq \xi}$
\quad (see (\ref{eq:divergence_zCDP}))&
$\{ (-)^{\sharp (\mathbf{\Delta}^\mathtt{zCDP},\xi,\rho)} \}_{0 \leq \xi,\rho \in \ol\RR}$ &
$\RR_{\geq 0} \times \ol\RR$
\\
tCDP &
$\mathbf{\Delta}^{\omega-\mathtt{tCDP}} = \{\Delta^{\omega-\mathtt{tCDP}}\}$
\quad (see (\ref{eq:divergence_tCDP}))&
$\{ (-)^{\sharp (\Delta^{\omega-\mathtt{tCDP}},\ast,\rho)} \}_{\ast \in \{\ast\},\rho \in \ol\RR}$ &
$\ol\RR$
\\
\bottomrule
\end{tabular}
\end{center}
\end{theorem}
%
%

\section{Case Study: the Program Logic Span-apRHL}
\label{sec:span-apRHL}
The previous section showed that the RDP, zCDP, and tCDP relaxations of differential
privacy can be captured by relational liftings with the same categorical
properties enjoyed by relational liftings for standard differential privacy.  As
a result, we can use these liftings to give the semantic foundation for formal
verification of these relaxations.  To demonstrate a concrete application, we design
a program logic span-apRHL that can prove DP, RDP, zCDP, and tCDP for randomized
algorithms, supporting both discrete and continuous random samplings.

\PARAGRAPH{The Language pWHILE}
We take a standard, first-order language pWHILE, augmenting the usual imperative
commands with a random sampling statement (we omit the grammar of
expressions which is largely standard).
\begin{align*}
	\tau &
		::= \mathtt{bool} \mid \mathtt{int} \mid \mathtt{real} \mid \tau^d ~(d \in\NN) \mid \ldots
		&\text{(basic types)}\\
    \iffull
	  e &
	  	::= x \mid b \in \BB \mid n \in \ZZ \mid r \in \RR \mid e_1 \mathbin{\oplus} e_2 \mid e_1 \mathbin{\bowtie} e_2 \mid e_1[e_2] \mid \ldots
	  	&\text{(expressions)}\\
	  \lefteqn{\qquad{\oplus} ::= {+} \mid {-} \mid {\ast} \mid {/} \mid {\min} \mid {\max} \mid {\land} \mid {\vee}
	  \qquad {\bowtie}::= {\leq} \mid {\geq} \mid {=} \mid {\neq} \mid {<} \mid {>}}\\
    \fi
	\nu &
		::=
		\mathtt{Dirac}(e) \mid \mathtt{Bern}(e)
		\mid \mathtt{Lap}(e_1,e_2) \mid \mathtt{Gauss}(e_1,e_2) \mid \ldots
		&\text{(probabilistic expression)}\\
	c &
		::= \mathtt{skip} \mid x \xleftarrow{\$} \nu \mid c_1; c_2 \mid  \mathtt{if}~e~\mathtt{then}~c_1~\mathtt{else}~c_2 \mid \mathtt{while}~e~\mathtt{do}~c
		&\text{(commands)}
\end{align*}
\iffull
Here, $b$, $n$, and $r$ are constants; $\tau$ is a \emph{value type}; $x$ is a
\emph{variable}; $e$ is an \emph{expression}; $\nu$ is a \emph{probabilistic
expression}; $\mathtt{Dirac}$,
$\mathtt{Bern}$,
$\mathtt{Lap}$, and
$\mathtt{Gauss}$ represent the Dirac,
Bernoulli,
Laplace, and the Gaussian
distributions, respectively; $c$ is a \emph{command/program}.
We will use the following shorthands:
$
x \leftarrow e
 \labeleq{\text{def}} x \xleftarrow{\$} \mathtt{Dirac}(e)
$
and
$
\mathtt{if}~b~\mathtt{then}~c 
 \labeleq{\text{def}} \mathtt{if}~b~\mathtt{then}~c~\mathtt{else}~\mathtt{skip}
$.
We consider programs that are well typed. The type system is largely
standard, with three kinds of judgments: $\Gamma \vdash^t e \colon
\tau$, $\Gamma \vdash^p \nu : \tau$, and $\Gamma \vdash
c$ for expressions, distributions and programs, respectively.
For details, see Appendix.
\fi

\PARAGRAPH{Relational Assertions}
Our assertion logic uses formulas of the form
\[
  \Phi,\Psi ::= \mathcal{E} \mid\Phi \land \Psi \mid \Phi \vee \Psi \mid \neg \Phi
\]
where $\mathcal{E}$ represents basic relational expressions,
\iffull
namely:
\[
  \mathcal{E} ::= e_1\lrangle{1} \mathbin{\bowtie} e_2\lrangle{2}
	\mid (e_1\lrangle{1} \oplus_1  e_2\lrangle{2}) \bowtie (e_3\lrangle{1} \oplus_2 e_4\lrangle{2}) .
\]
\fi
As usual in relational logics, we use the tags $\lrangle{1}$ and $\lrangle{2}$
to distinguish expressions evaluated in the first and second memory,
respectively.
\iffull
For simplicity, we consider only the relations given in the above syntax, the
language can be easily extended with other constructions. In the
following we will use some syntactic sugar for constant $k$: $(e\lrangle{1}
\bowtie k) \labeleq{\text{def}} (e \bowtie k)\lrangle{1} =
\mathtt{true}\lrangle{2}$, and $  (e\lrangle{2} \bowtie k) \labeleq{\text{def}}
\mathtt{true}\lrangle{1} = (e \bowtie k)\lrangle{2}$.
%
%
We consider only  relation expression $\Phi$ that are well-formed in a context
$\Gamma$, and we denote this by the judgment $\Gamma \vdash^R \Phi$. Rules for
deriving this kind of judgments are standard, and postponed to Appendix.
\fi

Since we use span-liftings instead of relational liftings, we interpret
relational assertions as spans, that is, as $\Span(\Meas)$-objects. This can be
done by first interpreting assertions $\Gamma \vdash^R \Phi$ as binary relations
$\interpret{\Phi} \subseteq \interpret{\Gamma} \times \interpret{\Gamma}$, and
then converting to spans
$(\interpret{\Gamma},\interpret{\Gamma},\interpret{\Phi},\pi_1,\pi_2)$.
\iffull
We describe the semantics of relation assertions in the next section.

We will also use implications of relations $\Gamma \vdash^I \Phi \implies \Psi$,
which is defined when $\Gamma \vdash^R \Phi$ and $\Gamma \vdash^R \Psi$, and the
implication $\Phi \implies \Psi$ forms a tautology under the typing context
$\Gamma$. For example, we have the following inclusion, where $\Gamma \vdash^t
x \colon \mathtt{real}$:
 \[
 \Gamma \vdash^I ((x\lrangle{1} \leq x\lrangle{2}) \wedge (x\lrangle{1} \geq x\lrangle{2})) \implies (x\lrangle{1} = x\lrangle{2}).
 \]
\fi

\PARAGRAPH{Relational Program Logic Judgments, Axioms and Rules}
In span-apRHL we can prove three kinds of judgments corresponding to
differential privacy, RDP, zCDP, and tCDP.
For well-typed commands $\Gamma \vdash c_1$ and $\Gamma \vdash c_2$ and assertions $\Gamma \vdash^R \Phi$ and $\Gamma
\vdash^R \Psi$, we define judgments:
\begin{align*}
\Gamma \vdash c_1 \sim^\mathbf{\mathtt{DP}}_{\varepsilon, \delta} c_2 &\colon \Phi \implies \Psi
	\quad \text{$(\varepsilon, \delta)$-differential privacy (DP)}\\
\Gamma \vdash c_1 \sim^\mathbf{\alpha\mathtt{-RDP}}_{\rho} c_2 &\colon \Phi \implies \Psi
	\quad \text{$(\alpha, \rho)$-R\'enyi differential privacy (RDP)}\\
\Gamma \vdash c_1 \sim^\mathbf{\mathtt{zCDP}}_{\xi, \rho} c_2 &\colon \Phi \implies \Psi
	\quad \text{$(\xi, \rho)$-zero-concentrated differential privacy (zCDP)}\\
\Gamma \vdash c_1 \sim^\mathbf{\omega-\mathtt{tCDP}}_{\rho} c_2 &\colon \Phi \implies \Psi
	\quad \text{$(\omega, \rho)$-truncated-concentrated differential privacy (tCDP)}
\end{align*}
\iffull


We divide the proof rules of span-apRHL in four classes: basic rules
(Figure~\ref{fig:basic-rules}), rules for basic
mechanisms (Figure~\ref{fig:mechanisms}), rules for reasoning about 
transitivity (Figure~\ref{fig:transitivity}), and rules for
conversions (Figure~\ref{fig:conversions}). 
The basic rules can be used to reason about either differential privacy, RDP, zCDP, and tCDP. We describe the basic rules in a parametric way by considering $\{
  \sim^{\mathbf{\Delta}}_{\alpha,\delta} \}_{\alpha \in A, 0 \leq \delta}$ to
  stand for one of the families
$\{ \sim^\mathbf{\mathrm{DP}}_{\varepsilon,\delta} \}_{0 \leq \varepsilon, 0 \leq \delta}$,
$\{ \sim^\mathbf{\alpha-\mathrm{RDP}}_{\rho} \}_{\ast \in \{\ast\}, 0 \leq \rho}$,
\mbox{$\{ \sim^\mathbf{\mathrm{zCDP}}_{\xi,\rho} \}_{0 \leq \xi, 0 \leq \rho}$}, and 
\mbox{$\{ \sim^\mathbf{\omega-\mathrm{tCDP}}_{\rho} \}_{0 \leq \rho}$}.
We give a selection of the proof rules in Figure~\ref{fig:basic-rules}; the rest
of the rules are standard and we defer them to the appendix.
Here, we comment
briefly on the rules. The [assn] rule for assignment is mostly standard, the
only non-standard aspect is that depending on which notion of privacy we want to
use, we need to select the corresponding unit $1_A$. The rule [seq] is the
sequential composition of commands and takes the same form no matter which
family of divergence we consider. The rule [weak] is our version of the usual
consequence rule, where additionally we can weaken also the privacy parameters
for each of the privacy definitions.

In Figure~\ref{fig:mechanisms}, we show some rules for the basic mechanisms that
we support: Bernoulli, Laplace, and Gauss. We give several of them to show the
difference, in terms of the parameters, for the same mechanism, that we have in
the different logics. All of them are supported in the continuous case.
\mg{I find it strange that we give all these primitives if then
  we don't used them. I suggest that either we drop some of them, e.g.
Bernoulli, or we add an example using some of them. Can Bernoulli be
used to do the attribute mean example but using a randomized-response
like approach?}
\tetsuya{I postpone some of Bernoulli and Laplace axioms, and added the below.}
We show only DP rules for Bernoulli and Laplace mechanisms, and 
postpone other Bernoulli and Laplace mechanism rules to the Appendix.

In Figure~\ref{fig:transitivity}, we show rules for transitivity in span-apRHL.
Transitivity is important because it allows one to reason about group
privacy~\citep{DworkRothTCS-042}. The different flavors of the logic have
different numeric parameters for these rules, reflecting the slight differences
in group privacy~\citep{DworkRothTCS-042,BunS16,Mironov17}.
\mg{We should talk about group privacy much earlier in the paper.}
Finally, Figure~\ref{fig:conversions} gives rules for converting between
judgments for different flavors of differential privacy. In some of them we have
a loss in the parameters, in others there is no loss. These rules correspond to
the different conversion theorems for the different
logics~\citep{BunS16,Mironov17}. Notice that most of these rules require
lossless programs because they have been formulated in terms of distributions,
rather than subdistributions.

\begin{figure*}
\begin{mathpar}
\Gamma \vdash x_1 \leftarrow e_1 \sim^{\mathbf{\Delta}}_{1_A,0} x_2 \leftarrow e_2 \colon \Phi\{e_1\lrangle{1}, e_2\lrangle{2} /x_1\lrangle{1},x_2\lrangle{2}\} \implies \Phi \quad \text{[assn]}

	\AxiomC{$
		\Gamma \vdash c_1 \sim^{\mathbf{\Delta}}_{\alpha,\delta} c_1' \colon \Phi \implies \Phi'
		\quad
		\Gamma \vdash c_2 \sim^{\mathbf{\Delta}}_{\beta,\gamma} c_2' \colon \Phi' \implies \Psi
		$}
	\RightLabel{{[seq]}}
	\UnaryInfC{$\Gamma \vdash c_1;c_2 \sim^{\mathbf{\Delta}}_{\alpha\beta,\delta+\gamma} c_1';c_2' \colon \Phi \implies \Psi$}
	\DisplayProof

	\AxiomC{$
			~~{\Gamma \vdash^I \Phi' \implies \Phi}~~{\Gamma \vdash^I \Psi \implies \Psi'}~~
			{\Gamma \vdash c_1 \sim^{\mathbf{\Delta}}_{\alpha, \delta} c_2 \colon \Phi \implies \Psi}
			~~{\alpha \leq \beta}~~{\delta \leq \gamma}
	$}
	\RightLabel{{[weak]}}
	\UnaryInfC{$
		\Gamma \vdash c_1 \sim^{\mathbf{\Delta}}_{\beta, \gamma} c_2 \colon \Phi' \implies \Psi'
	$}
	\DisplayProof
\end{mathpar}
\caption{Selection of span-apRHL basic rules.}
\label{fig:basic-rules}
\end{figure*}
\begin{figure*}
\begin{align*}
\lefteqn{\Gamma \vdash x_1 \xleftarrow{\$} \mathtt{Bern}(e_1) \sim^{\mathtt{DP}}_{\log\max(p,1-p)-\log\min(p,1-p),0} x_2 \xleftarrow{\$} \mathtt{Bern}(e_2) \colon}\\
&\qquad\qquad\qquad\qquad\qquad((e_1\lrangle{1} = p)\wedge(1 - e_1\lrangle{1} = e_2\lrangle{2}) \implies (x_1\lrangle{1} = x_2\lrangle{2}) &\text{[DP-Bern]}\\
\lefteqn{\Gamma \vdash x_1 \xleftarrow{\$} \mathtt{Bern}(e_1) \sim^{\mathtt{DP}}_{0,0} x_2 \xleftarrow{\$} \mathtt{Bern}(e_2) \colon}\\
&\qquad\qquad\qquad\qquad\qquad(e_1\lrangle{1} = e_2\lrangle{2}) \implies (x_1\lrangle{1} = x_2\lrangle{2}) &\text{[DP-Bern-Eq]}\\
\lefteqn{\Gamma \vdash x_1 \xleftarrow{\$} \mathtt{Lap}(e_1,\lambda) \sim^{\mathtt{DP}}_{r/\lambda,0} x_2 \xleftarrow{\$} \mathtt{Lap}(e_2,\lambda) \colon}\\
& \qquad\qquad\qquad\qquad\qquad (|e_1\lrangle{1} - e_2\lrangle{2}| \leq r) \implies (x_1\lrangle{1} = x_2\lrangle{2})& \text{[DP-Lap]}\\
\lefteqn{\Gamma \vdash x_1 \xleftarrow{\$} \mathtt{Gauss}(e_1,\sigma^2) \sim^{\alpha-\mathtt{RDP}}_{\alpha r^2/2\sigma^2} x_2 \xleftarrow{\$} \mathtt{Gauss}(e_2,\sigma^2)}\\
&\qquad\qquad\qquad\qquad\qquad (|e_1\lrangle{1} - e_2\lrangle{2}| \leq r) \implies (x_1\lrangle{1} = x_2\lrangle{2}) & \text{[RDP-G]}\\
\lefteqn{\Gamma \vdash x_1 \xleftarrow{\$} \mathtt{Gauss}(e_1,\sigma^2) \sim^{\mathtt{zCDP}}_{0, r^2/2\sigma^2} x_2 \xleftarrow{\$} \mathtt{Gauss}(e_2,\sigma^2)}\\
&\qquad\qquad\qquad\qquad\qquad (|e_1\lrangle{1} - e_2\lrangle{2}| \leq r) \implies (x_1\lrangle{1} = x_2\lrangle{2}) & \text{[zCDP-G]}\\
\lefteqn{\Gamma \vdash x_1 \xleftarrow{\$} \mathtt{Gauss}(e_1,\sigma^2) \sim^{\mathtt{tCDP}}_{0, r^2/2\sigma^2} x_2 \xleftarrow{\$} \mathtt{Gauss}(e_2,\sigma^2)}\\
&\qquad\qquad\qquad\qquad\qquad (|e_1\lrangle{1} - e_2\lrangle{2}| \leq r) \implies (x_1\lrangle{1} = x_2\lrangle{2}) & \text{[tCDP-G]}\\
&\AxiomC{
	$
		\begin{array}{l@{}}
		\exists{c >\frac{1+\sqrt{3}}{2}}.~
(2\log(0.66/\delta) \leq c^2)\wedge (\frac{c r}{\varepsilon} \leq \sigma)
		\end{array}
	$}
	\UnaryInfC{$
	\begin{array}{l@{}}
		\Gamma \vdash x_1 \xleftarrow{\$} \mathtt{Gauss}(e_1,\sigma^2) \sim^{\mathtt{DP}}_{\varepsilon,~\delta} x_2 \xleftarrow{\$} \mathtt{Gauss}(e_2,\sigma^2) \colon\\
		\qquad\qquad\qquad\qquad\qquad\qquad
		 (|e_1\lrangle{1} - e_2\lrangle{2}| \leq r) \implies (x_1\lrangle{1} = x_2\lrangle{2})
	\end{array}
	$}
	\DisplayProof
 & \text{[DP-G]}\\
	&\AxiomC{$
	1 < 1/\sqrt{\rho} \leq A/\delta
	$}
	\UnaryInfC{$
	\begin{array}{l@{}}
		\Gamma \vdash x_1 \xleftarrow{\$} e_1 + A \cdot\mathtt{arsinh}\left(\frac{1}{A}\mathtt{Gauss}(0,\delta^2/2\rho)\right)\\ \qquad\qquad\qquad\qquad\sim^{A/8\delta-\mathtt{tCDP}}_{16\rho} x_2 \xleftarrow{\$} e_2 + A\cdot \mathtt{arsinh}\left(\frac{1}{A}\mathtt{Gauss}(0,\delta^2/2\rho)\right) \colon\\
		\qquad\qquad\qquad\qquad\qquad\qquad
		 (|e_1\lrangle{1} - e_2\lrangle{2}| \leq \delta) \implies (x_1\lrangle{1} = x_2\lrangle{2})
	\end{array}
	$}
	\DisplayProof
 & \text{[tCDP-SinhG]}
\end{align*}
\caption{Rules for basic mechanisms for DP, RDP, zCDP, and tCDP in span-apRHL.}
\label{fig:mechanisms}
\end{figure*}

\begin{figure*}
\begin{align*}
&\AxiomC{$
	\begin{array}{l@{}}
	\Gamma \vdash c_1 \sim^{\mathtt{DP}}_{\varepsilon_1,\delta_1} c_2 \colon \Phi \implies x_1\lrangle{1}=x_2\lrangle{2}
	\quad
	\Gamma \vdash c_2 \sim^{\mathtt{DP}}_{\varepsilon_2,\delta_2} c_3 \colon \Psi \implies x_2\lrangle{1}=x_3\lrangle{2}\\
	\end{array}
	$}
\UnaryInfC{$\Gamma \vdash c_1 \sim^{\mathtt{DP}}_{\varepsilon_1+\varepsilon_2,~ \max(e^{\varepsilon_2}\delta_1+\delta_2,e^{\varepsilon_1}\delta_2+\delta_1)} c_3 \colon \Phi \circ \Psi \implies x_1\lrangle{1}=x_3\lrangle{2}$}
\DisplayProof
&\text{[DP-Trans]}\\
&\AxiomC{$
	\begin{array}{l@{}}
	\Gamma \vdash c_1 \sim^{p\alpha-\mathtt{RDP}}_{\rho_1} c_2 \colon \Phi \implies x_1\lrangle{1}=x_2\lrangle{2}
	\\
	\Gamma \vdash c_2 \sim^{q(p\alpha-1)/p-\mathtt{RDP}}_{\rho_2} c_3 \colon \Psi \implies x_2\lrangle{1}=x_3\lrangle{2}
	\quad
	{\frac{1}{p} + \frac{1}{q} = 1} \quad {1 < p} \quad {1 < q}
	\end{array}
	$}
\UnaryInfC{$\Gamma \vdash c_1 \sim^{\alpha-\mathtt{RDP}}_{((p\alpha - 1)\rho_1/p(\alpha-1)) + \rho_2} c_3 \colon \Phi \circ \Psi \implies x_1\lrangle{1}=x_3\lrangle{2}$}
\DisplayProof
&\text{[RDP-Trans]}\\
&\AxiomC{$
	\begin{array}{l@{}}
	\Gamma  \vdash c_1 \sim^{\mathtt{zCDP}}_{\xi(k-1)\sum_{i = 1}^{k-1},(k^2 - 1)\rho} c_2 \colon \Phi \implies x_1\lrangle{1}=x_2\lrangle{2}
	\\
	\Gamma  \vdash c_2 \sim^{\mathtt{zCDP}}_{\xi,\rho} c_3 \colon \Psi \implies x_2\lrangle{1}=x_3\lrangle{2}
	\quad {k \in \mathbb{N}} \quad {1 < k} 
	\end{array}
	$}
\UnaryInfC{$\Gamma  \vdash c_1 \sim^{\mathtt{zCDP}}_{\xi k \sum_{i = 1}^{k},k^2 \rho} c_3 \colon \Phi \circ \Psi \implies x_1\lrangle{1}=x_3\lrangle{2}$}
\DisplayProof
&\text{[zCDP-Trans]}
\end{align*}
\caption{Span-apRHL transitivity rules for group privacy}
\label{fig:transitivity}
\end{figure*}

\begin{figure*}
\begin{mathpar}
\AxiomC{
	$\Gamma \vdash c_1 \sim^{\mathtt{DP}}_{\varepsilon,~0} c_2 \colon \Phi \implies \Psi$
	\quad $c_1, c_2$: lossless
	}
\doubleLine
\RightLabel{{[D/z]}}
\UnaryInfC{$\Gamma \vdash c_1 \sim^{\mathtt{zCDP}}_{\varepsilon,~0} c_2 \colon \Phi \implies \Psi$
	\quad $c_1, c_2$: lossless}
\DisplayProof

\AxiomC{
	$\Gamma \vdash c_1 \sim^{\mathtt{zCDP}}_{0,~\rho} c_2 \colon \Phi \implies \Psi$
	}
\RightLabel{{[z/R]}}
\UnaryInfC{$\forall \alpha > 1.~\Gamma \vdash c_1 \sim^{\alpha-\mathtt{RDP}}_{\rho} c_2 \colon \Phi \implies \Psi$}
\DisplayProof

\AxiomC{
	$\Gamma \vdash c_1 \sim^{\mathtt{zCDP}}_{\xi,~\rho} c_2 \colon \Phi \implies \Psi$
	\quad $c_1, c_2$: lossless
	\quad $0 < \delta < 1$
	}
\RightLabel{{[z/D]}}
\UnaryInfC{$\Gamma \vdash c_1 \sim^{\mathtt{DP}}_{\xi+\rho+2\sqrt{\rho\log(1/\delta)},~\delta} c_2 \colon \Phi \implies \Psi$}
\DisplayProof

\AxiomC{
	$\Gamma \vdash c_1 \sim^{\omega-\mathtt{tCDP}}_{\rho} c_2 \colon \Phi \implies \Psi$,
	 $c_1, c_2$: lossless,
	 $\beta = \min(\omega,1+\sqrt{\log(1/\delta)/\rho})$,
	 $0 < \delta < 1$
	}
\RightLabel{{[t/D]}}
\UnaryInfC{$\Gamma \vdash c_1 \sim^{\mathtt{DP}}_{\rho\beta+\log(1/\delta)/(\beta-1),~\delta} c_2 \colon \Phi \implies \Psi$}
\DisplayProof

\AxiomC{
	$\Gamma \vdash c_1 \sim^{\alpha-\mathtt{RDP}}_{\rho} c_2 \colon \Phi \implies \Psi$
	\quad $c_1, c_2$: lossless
	\quad $0 < \delta < 1$
	}
\RightLabel{{[R/D]}}
\UnaryInfC{$\Gamma \vdash c_1 \sim^{\mathtt{DP}}_{\rho-{\log\delta}/(\alpha-1),~\delta} c_2 \colon \Phi \implies \Psi$}
\DisplayProof
\end{mathpar}
\caption{Rules for conversions between DP, RDP and zCDP in span-apRHL.}
\label{fig:conversions}
\end{figure*}
\fi

\iffull
\subsection{Denotational Semantics of pWHILE}
To prove the soundness of span-apRHL we interpret pWHILE in $\Meas$
using the sub-Giry monad $\mathcal{G}$. Most of the definitions are
standard. The value types are interpreted as expected.
To give a semantics to expressions, distribution expressions, and commands, we
interpret their associated typing/well-formedness judgments in some context
$\Gamma$, which is interpreted as usual as a product. We interpret an expression judgment $\Gamma \vdash^t e \colon \tau$
as a measurable function $\interpret{\Gamma \vdash^t e \colon \tau} \colon
\interpret{\Gamma} \to \interpret{\tau}$; for instance, the variable case
$\Gamma \vdash^t x\colon \tau$ is interpreted as the projection $\pi_x \colon
\interpret{\Gamma} \to \interpret{\tau}$.
Note that all operators ${\oplus}$ and comparisons ${\bowtie}$ are interpreted
to measurable functions ${\oplus} \colon \interpret{\tau} \times
\interpret{\tau} \to \interpret{\tau}$ and ${\bowtie} \colon \interpret{\tau}
\times \interpret{\tau} \to \interpret{\mathtt{bool}}$ respectively.  
%
Likewise, we interpret a distribution expression judgment $\Gamma \vdash^p \nu
\colon \tau$ as a measurable function $\interpret{\Gamma \vdash^p \nu \colon
\tau} \colon \interpret{\Gamma} \to \mathcal{G}\interpret{\tau}$; for
instance, the Gaussian expression 
$\Gamma \vdash^p \mathtt{Gauss}(e_1,e_2) \colon \mathtt{real}$
is interpreted   as a Gaussian distribution.
$\mathcal{N}(\interpret{\Gamma \vdash^t e_1 \colon \mathtt{real}},\interpret{\Gamma \vdash^t e_2 \colon \mathtt{real}})$.
%
Finally, we interpret a command judgment $\Gamma \vdash c$ as a measurable
function $\interpret{\Gamma \vdash c } \colon \interpret{\Gamma} \to
\mathcal{G}\interpret{\Gamma}$ defined inductively as
\begin{mathpar}
\interpret{\Gamma\vdash x \xleftarrow{\$} \nu}
	= \mathcal{G}(\mathrm{rw}\lrangle{\Gamma \mid x\colon\tau})
        \circ\mathrm{st}_{\interpret{\Gamma},\interpret{\tau}}
        \circ\lrangle{\mathrm{id}_{\interpret{\Gamma}},\interpret{\nu}},

\interpret{\Gamma\vdash c_1 ; c_2}
	= {\interpret{\Gamma\vdash c_2}}^\sharp \circ
        \interpret{\Gamma\vdash c_1},

\interpret{\Gamma\vdash\mathtt{skip}}
	 = \eta_{\interpret{\Gamma}}

\interpret{\Gamma \vdash \mathtt{if}~b~\mathtt{then}~c_1~\mathtt{else}~c_2}
	=\left[\interpret{\Gamma\vdash c_1},\interpret{\Gamma\vdash c_2} \right]
	\circ \mathrm{br}\lrangle{\Gamma}
	\circ\lrangle{\interpret{\Gamma\vdash b} ,\mathrm{id}_{\interpret{\Gamma}}}
\end{mathpar}
Here, $\mathrm{rw}\lrangle{\Gamma \mid x\colon\tau} \colon \interpret{\Gamma}
\times \interpret{x\colon\tau} \to \interpret{\Gamma}$ (${x\colon\tau} \in
\Gamma $) is an overwriting operation of memory $((a_1,\ldots,a_k,\ldots,
a_n),b_k) \mapsto (a_1,\ldots,b_k,\ldots, a_n)$, which is given from the
Cartesian products in $\Meas$.  The function $\mathrm{br}\lrangle{\Gamma} \colon
2 \times \interpret{\Gamma} \to \interpret{\Gamma} + \interpret{\Gamma}$ comes
from the canonical isomorphism $2 \times \interpret{\Gamma} \cong
\interpret{\Gamma} + \interpret{\Gamma}$ given from the distributivity of
$\Meas$.

To interpret loops, we introduce the dummy ``abort'' command $\Gamma \vdash
\mathtt{null}$ that is interpreted by the null/zero measure $\interpret{\Gamma
\vdash \mathtt{null}} = 0$, and the following commands corresponding to the
finite unrollings of the loop:
\[
[\mathtt{while}~b~\mathtt{do}~c]_n
=
\begin{cases} 
\mathtt{if}~b~\mathtt{then}~\mathtt{null}~\mathtt{else}~\mathtt{skip}, & \text{ if } n=0
\\
\mathtt{if}~b~\mathtt{then}~c;[\mathtt{while}~b~\mathtt{do}~c]_{k}, & \text{ if } n=k+1
\end{cases}
\]
We then interpret loops as:
$
\interpret{\Gamma \vdash \mathtt{while}~b~\mathtt{do}~c}
	= \sup_{ n\in\mathbb{N}} \interpret{\Gamma\vdash [\mathtt{while}~e~\mathtt{do}~c]_n}.
$
This is well-defined, since the family $\{ \interpret{\Gamma\vdash
[\mathtt{while}~e~\mathtt{do}~c]_n} \}_{n \in \mathbb{N}}$ is an $\omega$-chain
with respect to the $\omega\mathbf{CPO}_\bot$-enrichment $\sqsubseteq$ of
$\Meas_{\mathcal{G}}$.

\subsection{Semantics of Relations}
\label{sec:semantics_relations}
Since we use span-liftings instead of relational liftings, we need to interpret relation
expressions to spans, that is, $\Span(\Meas)$-objects.  We proceed in two steps:
first interpreting expressions as binary relations, and then converting
relations to spans. In the first step, we interpret a relation expression
$\Gamma \vdash^R \Phi$ as a binary relation over $\interpret{\Gamma}$:

\begin{align*}
\lefteqn{\Rinterpret{\Gamma \vdash^R e_1\lrangle{1} \mathbin{\bowtie} e_2\lrangle{2}}}\\
&\qquad \qquad = \Set{(m_1,m_2) \in \interpret{\Gamma} \times \interpret{\Gamma}
\mid
\interpret{\Gamma \vdash^t e_1 \colon \tau}(m_1) \bowtie \interpret{\Gamma \vdash^t e_2 \colon \tau}(m_2)
 }
\\
\lefteqn{
\Rinterpret{\Gamma \vdash^R  (e_1\lrangle{1} \otimes_1  e_2\lrangle{2}) \bowtie (e_3\lrangle{1} \otimes_2 e_4\lrangle{2})}
}\\
&\qquad \qquad = \Set{(m_1,m_2) \in \interpret{\Gamma} \times \interpret{\Gamma}
\mid
\begin{array}{l@{}}
\interpret{\Gamma \vdash^t e_1 \colon \tau}(m_1) \otimes_1 \interpret{\Gamma \vdash^t e_2 \colon \tau}(m_2)\\
\quad \bowtie \interpret{\Gamma \vdash^t e_3 \colon \tau}(m_1) \otimes_2 \interpret{\Gamma \vdash^t e_4 \colon \tau}(m_2)
\end{array}
}
\end{align*}
We interpret the connectives in the expected way:
\begin{mathpar}
  \Rinterpret{\Gamma \vdash^R \Phi \wedge \Psi} = \Rinterpret{\Gamma \vdash^R \Phi} \cap \Rinterpret{\Gamma \vdash^R \Psi}

  \Rinterpret{\Gamma \vdash^R \Phi \vee \Psi} = \Rinterpret{\Gamma \vdash^R \Phi} \cup \Rinterpret{\Gamma \vdash^R \Psi}

  \Rinterpret{\Gamma \vdash^R \neg \Phi} = (\interpret{\Gamma} \times \interpret{\Gamma}) \setminus \Rinterpret{\Gamma \vdash^R \Phi}
\end{mathpar}
Then, we can convert the binary relation $\Rinterpret{\Gamma \vdash^R \Phi}
\subseteq \interpret{\Gamma} \times \interpret{\Gamma}$ to the span
\[
  \interpret{\Gamma \vdash^R \Phi}  = (\interpret{\Gamma},\interpret{\Gamma},\Rinterpret{\Gamma \vdash^R \Phi},\pi_1|_\Rinterpret{\Gamma \vdash^R \Phi}, \pi_2|_\Rinterpret{\Gamma \vdash^R \Phi}) .
\]
We interpret the implication $\Gamma \vdash^I \Phi \implies \Psi$ by the
following morphism in $\Span(\Meas)$:
\[
\interpret{\Gamma \vdash^I \Phi \implies \Psi}
= (\mathrm{id}_\interpret{\Gamma},\mathrm{id}_\interpret{\Gamma}, (\mathrm{id}_\interpret{\Gamma} \times \mathrm{id}_\interpret{\Gamma})|_\Rinterpret{\Gamma \vdash^R \Phi}) \colon \interpret{\Gamma \vdash^R \Phi} \to \interpret{\Gamma \vdash^R \Psi}.
\]
\fi

\PARAGRAPH{Validity of Judgments}
We say a judgment $\Gamma \vdash c_1 \sim^\mathbf{\Delta}_{\alpha, \delta} c_2 \colon \Phi
\implies \Psi$ is valid if there exists a measurable function $l \colon
\Rinterpret{\Gamma \vdash^R \Phi} \to W(\interpret{\Gamma \vdash^R
\Psi},\mathbf{\Delta},\alpha,\delta)$ (we call it a \emph{witness function})
such that 
\[
  (\interpret{\Gamma \vdash c_1},\interpret{\Gamma \vdash c_2},l)
  \colon \interpret{\Gamma \vdash^R \Phi} \to \interpret{\Gamma \vdash^R \Psi}^{\sharp(\mathbf{\Delta},\alpha,\delta)}
\]
is a morphism in $\Span(\Meas)$. Concretely, we define the validity in span-apRHL as follows:
\begin{align*}
\Gamma \models c_1 \sim^\mathbf{\mathtt{DP}}_{\varepsilon, \delta} c_2 \colon \Phi \implies \Psi
& \mathbin{\text{ iff }}
\exists l.~(\interpret{\Gamma \vdash c_1}, \interpret{\Gamma \vdash c_2}, l)
\colon \interpret{\Gamma \vdash^R \Phi} \to \interpret{\Gamma \vdash^R \Psi}^{\sharp(\mathbf{\Delta}^\mathtt{DP},\varepsilon,\delta)},\\
\Gamma \models c_1 \sim^\mathbf{\alpha-\mathtt{RDP}}_{\rho} c_2 \colon \Phi \implies \Psi
& \mathbin{\text{ iff }}
\exists l.~(\interpret{\Gamma \vdash c_1}, \interpret{\Gamma \vdash c_2}, l)
\colon \interpret{\Gamma \vdash^R \Phi} \to \interpret{\Gamma \vdash^R \Psi}^{\sharp(D^\alpha,\ast,\rho)},\\
\Gamma \models c_1 \sim^\mathbf{\mathtt{zCDP}}_{\xi,\rho} c_2 \colon \Phi \implies \Psi
& \mathbin{\text{ iff }}
\exists l.~(\interpret{\Gamma \vdash c_1}, \interpret{\Gamma \vdash c_2}, l)
\colon \interpret{\Gamma \vdash^R \Phi} \to \interpret{\Gamma \vdash^R \Psi}^{\sharp(\mathbf{\Delta}^\mathtt{zCDP},\xi,\rho)}\\
\Gamma \models c_1 \sim^\mathbf{\omega{-}\mathtt{tCDP}}_{\rho} c_2 \colon \Phi \implies \Psi
& \mathbin{\text{ iff }}
\exists l.~(\interpret{\Gamma \vdash c_1}, \interpret{\Gamma \vdash c_2}, l)
\colon \interpret{\Gamma \vdash^R \Phi} \to \interpret{\Gamma \vdash^R \Psi}^{\sharp(\mathbf{\Delta}^{\omega{-}\mathtt{tCDP}},\ast,\rho)}.
\end{align*}

\PARAGRAPH{Soundness}

%
\begin{theorem}
If $\Gamma \vdash c_1 \sim^\mathbf{\Delta}_{\alpha, \delta} c_2 \colon \Phi \implies \Psi$ is derivable in span-apRHL, then it is valid.
\end{theorem}
\begin{proof}[Proof sketch]
The soundness of the basic rules is derived from the unit, graded
Kleisli liftings, and inclusions of the graded span-lifting $\{
(-)^{\sharp (\mathbf{\Delta},\alpha,\delta)} \}_{\alpha,\delta}$ given
in Section \ref{sec:span-liftings}.
We focus here on the soundness of the [seq] rule.
Since the judgments $\Gamma \vdash c_1\sim^{\mathbf{\Delta}}_{\alpha,\delta} c'_1 \colon \Phi \implies \Phi'$ and
$\Gamma \vdash c_2 \sim^{\mathbf{\Delta}}_{\alpha\beta,\delta+\gamma} c'_2 \colon \Phi' \implies \Psi$ are valid,
for some witness functions $l_1$ and $l_2$ we have
\[
(\interpret{\Gamma \vdash c_1}, \interpret{\Gamma \vdash c_1'}, l_1)\colon \interpret{\Phi}\!\to\! \interpret{\Phi'}^{\sharp(\mathbf{\Delta},\alpha,\delta)},\quad\hspace{-0.8em}
(\interpret{\Gamma \vdash c_2}, \interpret{\Gamma \vdash c_2'}, l_2)\colon \interpret{\Phi'}\!\to\! \interpret{\Psi}^{\sharp(\mathbf{\Delta},\beta,\gamma)} 
.
\]
By taking the graded Kleisli extension of the second morphism
$(\interpret{\Gamma \vdash c_2}, \interpret{\Gamma \vdash c_2'}, l_2)$, for some
witness function $l_3$ given by the construction in Theorem
\ref{thm:span-lifting:conclusion} (Kleisli lifting), we have the following
morphism in the category $\Span(\Meas)$:
\[
(\interpret{\Gamma \vdash c_2}^\sharp, \interpret{\Gamma \vdash c_2'}^\sharp, l_3) \colon \interpret{\Phi'}^{\sharp(\mathbf{\Delta},\alpha,\delta)} \to \interpret{\Psi}^{\sharp(\mathbf{\Delta},\alpha\beta,\delta+\gamma)}.
\]
Composing them, we conclude the validity of $\Gamma \vdash c_1;c_2 \sim^{\mathbf{\Delta}}_{\alpha\beta,\delta+\gamma} c'_1;c'_2 \colon \Phi \implies \Psi$.

The soundness of the mechanism rules are proved by interpreting known results of
mechanisms for DP, RDP, zCDP, and tCDP to span-liftings.  For example, the
soundness of [RDP-G] proved by interpreting the R\'enyi differential privacy
of Gaussian mechanism to span-liftings.  First, the function $f =
{\mathcal{N}({-},\sigma^2)} \colon \mathbb{R} \to \mathcal{G}\mathbb{R}$
describing a Gaussian mechanism is measurable.  From the previous result
\citet[Proposition 7]{Mironov17} of R\'enyi differential privacy of the Gaussian
mechanism, the measurable function $f$ satisfies the following implication:
\[
|x - y| \leq r
\implies
D^\alpha(f(x)||f(y)) \leq \alpha r^2/2\sigma^2.
\]
This implies that we have the below morphism in the category $\Span(\Meas)$:
\[
(f,f,(f \times f)|_\Phi) \colon \Set{(x,y) \in \mathbb{R} \times \mathbb{R}|{|x - y|} \leq r} \to \mathrm{Eq}_{\mathbb{R}}^{\sharp(D^\alpha,\ast,\alpha r^2/2\sigma^2)}.
\]
From this, by straightforward calculations, we obtain the soundness of [RDP-G].

Note that we need to give \emph{measurable} functions $l$ selecting witness
distributions when proving these rules---in the discrete case, these functions
can be obtained by the axiom of choice.
In the case of [RDP-G], we could give the witness $l = f \times f$ directly.

Similary, the soundness of the rest of mechanism rules follows from the following previous results on DP, RDP, zCDP and tCDP:
\citet[Propositions 6]{Mironov17},
\citet[Proposition 1]{DworkMcSherryNissimSmith2006},
\citet[Lemma 4.2]{Sato2016MFPS}
(an enhancement of \citet[Theorem 3.22]{DworkRothTCS-042}), and
\citet[Theorem 19]{BDRS18},
and the soundness of transitive rules follows from:
\citet[Lemma 4.2(iii)]{olmedo2014approximate},
\citet[Proposition 27]{BunS16} and \citet[Lemma 4.1]{Langlois2014}.
The soundness of the conversion rules follows by applying the
comparison theorems of divergences
\citet[Proposition 4]{BunS16},
\citet[Proposition 3]{Mironov17},
\citet[Lemmas 3.2, 3.5]{BunS16},
\citet[Lemma 8]{BDRS18}
to the following inclusion between the approximate span-liftings:
\[
(\Delta^1_\alpha \leq \delta \implies \Delta^2_\beta \leq \gamma)
\implies
((\mathrm{id},\mathrm{id},\mathrm{id}) \colon (\Phi)^{\sharp(\Delta^1,\alpha,\delta)} \to (\Phi)^{\sharp(\Delta^2,\beta,\gamma)} \text{ in } \Span(\Meas)).
\]
\end{proof}

\section{Verification Examples} \label{sec:examples}
\iffull
We show how we can use the span-pRHL program logic to verify concrete
programs. We stress an important point here, since the guarantees
provided by RDP, zCDP, and tCDP can all be converted in guarantees
about $(\epsilon,\delta)$-differential privacy, one could just use the
latter for analyze all the examples we will show. The interest however
in performing as much reasoning as possible using these relaxations is
that one can achieve better values of the parameters. This will become
particularly evident in the last example.  
\mg{I am a bit hesitating with respect to the examples. Here my
  suggestion on how we should present them:
  \begin{itemize}
    \item One way marginal with Gaussian noise and with SinhNormal
      noise. This would be a warm up example showing how to use the
      logic for the relaxations.
\item Histogram in RDP to show how composition is just addition.
\item k-fold composition of Gaussian to show that using zCDP gives a
  better bound that DP.
  \end{itemize}
If we agree with this, we need to fix the examples. 
}
\tetsuya{
I did:
 \begin{itemize}
  \item One way marginal $\Rightarrow$ Add SinhNormal noise version; and tCDP verification; 
  \item Histogram$ \Rightarrow$ Fix typo and emphasis that we use additions of grading. 
  \item k-fold composition $\Rightarrow$ Give more expression.
\end{itemize}
}

\subsection{One-way Marginals}
As a warm up, we begin with the following classic example of a one-way
marginal algorithm with additive noise.
\begin{algorithm}[H]
	\caption{A mechanism estimates the attribute means}
\begin{algorithmic}[1]
    \Procedure{$\mathtt{AttMean}$}{$n \colon \mathtt{int}$, $\rho \colon \mathtt{real}$ (const.),
$x \colon \mathtt{bool}^n$ (dataset),
$i \colon \mathtt{int}$,
$y,z,w \colon \mathtt{real}$}
    	\State{$i \leftarrow 0$;$y \leftarrow 0$;}
            \While{$i < n$}				
                \State{$y \leftarrow y + x[i]$;~$i \leftarrow  i + 1$;}
            \EndWhile
	\State{$z \leftarrow y / n$;}
        \State{$w \xleftarrow{\$} \mathtt{Gauss}(z, 1/2n^2\rho)$;}
    \EndProcedure
\end{algorithmic}
\end{algorithm}
\noindent
We first show the R\'enyi-differential privacy of $\mathtt{AttMean}$.
We set a typing context $\Gamma$ of $\mathtt{AttMean}$ by
$x \colon \mathtt{bool}^n$ (dataset),
$i \colon \mathtt{int}$, and
$y,z,w \colon \mathtt{real}$.
We show the following judgment:
\[
\Gamma\vdash \mathtt{AttMean} \sim^{\mathtt{RDP}}_{\alpha\rho} \mathtt{AttMean}
\colon \mathtt{adj}(x\lrangle{1},x\lrangle{2}) \implies w\lrangle{1} = w\lrangle{2}.
\]
Here, the adjacent relation $\mathtt{adj}(x\lrangle{1},x\lrangle{2})$ means that two datasets $x\lrangle{1}$ and $x\lrangle{2}$ differs at most in one record.
Explicitly, we define it by the following relation expression:
\[
\mathtt{adj}(x\lrangle{1},x\lrangle{2})
= \bigwedge_{1 \leq i \leq n} \left( (x[i]\lrangle{1} \neq x[i]\lrangle{2}) \implies \bigwedge_{1 \leq j < i, i < j \leq n}(x[j]\lrangle{1} = x[j]\lrangle{2}) \right).
\]
The proof of this judgment follows by splitting $\mathtt{AttMean}$ into two
commands $\mathtt{LoopAM};\mathtt{NoiseG}$ where $\mathtt{NoiseG}= w
\xleftarrow{\$} \mathtt{Gauss}(z, 1/2n^2\rho)$, and $\mathtt{LoopAM}$ is the
rest of the program.
Since the loop part $\mathtt{LoopAM}$ is deterministic, by standard reasoning, we obtain:
\[
\Gamma \vdash \mathtt{LoopAM} \sim^{\alpha-\mathtt{RDP}}_{0} \mathtt{LoopAM} \colon \mathtt{adj}(x\lrangle{1},x\lrangle{2})
\implies (|z\lrangle{1} - z\lrangle{2}| \leq 1/n)
.
\]
By applying [RDP-G], for the noise-adding step $\mathtt{NoiseG}$ we have:
\[
\Gamma \vdash \mathtt{NoiseG} \sim^{\alpha-\mathtt{RDP}}_{\alpha\rho} \mathtt{NoiseG}  \colon  (|z\lrangle{1} - z\lrangle{2}| \leq 1/n) \implies  (w\lrangle{1} = w\lrangle{2}).
\]
Thus, by applying [seq] we complete the proof.
A similar proof could have been carried out with both the rules for differential privacy, zCDP, and tCDP.
Due to the simplicity of the example (that is, $\mathtt{LoopAM}$ is deterministic),
the resulting guarantee would have been the same.

\begin{algorithm}[H]
	\caption{A mechanism estimates the attribute means with SinhNormal noise}
\begin{algorithmic}[1]
    \Procedure{$\mathtt{AMSinh}$}{$n \colon \mathtt{int}$, $\rho \colon \mathtt{real}$ (const.),
$x \colon \mathtt{bool}^n$ (dataset),
$i \colon \mathtt{int}$,
$y,z,w \colon \mathtt{real}$}
    	\State{$i \leftarrow 0$;$y \leftarrow 0$;}
            \While{$i < n$}				
                \State{$y \leftarrow y + x[i]$;~$i \leftarrow  i + 1$;}
            \EndWhile
	\State{$z \leftarrow y / n$;}
        \State{$w \xleftarrow{\$} w + A \cdot\mathtt{arsinh}\left(\frac{1}{A}\mathtt{Gauss}(0,/2n^2\rho)\right)$;}
    \EndProcedure
\end{algorithmic}
\end{algorithm}
\noindent
We change the noise in the algorithm $\mathtt{AttMean}$ from Gaussian noise to SinhNormal noise.
Explicitly, we define a new algorithm $\mathtt{AMSinh}=\mathtt{LoopAM};\mathtt{NoiseSinh}$ where the noise-adding part is changed to 
 $\mathtt{NoiseSinh} = w \xleftarrow{\$} w + A \cdot\mathtt{arsinh}\left(\frac{1}{A}\mathtt{Gauss}(0,/2n^2\rho)\right)$, where $A$ is a constant satisfying $1 < 1/\sqrt{\rho} \leq A/n$.
In the similar way as the previous example $\mathtt{AttMean}$, for the loop part $\mathtt{LoopAM}$, we obtain:
\[
\Gamma \vdash \mathtt{LoopAM} \sim^{n \cdot A/8-\mathtt{tCDP}}_{0} \mathtt{LoopAM} \colon \mathtt{adj}(x\lrangle{1},x\lrangle{2})
\implies (|z\lrangle{1} - z\lrangle{2}| \leq 1/n).
\]
By applying [tCDP-SinhG], the noise-adding part $\mathtt{NoiseSinh}$ satisfies 
\[
\Gamma \vdash \mathtt{NoiseSinh} \sim^{n \cdot A/8-\mathtt{tCDP}}_{16\rho} \mathtt{NoiseSinh} \colon  (|z\lrangle{1} - z\lrangle{2}| \leq 1/n) \implies  (w\lrangle{1} = w\lrangle{2}).
\]
Thus, by applying [seq], we conclude that the algorithm $\mathtt{AMSinh}$ is $(16\rho,n \cdot A/8)$-tCDP.
\subsection{Histograms}
The following algorithm gives the histograms of dataset $x$ over the finite set $T$ with additive noise.
We use a primitive data type $T$ as a finite set of size $T$.
\begin{algorithm}[H]
	\caption{A mechanism estimates the histogram}
\begin{algorithmic}[1]
    \Procedure{$\mathtt{Histogram}$}{$n \mathtt{int}$, $\rho \colon \mathtt{real}$ (const.), $x \colon [T]^n$ (dataset), $y,z \colon \mathtt{real}^T$,$i \colon \mathtt{int}$}
    	\State{$i \leftarrow 0$; $y \leftarrow (0,\ldots,0)$;}
            \While{$i < n$}				
                \State{$y[x[i]] \leftarrow y[x[i]] + 1 $;~$i \leftarrow  i + 1$;}
            \EndWhile
	\State{$i \leftarrow 0$; $z \leftarrow (0,\ldots,0)$;}
          \While{$i < T$}				
                \State{$z[i] \xleftarrow{\$} \mathtt{Gauss}(y[i],1/\rho)$;~$i \leftarrow  i + 1$;}
            \EndWhile
    \EndProcedure
\end{algorithmic}
\end{algorithm}
\noindent
We show the zCDP of the algorithm $\mathtt{Histogram}$.
We set a typing context $\Gamma$ by $x \colon [T]^n$ (dataset), $y,z \colon \mathtt{real}^T$, and $i \colon \mathtt{int}$.
We want to prove the validity of the following judgment:
\[
\Gamma \vdash \mathtt{Histogram} \sim^{\mathtt{zCDP}}_{0,\rho} \mathtt{Histogram}
\colon \mathtt{adj}(x\lrangle{1},x\lrangle{2}) \implies z\lrangle{1} = z\lrangle{2}.
\]
Here, $\mathtt{adj}(x\lrangle{1},x\lrangle{2})$ is defined in the similar way as the previous algorithm.
We split the algorithm $\mathtt{Histogram}$ into $\mathtt{Histogram} = \mathtt{HGCalc};\mathtt{HGNoise}$ where 
$\mathtt{HGNoise}$ is the second loop for adding noise, and $\mathtt{HGCalc}$ is
the rest of the program that calculates a histogram without noise.
We can now define two additional assertions for $0 \leq K \neq L < T$ and $0 \leq I < n$:
\begin{align*}
\Phi_{I,K,L}
&= (x[I]\lrangle{1} \neq x[I]\lrangle{2}) \wedge (i \neq I \implies x[i]\lrangle{1} = x[i]\lrangle{2})\wedge (x[I]\lrangle{1} = K) \wedge (x[I]\lrangle{2} = L) 
\\
\Psi_{K,L}
&= (y[K]\lrangle{1} =  y[K]\lrangle{2} + 1) \wedge (y[L]\lrangle{1} + 1 =  y[L]\lrangle{2})\wedge (j \neq K,L \implies y[j]\lrangle{1} =  y[j]\lrangle{2}).
\end{align*}
It is easy to see that  $\mathtt{adj}(x\lrangle{1},x\lrangle{2}) \iff
\exists{I,K,L}.~\Phi_{I,K,L}$. Using this and some standard reasoning, we have 
\[
\Gamma \vdash \mathtt{HGCalc} \sim^{\mathtt{zCDP}}_{0,0} \mathtt{HGCalc}\colon \Phi(I,K,L)(x\lrangle{1},x\lrangle{2}) \implies \Theta(K,L)\land (i\lrangle{1} = 0)
\]
where $\Theta(K,L) = \Psi(K,L) \land (z\lrangle{1} = z\lrangle{2})\land(i\lrangle{1} = i\lrangle{2})$.
For proving the right judgment for $\mathtt{HGNoise}$ we also use the following
additional axiom for zCDP that concludes $(0,0)$-zCDP if both noises and inputs are the same (the soundness is rather straightforward):
\[\Gamma \vdash x_1 \xleftarrow{\$} \mathtt{Gauss}(e_1,\sigma^2) \sim^{\mathtt{zCDP}}_{0,0} x_2 \xleftarrow{\$} \mathtt{Gauss}(e_2,\sigma^2) \colon
(e_1\lrangle{1} = e_2\lrangle{2}) \implies (x_1\lrangle{1} = x_2\lrangle{2}).
\]
Now by using this axiom, [zCDP-G], and some basic reasoning for the loop we obtain:
\[
\Gamma \vdash \mathtt{HGNoise}\sim^{\mathtt{zCDP}}_{0,\rho} \mathtt{HGNoise}
\colon \Theta(K,L)\land (i\lrangle{1} = 0) \implies \Theta(K,L)\land (i\lrangle{1} = T).
\]
Roughly speaking, we may regard $\mathtt{HGNoise}$ as a composition $c[0];c[1];\cdots;c[T-1]$
where $c[j]$ is the $j$-th execution of the loop body of $\mathtt{HGNoise}$.
For $j \neq K,L$ by using the new axiom,
\[
\Gamma \vdash c[j] \sim^{\mathtt{zCDP}}_{0,0} c[j] \colon \Theta(K,L) \land (i\lrangle{1} = j) \implies \Theta(K,L)\land(i\lrangle{1} = j+1).
\]
On the other hand, for $j = K,L$ by applying [zCDP-G] (with $\sigma^2 = \rho/2$), we obtain
\[
\Gamma \vdash c[j] \sim^{\mathtt{zCDP}}_{0,\rho/2} c[j] \colon \Theta(K,L) \land (\lrangle{1} = j) \implies \Theta(K,L)\land(i\lrangle{1} = j+1).
\]
Note that the second case occurs twice.  The [seq] rule sums up the grading of
each execution $c[j]$, and we conclude $\rho$-zCDP of $\mathtt{HGNoise}$.
Finally, by using [seq] and some conditional computations, we complete the proof. 
\fi

\iffull \subsection{A $k$-fold Gaussian mechanism} \fi
%
%
Consider a type $\mathtt{DATA}$ of dataset and an predicate $\mathtt{ADJ}({-},{=})$ of adjacency for the type $\mathtt{DATA}$, and 
consider $K$ queries $q(i,-) \colon \mathtt{DATA} \to \mathtt{real}$
($0 \leq i < K$) with sensitivity $1$, that is, 
\[
\mathtt{ADJ}(D,D') \implies |q(i,D) - q(i,D') | \leq 1.
\]
We want now to prove private the following $K$-fold Gaussian mechanism.  Even
though standard DP can already be handled by other verification techniques, our
proof applies the conversion rules between DP and zCDP along with composition in
zCDP, yielding a more precise analysis for standard DP.

\begin{algorithm}[H]
	\caption{Sum of $K$ Gaussian mechanisms}
\begin{algorithmic}[1]
    \Procedure{$\mathtt{FoldG}_K$}{$K \colon \mathtt{int}$, $\sigma \colon \mathtt{real}$ (const.), $D \colon \mathtt{DATA}$, $x,y,z \colon\mathtt{real}$, $i \colon \mathtt{int}$ }
    	\State{$i \leftarrow 0; z \leftarrow 0;$}
            \While{$i < K$}				
                \State{$x \leftarrow q(i,D);
                y \xleftarrow{\$} \mathtt{Gauss}(0,\sigma);
                z \leftarrow x+y+z;
                i \leftarrow i + 1;$}\label{algorithm_loopbody}
            \EndWhile
\EndProcedure
\end{algorithmic}
\end{algorithm}
\noindent
\iffull
We set a typing context of $\mathtt{FoldG}_K$ by
$D \colon \mathtt{DATA}$, $x,y,z \colon\mathtt{real}$, and $i \colon \mathtt{int}$.
\fi
Following sensitivity of queries $q$, for any $0 \leq i < K$ we may assume 
\[
\textstyle
\Gamma \vdash x \leftarrow q(i,D) \sim^{\mathtt{zCDP}}_{0,0} x \leftarrow q(i,D)
\colon \mathtt{ADJ}(D\lrangle{1},D\lrangle{2}) \implies |x\lrangle{1} - x\lrangle{2}| \leq 1
.
\]
Thus, for the loop body $c$ (line \ref{algorithm_loopbody}), by applying [zCDP-G], [seq] and [assn], we have
\[
\textstyle
\Gamma \vdash c \sim^{\mathtt{zCDP}}_{0,1/2\sigma^2} c
\colon \mathtt{ADJ}(D\lrangle{1},D\lrangle{2})\land (z\lrangle{1} = z\lrangle{2}) \implies z\lrangle{1} = z\lrangle{2}
.
\]
Then, by applying [assn], [seq], and [while] (the proof rule for while-loop) rules, 
we conclude
\[
\Gamma \vdash \mathtt{FoldG}_K \sim^{\mathtt{zCDP}}_{0,K/2\sigma^2} \mathtt{FoldG}_K \colon
\mathtt{ADJ}(D\lrangle{1},D\lrangle{2}) \implies z\lrangle{1} = z\lrangle{2}.
\]
Hence, the algorithm $\mathtt{FoldG}_K$ is $(0,K/2\sigma^2)$-zCDP.
Furthermore, by applying [z/D], we conclude that the algorithm $\mathtt{FoldG}_K$ is $\left(\frac{K}{2\sigma^2} + \frac{\sqrt{2K\log(1/\delta)}}{\sigma},\delta \right)$-DP for any $0 < \delta < 1/2$.

This analysis gives a more precise bound compared to reasoning in terms of
standard differential privacy. First, by [DP-G], [seq] and [assn], for any
$0 < \delta_1 < 1/2$, the loop body $c$ satisfies
\[
\Gamma \vdash c \sim^{\mathtt{DP}}_{\max((1+\sqrt{3})/2\sigma, \sqrt{2\log(0.66/\delta_1)}/\sigma),\delta_1} c
\colon\mathtt{adj}(D\lrangle{1},D\lrangle{2})\land (z\lrangle{1} = z\lrangle{2}) \implies z\lrangle{1} = z\lrangle{2}.
\]
Let $\varepsilon = \max((1+\sqrt{3})/2
\sigma, \sqrt{2\log(0.66/\delta_1)}/\sigma)$.
The algorithm $\mathtt{FoldG}_K$ can be seen as $K$-fold adaptive composition of
the loop body $c; \cdots ; c$.  By applying the advanced composition theorem
\citep[Theorem 3.20]{DworkRothTCS-042}, the algorithm $\mathtt{FoldG}_K$ is
\[\left(\varepsilon \cdot\sqrt{2K\log(1/\delta_2)} + K\varepsilon^2, K\delta_1 +
\delta_2 \right)\text{-DP} \quad\text{for any } 0 < \delta_1,\delta_2 < 1/2.
\] 
We compare this bound and the bound given in the avove.
When $\delta_2 < 0.4$,
we have $2\log(0.66/\delta_2) > 1$.
We also have $\varepsilon > 1.36/\sigma$ by the definition.
Then, we can compute:
\[
\frac{K}{2\sigma^2} + \frac{\sqrt{2K\log(1/\delta_2)}}{\sigma}
< \frac{K}{2\sigma^2} + \frac{\sqrt{2K\log(1/\delta_2)}}{\sigma} \cdot \sqrt{2\log(0.66/\delta_1)} \leq \varepsilon \cdot\sqrt{2K\log(1/\delta_2)} + K\varepsilon^2.
\]
Hence,
$\varepsilon \cdot\sqrt{2K\log(1/\delta_2)} + K\varepsilon^2 >
\frac{K}{2\sigma^2} + \frac{\sqrt{2K\log(1/\delta)}}{\sigma}$ whenever $\delta = K\delta_1 + \delta_2$ and $\delta_2 < 0.4$.

We can conclude that verification via zCDP is actually better than advanced
composition for the algorithm $\mathtt{FoldG}$.  First, in the verification via
zCDP, the approximation error $\delta$ is given regardless of the number of
queries $K$.  Second, if the approximation error satisfies $\delta < 0.4$ then
the verification is significantly better than advanced composition.  The
restriction $\delta < 0.4$ is quite weak since the approximation error $\delta$
in the $(\varepsilon,\delta)$-DP is thought as the probability of failure of
$\varepsilon$-DP.  Moreover in practical use of $(\varepsilon,\delta)$-DP, the
parameter $\delta$ is usually taken to be quite small (e.g., $\delta \approx
10^{-5}$).
\section{Related Works}
\label{sec:rw}

\PARAGRAPH{Relational liftings for $f$-divergences}
As we have mentioned, our work is inspired by work on verifying probabilistic
relational properties involving $f$-divergences by \citet{BartheOlmedo2013}; we
generalize their results to a broader class of divergences and also to handle
continuous distributions.
%
%
Barthe and Olmedo also consider $f$-divergences that satisfy a more limited
version of composability, called \emph{weak composability}. Roughly, these
composition results only apply when corresponding pairs of distributions have
equal weight; the KL-divergence, Hellinger distance, and $\chi^2$ divergences
only satisfy this weaker version of composability. While we do not detail this
extension, our framework can naturally handle weakly composable
divergences in the continuous case.

A similar approach has also been used by \citet{BartheFGAGHS16} in
the context of an higher order functional language for reasoning about Bayesian
inference. Their type system uses a graded monad to reason about
$f$-divergences. The graded monad supports only discrete distributions and is
interpreted via a set-theoretic semantics, again using the lifting by
\citet{BartheOlmedo2013}.

\PARAGRAPH{Relational liftings for differential privacy}
Approximate relational liftings were originally proposed for program logics
targeting differential privacy.
The first such system used a one-witness definition of
lifting~\citep{Barthe:2012:PRR:2103656.2103670}, which was subsequently refined
to several notions of two-witness lifting~\citep{BartheOlmedo2013,Barthe:2016:APC:2976749.2978391}.
\citet{Sato2016MFPS} developed approximate liftings and a program logic for
continuous distribution using witness-free lifting based on a categorical monad
lifting~\citep{Katsumata2005,katsumata_et_al:LIPIcs:2015:5532}.
\iffull
A \emph{witness-free} relational lifting for differential privacy was
introduced by~\citet{Sato2016MFPS}.  This can be seen as an application of the
general construction of \emph{graded relational lifting} \citep[Section
5]{Katsumata2014PEM} to the Giry monad, using the technique of \emph{codensity
lifting} \citep[Section 3.3]{katsumata_et_al:LIPIcs:2015:5532} instead of
$\top\top$-lifting.
The witness-free relational lifting by~\citet{Sato2016MFPS} sends a binary relation $R$ between
measurable spaces $X,Y$ to the following one between
$\mathcal{G}X,\mathcal{G}Y$:
\begin{align*}
  R^{\top\top(\varepsilon,\delta)}  &= \bigcap_{(k,l) \colon R \dot\to S^{(\varepsilon',\delta')}} \inverse{(k^\sharp \times l^\sharp)}
   S^{(\varepsilon+\varepsilon',\delta+\delta')} \\
  \text{where}~~S^{(\varepsilon',\delta')} & =\Set{ (x,y) \in \mathcal{G}1
  \times \mathcal{G}1 \mid x \leq e^{\varepsilon'} y + \delta' }.
\end{align*}
where $\mathcal{G}$ is the sub-Giry monad, 
$k^\sharp$ and $l^\sharp$ denote the Kleisli extensions of $k$ and $l$ respectively,
$\dot\to$ denotes a relation-preserving map, and $\top\top$ is used to denote
the codensity lifting and to distinguish it from our 2-witness lifting.
Here, the intersection is taken over all measurable functions
$k:X\rightarrow\mathcal{G}1,l:Y\rightarrow\mathcal{G}1$ mapping pairs related by
$R$ to those related by $S^{(\varepsilon',\delta')}$. We note that the binary
relation $S^{(\varepsilon',\delta')}$ is a parameter of this witness-free
lifting, and by changing it, we can derive other graded relational liftings of
$\mathcal{G}$.

Checking the membership for
 $R^{\top\top(\varepsilon,\delta)}$ is complex: we have to test the
pair $(x,y)$ against every pair $(k,l)$ of measurable functions such
that $(k,l)\colon R \dot\to S^{(\varepsilon, \delta)}$. Fortunately,
since the divergence
$\Delta^{\mathtt{DP}(\varepsilon)}$ is defined by a linear inequality
of measures, the witness-free lifting
$R^{\top\top(\epsilon,\delta)}$ can be \emph{simplified} to the following
\[
R^{\top\top(\varepsilon,\delta)} = \Set{(d_1,d_2) \in \mathcal{G}X \times
\mathcal{G}Y \mid \forall{A \subseteq \Sigma_X}.~d_1(A) \leq e^\varepsilon
d_2(R(A))+\delta} .
\]

While we would like to generalize this lifting construction to handle more
general divergences for RDP, zCDP, and tCDP, there are at least two obstacles.  First,
it is not clear how to find a parameter $S$ to derive the suitable graded
relational lifting for a given general divergence. Second, even if we can find a
suitable parameter $S$, it is awkward to work with the lifting unless we can
simplify the large intersection into a more convenient form.
In contrast, 2-witness liftings seem more concrete and easier to work with: It
suffices to give witness distributions to check the membership of lifted
relations. 
\fi

In the discrete case, witness-free liftings are equivalent to the
witness-/span-based liftings by~\citet{barthe_et_al:LIPIcs:2017:7435}. Recent
work also considers liftings with more fine-grained parameters that can vary
over different pairs of samples~\citep{AH17}.

\PARAGRAPH{Other techniques for verifying privacy}
R\'enyi and zero-concentrated differential privacy were recently proposed in the
differential privacy literature; to the best of our knowledge, we are the first
to verify these properties. In contrast, there are now numerous systems
targeting differential privacy using a wide range of techniques beyond program
logics, including dynamic analyses~\citep{McSherry:2009:PIQ:1559845.1559850},
linear~\citep{ReedPierce10,GHHNP13,AGGH14} and
dependent~\citep{Barthe:2015:HAR:2676726.2677000} type systems, product
programs~\citep{6957126}, partial
evaluation~\citep{Winograd-Cort:2017:FAD:3136534.3110254}, and
constraint-solving~\citep{zhang2016autopriv,AH17}; see the recent
survey~\citep{Murawski:2016:2893582} for more details.
\section{Conclusion and Future Work}
\label{sec:conclusion}
We have developed a framework for reasoning about three relaxations of
differential privacy:  R\'enyi differential privacy, zero concentrated
differential privacy, and truncated concentrated differential privacy. We
extended the notion of divergences to a more general class, and to support
subprobability measures. Additionally, we have introduced a novel notion of
approximate span-lifting supporting these divergences and continuous
distributions.

One promising direction for future work is to study the moment-accountant
composition method~\citep{AbadiCGMMT016}. This composition method tracks the
moments of the privacy loss random variable, although it does not directly
correspond to composition for RDP or zCDP. Another interesting direction would
be to analyze recently-proposed RDP mechanisms for posterior
sampling~\citep{Geumlek17nips}, and the GAP-Max tCDP algorithm
by~\citet{BDRS18}.

\bibliography{header,reference}
\iffull
\appendix
\newpage
\section{Continuity of $f$-divergences of Subprobability Measures}
\label{sec_cont_f-div}

In this section we show the subprobability version of continuity of
$f$-divergences~\citep[Theorem 16]{1705001_2006} in a different way from the paper~\citep{1705001_2006}.

\begin{theorem}[Theorem \ref{f-divergence:continuity} / Subprobability version of{\citep[Theorem 16]{1705001_2006}}]
For any weight function $f$, the $f$-divergence $\Delta^f$ is continuous: for
any \emph{subprobability measures} $\mu_1, \mu_2 \in \mathcal{G}X$ on $X$, we
have
\[
\Delta^f_X(\mu_1,\mu_2)
= \sup \Set{\sum_{i = 0}^n \mu_2(A_i)f\left(\frac{\mu_1(A_i)}{\mu_2(A_i)}\right)
  \mid \{A_i\}_{i = 0}^n
\text{ is a measurable finite partition of } X}.
\]
\end{theorem}
To prove this proposition,
we introduce the singularity of measures.
%
%
Two measures $\mu_1$ and $\mu_2$ on $X$ are said to be mutually singular (written $\nu_1 \perp \nu_2$) if there are partition $A_1, A_2 \in \Sigma_X$ of $X$ such that $\mu_i(E) = \mu_i(A_i \cap E)$ for any $E \in \Sigma_X$ ($i = 1,2$).

\begin{lemma}[Lebesgue's Decomposition Theorem]
Let $\mu_1$ and $\mu_2$ be $\sigma$-finite measures on $X$.
There are unique finite measures $\mu_1^\bullet$ and $\mu_1^\perp$ on $X$ such that $\mu_1^\bullet \ll \mu_2$ and $\mu_1^\perp \perp \mu_2$.
\end{lemma}

We recall that the $f$-divergence for subprobability measures is defined by
for any $\mu_1,\mu_2, \mu \in \mathcal{G}X$ such that $\mu_1,\mu_2 \ll \mu$,
\[
\Delta^f_X(\mu_1,\mu_2) = \int_X \frac{d\mu_2}{d\mu} f\left(\frac{d\mu_1/d\mu}{d\mu_2/d\mu}\right) ~d\mu.
\]
We remark that $\mu$ satisfying $\mu_1,\mu_2 \ll \mu$ always exists (e.g. $(\mu_1 + \mu_2)/2$), and 
the $\Delta^f_X(\mu_1,\mu_2)$ does not depend on the choice of $\mu$.
We want to prove the continuity:
\[
\Delta^f_X(\mu_1,\mu_2) =
\sup\Set{\sum_{i = 0}^n \mu_2(A_i) f\left(\frac{\mu_1(A_i)}{\mu_2(A_i)}\right)
\mid \{A_i\}_{i = 0}^n \text{ is a finite measurable partition of } X}.
\]

We define the following restricted sum of $f$-divergences.
For any measurable subset $D \in \Sigma_X$,
\begin{align*}
\Delta^{f}_X(\mu_1,\mu_2)|_D
&= \int_D \frac{d\mu_2}{d\mu} f\left(\frac{d\mu_1 / d\mu}{d\mu_2 / d\mu} \right) d\mu.\\
\overline{\Delta}^{f}_X(\mu_1,\mu_2)|_D
&= \sup\Set{\sum_{i = 0}^n
\mu_2(A_i)f\left(\frac{\mu_1(A_i)}{\mu_2(A_i)}\right) \mid \{A_i\}_{i = 0}^n \text{ is a finite measurable partition of } D}\\
&= \sup\Set{\sum_{i \in I}
\mu_2(\inverse{k}(i))f\left(\frac{\mu_1(\inverse{k}(i))}{\mu_2(\inverse{k}(i))}\right)
\mid I \in \FinSet, k \colon D \to I}
\end{align*}
Of course, $\Delta^{f}_X(\mu_1,\mu_2) = \Delta^{f}_X(\mu_1,\mu_2)|_X$.
We write $\overline{\Delta}^{f}_X(\mu_1,\mu_2) =  \overline{\Delta}^{f}_X(\mu_1,\mu_2)|_X$ 

We temporary consider a positive weight function $f$.

\begin{lemma}\label{lem:f-div:conti:1}
If $\mu_1 \ll \mu_2$ then $\Delta^{f}_X(\mu_1,\mu_2)|_D \leq \overline{\Delta}^{f}_X(\mu_1,\mu_2)|_D$ for any $D \in \Sigma_X$.
\end{lemma}
\begin{proof}
Since $\mu_1 \ll \mu_2$, we may assume $\mu = \mu_2$ (hence $d\mu_2 / d\mu = 1$).
Then, we have $\Delta^{f}_X(\mu_1,\mu_2)|_D = \int_D f(\frac{d\mu_1}{d\mu_2}) d\mu_2$.
Since $f$ is convex, there is $\alpha \in \mathbb{R}_{\geq 0}$ which makes that
$f$ is monotone increasing on the interval $[0,\alpha)$ and monotone decreasing on $[\alpha, \infty)$.
Let $\{A_i\}_{i = 0}^n$ be an arbitrary finite partition of $D$ which is finer than the partition
$\{\inverse{\frac{d\mu_1}{d\mu_2}}([0,\alpha))\cap D,\inverse{\frac{d\mu_1}{d\mu_2}}([\alpha,\infty))\cap D\}$.
The function $f\circ \frac{d\mu_1}{d\mu_2}$ is either monotone increasing or monotone decreasing, on each partition $A_i$.
Hence, $\inf_{x \in A_i}f(\frac{d\mu_1}{d\mu_2})(x)$ is either $f(\inf_{x \in A_i}\frac{d\mu_1}{d\mu_2}(x))$ or $f(\sup_{x \in A_i}\frac{d\mu_1}{d\mu_2}(x))$.
From the mean-value theorem for measures, we obtain
\[
\inf_{x \in A_i}\frac{d\mu_1}{d\mu_2}(x)
\leq \frac{\mu_1(A_i)}{\mu_2(A_i)}
\leq \sup_{x \in A_i}\frac{d\mu_1}{d\mu_2}(x).
\]
Hence,
\[
\sum_{i = 0}^n \mu_2(A_i) \inf_{x \in A_i}f(\frac{d\mu_1}{d\mu_2})(x)
\leq \sum_{i = 0}^n \mu_2(A_i)f(\frac{\mu_1(A_i)}{\mu_2(A_i)}).
\]
Since $\{A_i\}_{i = 0}^n$ is arbitrary, we conclude $\Delta^{f}_X(\mu_1,\mu_2)|_D \leq \overline{\Delta}^{f}_X(\mu_1,\mu_2)|_D$.
\end{proof}

\begin{lemma}\label{lem:f-div:conti:2}
If $\mu_1 \ll \mu_2$ and the Radon-Nikodym derivative ${d\mu_1}/{d\mu_2}$ is bounded on $D$ then $\Delta^{f}_X(\mu_1,\mu_2)|_D = \overline{\Delta}^{f}_X(\mu_1,\mu_2)|_D$.
\end{lemma}
\begin{proof}
We fix a positive integer $1 \leq K \in \mathbb{N}$ such that $0 \leq \frac{d\mu_1}{d\mu_2} \leq M$.
For given $N \in \mathbb{N}$, we define the partition $\{A_i\}_{i=0}^{2^N K}$ of $D$ by
\[
A_i = \left(\inverse{\left(\frac{d\mu_1}{d\mu_2}\right)}(B_i)\right) \cap D, 
\quad
B_i =
\begin{cases}
[ \frac{i}{2^N}, \frac{i+1}{2^N} ) & 0 \leq i < 2^N\\
\{1\} & i = 2^N\\
( \frac{i-1}{2^N}, \frac{i}{2^N} ] & 2^N <  i \leq 2^N K.
\end{cases}
\]

Since $\mu_1 \ll \mu_2$ and $0f(0/0)=0$, 
if $\mu_2(A_i) = 0$ then $\mu_2(A_i)\frac{\mu_1(A_i)}{\mu_2(A_i)} = 0$.
If $\mu_2(A_i) > 0$ then
$
\left| \frac{d\mu_1}{d\mu_2}(x) - \frac{\mu_1(A_i)}{\mu_2(A_i)} \right| ~\leq~ 2^{-(N-1)}
$
for all $x \in A_i$, from the definition of $\{A_i\}_{i=0}^{2^N K}$,
\[
{\frac{i-1}{2^N}}
~\leq~{\inf_{x \in A_i} \frac{d\mu_1}{d\mu_2}(x)}
~\leq~{\frac{\mu_1(A_i)}{\mu_2(A_i)}}
~\leq~{\sup_{x \in A_i} \frac{d\mu_1}{d\mu_2}(x)}
~\leq~{\frac{i+1}{2^N}}.
\]

Consider an arbitrary $\varepsilon > 0$.
Since $f$ is uniformly continuous on the closed interval $[0,K]$,
there are large enough $N_2 \in \mathbb{N}$ and the corresponding partition $\{A_i\}_{i=0}^{2^N K}$
such that
\[
\mu_2(A_i) > 0 \implies 
\left| \inf_{x \in A_i}f(\frac{d\mu_1}{d\mu_2})(x) - f(\frac{\mu_1(A_i)}{\mu_2(A_i)}) \right| < \varepsilon
\]
Hence, for any partition $\{C_i\}_{i=0}^{n}$ of $D$ finer than $\{A_i\}_{i=0}^{2^N K}$,
we obtain
\[
\sum_{i=0}^{n}\mu_2(C_i) f\left(\frac{\mu_1(C_i)}{\mu_2(C_i)}\right)
~\leq~ \sum_{i=0}^{n}\mu_2(C_i) \left(\inf_{x \in C_i}f\left(\frac{d\mu_1}{d\mu_2}\right)(x)\right) + \varepsilon.
\]
This implies $\overline{\Delta}^{f}_X(\mu_1,\mu_2)|_D \leq \Delta^{f}_X(\mu_1,\mu_2)|_D + \varepsilon$.
Since $\varepsilon > 0$ is arbitrary, we conclude $\overline{\Delta}^{f}_X(\mu_1,\mu_2)|_D \leq \Delta^{f}_X(\mu_1,\mu_2)|_D$.
\end{proof}

\begin{lemma}\label{lem:f-div:conti:3}
We have $\Delta^{f}_X(\mu_1,\mu_2)|_D = \overline{\Delta}^{f}_X(\mu_1,\mu_2)|_D$ when $\mu_1 \ll \mu_2$.
\end{lemma}
\begin{proof}
Let $D_n = \left(\inverse{\left(\frac{d\mu_1}{d\mu_2}\right)}[n,n+1)\right) \cap D$ ($n \in \mathbb{N}$).
From Jensen's inequality, we obtain for any partition $\{A_i\}_{i = 0}^m$ of $D$,
\begin{align*}
\sum_{i = 0}^m \mu_2(A_i) f\left( \frac{\mu_1(A_i)}{\mu_2(A_i)}\right)
&= \sum_{i = 0}^m (\sum_{n \in \mathbb{N}} \mu_2(D_n \cap A_i)) f\left( \frac{\sum_{n \in \mathbb{N}} \mu_1(D_n \cap A_i)}{\sum_{n \in \mathbb{N}} \mu_2(D_n \cap A_i)}\right)\\
&\leq \sum_{i = 0}^m \sum_{n \in \mathbb{N}} \mu_2(D_n \cap A_i) f\left(\frac{\mu_1(D_n \cap A_i)}{\mu_2(D_n \cap A_i)}\right)\\
& =  \sum_{n \in \mathbb{N}} \sum_{i = 0}^m \mu_2(D_n \cap A_i) f\left(\frac{\mu_1(D_n \cap A_i)}{\mu_2(D_n \cap A_i)}\right)
\end{align*}
This implies $\overline{\Delta}^{f}_X(\mu_1,\mu_2)|_D \leq \sum_{n=0}^{\infty}\overline{\Delta}^{f}_X(\mu_1,\mu_2)|_{D_n}$ for each $n \in \mathbb{N}$.

Since the Radon-Nikodym derivative $\frac{d\mu_1}{d\mu_2}$ is bounded on each $D_n$, by Lemmas \ref{lem:f-div:conti:1} and \ref{lem:f-div:conti:2},
$\Delta^{f}_X(\mu_1,\mu_2)|_{D_n} = \overline{\Delta}^{f}_X(\mu_1,\mu_2)|_{D_n}$ for each $n \in \mathbb{N}$.
Hence,
\[
\overline{\Delta}^{f}_X(\mu_1,\mu_2)
\leq \sum_{n=0}^{\infty}\overline{\Delta}^{f}_X(\mu_1,\mu_2)|_{D_n}
=\sum_{n=0}^{\infty}\Delta^{f}_X(\mu_1,\mu_2)|_{D_n}
=\Delta^{f}_X(\mu_1,\mu_2)
\leq \overline{\Delta}^{f}_X(\mu_1,\mu_2).
\]
This implies $\Delta^{f}_X(\mu_1,\mu_2) = \overline{\Delta}^{f}_X(\mu_1,\mu_2)$.
\end{proof}

\begin{proof}[Theorem \ref{f-divergence:continuity}, Positive Case]
We show that for any positive weight function $f$, the continuity $\overline{\Delta}^{f}_X(\mu_1,\mu_2) = \Delta^{f}_X(\mu_1,\mu_2)$ holds.
Let $(\mu_1^\bullet, \mu^\perp_1)$ be the Lebesgue decomposition of $\mu_1$ with respect to $\mu_2$.
Since $(\mu_1^\bullet, \mu^\perp_1)$ is the Lebesgue decomposition of $\mu_1$ with respect to $\mu_2$, there is $A \in \Sigma_X$ such that $\mu_2(E) = \mu_2(E \setminus A)$ and $\mu_1^\perp(E) = \mu_1^\perp(E \cap A)$ for any $E \in \Sigma_X$.
The subset $A$ also satisfies $\mu_1(E \setminus A) = \mu_1^\bullet(E \setminus A)$ for any $E \in \Sigma_X$.
We then obtain 
\begin{align*}
\Delta^{f}_X(\mu_1,\mu_2)
&=
\Delta^{f}_X(\mu_1,\mu_2)|_{X \setminus A} + \Delta^{f}_X(\mu_1,\mu_2)|_{A}
= \Delta^{f}_X(\mu_1^\bullet,\mu_2)|_{X \setminus A} + \Delta^{f}_X(\mu_1^\perp,\mu_2)|_{A}\\
&= \overline{\Delta}^{f}_X(\mu_1^\bullet,\mu_2)|_{X\setminus A} + \overline{\Delta}^{f}_X(\mu_1^\perp,\mu_2)|_{A}
= \overline{\Delta}^{f}_X(\mu_1,\mu_2)|_{X\setminus A} + \overline{\Delta}^{f}_X(\mu_1,\mu_2)|_{A} = \overline{\Delta}^{f}_X(\mu_1,\mu_2)
\end{align*}
From Lemma \ref{lem:f-div:conti:3}, $\Delta^{f}_X(\mu_1^\bullet,\mu_2)|_{X \setminus A} = \overline{\Delta}^{f}_X(\mu_1^\bullet,\mu_2)|_{X\setminus A}$ holds, and using the dual $f^\ast$ we have
\[
\Delta^{f}_X(\mu_1^\perp,\mu_2)|_{A}
= \int_{A} \frac{d\mu_2}{d\mu}f\left(\frac{d\mu_1^\perp/d\mu}{d\mu_2/d\mu}\right)d\mu
= \int_{A} f^\ast(0) \frac{d\mu_1^\perp}{d\mu} d\mu
= f^\ast(0)\mu_1(A)
= \overline{\Delta}^{f}_X(\mu_1^\perp,\mu_2)|_{A}.
\]
\end{proof}

\begin{proof}[Theorem \ref{f-divergence:continuity}, General case]
We show the continuity of $\Delta^{f}$ for arbitrary weight function $f$.
Let $\alpha,\beta \colon \mathbb{R}_{\geq 0} \to \mathbb{R}$ the functions be defined by $\alpha(t) = a$ and $\beta(t) = bt$ respectively where $a, b \geq 0$.
Since $f$ is convex, there are $\alpha$ and $\beta$ that makes $f+\alpha+\beta$ positive.
Hence,
\begin{align*}
\overline{\Delta}^{f}_X(\mu_1,\mu_2) + a\mu_2(X) + b\mu_1(X)
& = \overline{\Delta}^{f+\alpha+\beta}_X(\mu_1,\mu_2) \\
& = \Delta^{f+\alpha+\beta}_X(\mu_1,\mu_2) \\
& = \Delta^{f}_X(\mu_1,\mu_2) + a\mu_2(X) + b\mu_1(X).
\end{align*}
This completes the proof.
\end{proof}

\section{Omitted Structures of the Program Logic}
\label{sec:logic_full}
\subsection{Typing Rules for Expressions and Programs}
\label{sec:typing_pWHILE}
Before we give the semantics of programs, we first give a type system for
expressions, distributions, and programs. A typing context is a finite set $\Gamma =
\{x_1\colon\tau_1,x_2\colon\tau_2,\ldots,x_n\colon\tau_n \}$ of pairs of a
variable and a value type such that each variable occurs only once in the
context. The type system is largely standard, with two kinds of judgments:
$\Gamma \vdash^t e \colon \tau$ states that expression $e$ has type $\tau$ in
context $\Gamma$, while $\Gamma \vdash^p \nu : \tau$ states that $\nu$ is a
distribution over $\tau$ in context $\Gamma$.  The third judgment $\Gamma \vdash
c$ states that program $c$ is well-typed in context $\Gamma$, e.g., all guards
are booleans, assignments are well-typed, etc. The expression typing rules are
as follows:
{\small
\[
	\AxiomC{$x \colon \tau \in \Gamma$}
	\UnaryInfC{$\Gamma \vdash^t x \colon \tau$}
	\DisplayProof
\quad
	\AxiomC{$\Gamma\vdash^t e_1\colon\tau$ \quad $\Gamma\vdash^t e_2\colon\tau$}
	\UnaryInfC{$\Gamma \vdash^t e_1 \oplus e_2 \colon \tau$}
	\DisplayProof
\quad
	\AxiomC{$\Gamma\vdash^t e_1\colon\tau$ \quad $\Gamma\vdash^t e_2\colon\tau$}
	\UnaryInfC{$\Gamma \vdash^t e_1 \bowtie e_2 \colon \mathtt{bool}$}
	\DisplayProof
\quad
 	\AxiomC{$\Gamma\vdash^t e_1\colon\tau^d$ \quad $\Gamma\vdash^t e_2\colon\mathtt{int}$}
 	\UnaryInfC{$\Gamma \vdash^t e_1[e_2]\colon \tau$}
 	\DisplayProof
\]
\[
\quad
	\AxiomC{$\Gamma \vdash^t e\colon \mathtt{real}$}
	\UnaryInfC{$\Gamma \vdash^p \mathtt{Bern}(e)\colon \mathtt{bool}$}
	\DisplayProof
\quad
	\AxiomC{$\Gamma \vdash^t e_1 \colon \mathtt{real} \quad \Gamma \vdash^t e_2 \colon \mathtt{real}$}
	\UnaryInfC{$\Gamma \vdash^p \mathtt{Lap}(e_1,e_2)\colon \mathtt{real}$}
	\DisplayProof
\quad
	\AxiomC{$\Gamma \vdash^t e_1 \colon \mathtt{real} \quad \Gamma \vdash^t e_2 \colon \mathtt{real}$}
	\UnaryInfC{$\Gamma \vdash^p \mathtt{Gauss}(e_1,e_2)\colon \mathtt{real}$}
	\DisplayProof
\]
\[
	\AxiomC{$\Gamma \vdash^t e\colon \tau$}
	\UnaryInfC{$\Gamma \vdash^p \mathtt{Dirac}(e)\colon \tau$}
	\DisplayProof
\quad
	\AxiomC{\hspace{1em}}
	\UnaryInfC{$\Gamma \vdash \mathtt{skip}$}
	\DisplayProof
\quad
	\AxiomC{
	$\Gamma, x\colon\tau \vdash^p \nu \colon \tau$}
	\UnaryInfC{$\Gamma, x\colon\tau \vdash x \xleftarrow{\$} \nu$}
	\DisplayProof
\quad
	\AxiomC{$\Gamma \vdash c_1 $\quad$\Gamma \vdash c_2$}
	\UnaryInfC{$\Gamma \vdash c_1 ; c_2$}
	\DisplayProof
\]
\[
	\AxiomC{$\Gamma \vdash^t b\colon\mathtt{bool}$\quad$\Gamma \vdash c_1$\quad$\Gamma \vdash c_2$}
	\UnaryInfC{$\Gamma \vdash \mathtt{if}~b~\mathtt{then}~c_1~\mathtt{else}~c_2$}
	\DisplayProof
\quad
	\AxiomC{$\Gamma \vdash^t b\colon\mathtt{bool}$\quad$\Gamma \vdash c$}
	\UnaryInfC{$\Gamma \vdash \mathtt{while}~b~\mathtt{do}~c$}
	\DisplayProof
\]
}

\subsubsection{Forming Relation Expressions}
The judgment $\Gamma \vdash^R \Phi$ states that the relation expression $\Phi$
is well-formed in context $\Gamma$.

{\small
\[
	\AxiomC{$\Gamma \vdash^t e_1 \bowtie e_2 \colon \mathtt{bool}$}
	\UnaryInfC{$\Gamma \vdash^R e_1 \lrangle{1} \bowtie e_2 \lrangle{2}$}
	\DisplayProof
\quad
	\AxiomC{$\Gamma \vdash^t (e_1 \oplus_1 e_2) \bowtie (e_3 \oplus_2 e_4) \colon \mathtt{bool}$}
	\UnaryInfC{$\Gamma \vdash^R (e_1 \lrangle{1} \oplus_1 e_2 \lrangle{2}) \bowtie (e_3 \lrangle{1} \oplus_2 e_4 \lrangle{2})$}
	\DisplayProof
\]
\[
	\AxiomC{$\Gamma \vdash^R \Phi$\quad$\Gamma \vdash^R \Psi$}
	\UnaryInfC{$\Gamma \vdash^R \Phi \wedge \Psi$}
	\DisplayProof
\quad
	\AxiomC{$\Gamma \vdash^R \Phi$\quad$\Gamma \vdash^R \Psi$}
	\UnaryInfC{$\Gamma \vdash^R \Phi \vee \Psi$}
	\DisplayProof
\quad
	\AxiomC{$\Gamma \vdash^R \Phi$}
	\UnaryInfC{$\Gamma \vdash^R \neg \Phi$}
	\DisplayProof
\]
}

\subsubsection{Basic proof rules}
The basic proof rules are given in
Figure~\ref{fig:app-basic-rules}.
\begin{figure*}
\[
\Gamma \vdash x_1 \leftarrow e_1 \sim^{\mathbf{\Delta}}_{1_A,0} x_2 \leftarrow e_2 \colon \Phi\{e_1\lrangle{1}, e_2\lrangle{2} /x_1\lrangle{1},x_2\lrangle{2}\} \implies \Phi \quad \text{[assn]}
\]
\[	\AxiomC{$
		\Gamma \vdash c_1 \sim^{\mathbf{\Delta}}_{\alpha,\delta} c_1' \colon \Phi \implies \Phi'
		\quad
		\Gamma \vdash c_2 \sim^{\mathbf{\Delta}}_{\beta,\gamma} c_2' \colon \Phi' \implies \Psi
		$}
	\RightLabel{{[seq]}}
	\UnaryInfC{$\Gamma \vdash c_1;c_2 \sim^{\mathbf{\Delta}}_{\alpha\beta,\delta+\gamma} c_1';c_2' \colon \Phi \implies \Psi$}
	\DisplayProof
\]
\[
\Gamma \vdash \mathtt{skip} \sim^{\mathbf{\Delta}}_{1_A,0} \mathtt{skip} \colon \Phi \implies \Phi \quad \text{[skip]}
\]
\[
\AxiomC{$
	\begin{array}{l@{}}
	\Gamma \vdash^I \Phi \implies b\lrangle{1}  = b' \lrangle{2} \quad\\
	\Gamma \vdash c_1 \sim^{\mathbf{\Delta}}_{\alpha,\delta} c_1' \colon \Phi\wedge b\lrangle{1} \implies \Psi \quad
	\Gamma \vdash c_2 \sim^{\mathbf{\Delta}}_{\alpha,\delta} c_2' \colon \Phi\wedge \neg b\lrangle{1} \implies \Psi
\end{array}
$}
\RightLabel{{[cond]}}
\UnaryInfC{$
	\vdash \mathtt{if}~b~\mathtt{then}~c_1~\mathtt{else}~c_2 \sim^{\mathbf{\Delta}}_{\alpha,\delta} \mathtt{if}~b'~\mathtt{then}~c_1'~\mathtt{else}~c_2' \colon \Phi\implies \Psi
$}
\DisplayProof
\]
\[
\AxiomC{$
	\begin{array}{l@{}}
		{\Gamma \vdash^t e \colon \mathtt{int}}\quad 
		{\Gamma \vdash^I \Theta \implies \Theta \wedge (b_1\lrangle{1}  = b_2\lrangle{2})}\quad
		{\Gamma \vdash^I \Theta \wedge (e \lrangle{1} \geq n )\implies \Theta \wedge \neg b_1\lrangle{1}}\\
		\forall{0 \leq k \leq n-1}.~{\Gamma \vdash c_1 \sim^{\mathbf{\Delta}}_{\alpha_k,\delta_k} c_2 \colon \Theta \wedge (e\lrangle{1} = k) \wedge (e \lrangle{1} \leq n) \implies \Theta \wedge (e\lrangle{1} > k)}\\
	\end{array}
$}
\RightLabel{{[while]}}
\UnaryInfC{$
	\begin{array}{l@{}}
		\Gamma \vdash \mathtt{while}~{b_1}\mathtt{do}~{c_1}
    \sim^{\mathbf{\Delta}}_{\prod_{k=0}^{n-1} \alpha_k, \sum_{k=0}^{n-1}
    \delta_k} \mathtt{while}~{b_2}\mathtt{do}~{c_2} \colon~\Theta\wedge
    b_1\lrangle{1}
    \wedge (e\lrangle{1} \geq 0) \implies \Theta \wedge \neg b_1\lrangle{1}
	\end{array}
$}
	\DisplayProof
\]
\[
	\AxiomC{$
		\begin{array}{l@{}}
			\Gamma\vdash c_1 \sim^{\mathbf{\Delta}}_{\alpha, \delta} c_2 \colon \Phi_1 \implies \Psi
			\quad
			\Gamma\vdash c_1 \sim^{\mathbf{\Delta}}_{\alpha, \delta} c_2 \colon \Phi_2 \implies \Psi
		\end{array}
	$}
	\RightLabel{{[case]}}
	\UnaryInfC{$
		\Gamma\vdash c_1 \sim^{\mathbf{\Delta}}_{\alpha, \delta} c_2 \colon \Phi_1 \vee \Phi_2 \implies \Psi
	$}
	\DisplayProof
\]
\[
	\AxiomC{$
			~~{\Gamma\vdash^I \Phi' \implies \Phi}~~{\Gamma\vdash^I \Psi \implies \Psi'}~~
			{\Gamma\vdash c_1 \sim^{\mathbf{\Delta}}_{\alpha, \delta} c_2 \colon \Phi \implies \Psi}~~{\alpha \leq \beta}~~{\delta \leq \gamma}
	$}
	\RightLabel{{[weak]}}
	\UnaryInfC{$
		\Gamma \vdash c_1 \sim^{\mathbf{\Delta}}_{\beta, \gamma} c_2 \colon \Phi' \implies \Psi'
	$}
	\DisplayProof
\]
\caption{Basic rules.}
\label{fig:app-basic-rules}
\end{figure*}

\subsection{mechanism rules}
Figure \ref{figure:mechanism:rules} is the list of mechanism rules in span-apRHL.
\begin{figure*}
\begin{align*}
\lefteqn{\Gamma \vdash x_1 \xleftarrow{\$} \mathtt{Bern}(e_1) \sim^{\mathtt{DP}}_{\log\max(p,1-p)-\log\min(p,1-p),0} x_2 \xleftarrow{\$} \mathtt{Bern}(e_2) \colon}\\
&\qquad\qquad\qquad\qquad\qquad((e_1\lrangle{1} = p)\wedge(1 - e_1\lrangle{1} = e_2\lrangle{2}) \implies (x_1\lrangle{1} = x_2\lrangle{2}) &\text{[DP-Bern]}\\
\lefteqn{\Gamma \vdash x_1 \xleftarrow{\$} \mathtt{Bern}(e_1) \sim^{\mathtt{DP}}_{0,0} x_2 \xleftarrow{\$} \mathtt{Bern}(e_2) \colon}\\
&\qquad\qquad\qquad\qquad\qquad(e_1\lrangle{1} = e_2\lrangle{2}) \implies (x_1\lrangle{1} = x_2\lrangle{2}) &\text{[DP-Bern-Eq]}\\
\lefteqn{\Gamma \vdash x_1 \xleftarrow{\$} \mathtt{Bern}(e_1) \sim^{\alpha-\mathtt{RDP}}_{\frac{1}{\alpha - 1}((1-p)^{1-\alpha}p^\alpha + p^{1-\alpha}(1-p)^\alpha)} x_2 \xleftarrow{\$} \mathtt{Bern}(e_2) \colon}\\
&\qquad\qquad\qquad\qquad\qquad(e_1\lrangle{1} = p)\wedge(1 - e_1\lrangle{1} = e_2\lrangle{2}) \implies (x_1\lrangle{1} = x_2\lrangle{2}) &\text{[RDP-Bern]}\\
\lefteqn{\Gamma \vdash x_1 \xleftarrow{\$} \mathtt{Bern}(e_1) \sim^{\alpha-\mathtt{RDP}}_{0} x_2 \xleftarrow{\$} \mathtt{Bern}(e_2) \colon}\\
&\qquad\qquad\qquad\qquad\qquad(e_1\lrangle{1} = e_2\lrangle{2}) \implies (x_1\lrangle{1} = x_2\lrangle{2}) &\text{[RDP-Bern-Eq]}\\
\lefteqn{\Gamma \vdash x_1 \xleftarrow{\$} \mathtt{Bern}(e_1) \sim^{\mathtt{zCDP}}_{\log\max(p,1-p)-\log\min(p,1-p),0} x_2 \xleftarrow{\$} \mathtt{Bern}(e_2) \colon}\\
&\qquad\qquad\qquad\qquad\qquad(e_1\lrangle{1} = p)\wedge(1 - e_1\lrangle{1} = e_2\lrangle{2}) \implies (x_1\lrangle{1} = x_2\lrangle{2}) &\text{[zCDP-Bern]}\\
\lefteqn{\Gamma \vdash x_1 \xleftarrow{\$} \mathtt{Bern}(e_1) \sim^{\mathtt{zCDP}}_{0,0} x_2 \xleftarrow{\$} \mathtt{Bern}(e_2) \colon}\\
&\qquad\qquad\qquad\qquad\qquad(e_1\lrangle{1} = e_2\lrangle{2}) \implies (x_1\lrangle{1} = x_2\lrangle{2}) &\text{[zCDP-Bern-Eq]}\\
\lefteqn{\Gamma \vdash x_1 \xleftarrow{\$} \mathtt{Lap}(e_1,\lambda) \sim^{\mathtt{DP}}_{r/\lambda,0} x_2 \xleftarrow{\$} \mathtt{Lap}(e_2,\lambda) \colon}\\
& \qquad\qquad\qquad\qquad\qquad (|e_1\lrangle{1} - e_2\lrangle{2}| \leq r) \implies (x_1\lrangle{1} = x_2\lrangle{2})& \text{[DP-Lap]}\\
\lefteqn{\Gamma \vdash x_1 \xleftarrow{\$} \mathtt{Lap}(e_1,\lambda) \sim^{\alpha-\mathtt{RDP}}_{\frac{1}{\alpha - 1}\log\{\frac{\alpha}{2\alpha - 1}e^{(\alpha-1)/\lambda} + \frac{\alpha - 1}{2\alpha - 1}e^{-\alpha/\lambda}\}} x_2 \xleftarrow{\$} \mathtt{Lap}(e_2,\lambda) \colon}\\
& \qquad\qquad\qquad\qquad\qquad (|e_1\lrangle{1} - e_2\lrangle{2}| \leq 1) \implies (x_1\lrangle{1} = x_2\lrangle{2})& \text{[RDP-Lap]}\\
\lefteqn{\Gamma \vdash x_1 \xleftarrow{\$} \mathtt{Lap}(e_1,\lambda) \sim^{\mathtt{zCDP}}_{r/\lambda,0} x_2 \xleftarrow{\$} \mathtt{Lap}(e_2,\lambda) \colon}\\
& \qquad\qquad\qquad\qquad\qquad (|e_1\lrangle{1} - e_2\lrangle{2}| \leq r) \implies (x_1\lrangle{1} = x_2\lrangle{2})& \text{[zCDP-Lap]}\\
\lefteqn{\Gamma \vdash x_1 \xleftarrow{\$} \mathtt{Gauss}(e_1,\sigma^2) \sim^{\alpha-\mathtt{RDP}}_{\alpha r^2/2\sigma^2} x_2 \xleftarrow{\$} \mathtt{Gauss}(e_2,\sigma^2)}\\
&\qquad\qquad\qquad\qquad\qquad (|e_1\lrangle{1} - e_2\lrangle{2}| \leq r) \implies (x_1\lrangle{1} = x_2\lrangle{2}) & \text{[RDP-G]}\\
\lefteqn{\Gamma \vdash x_1 \xleftarrow{\$} \mathtt{Gauss}(e_1,\sigma^2) \sim^{\mathtt{zCDP}}_{0, r^2/2\sigma^2} x_2 \xleftarrow{\$} \mathtt{Gauss}(e_2,\sigma^2)}\\
&\qquad\qquad\qquad\qquad\qquad (|e_1\lrangle{1} - e_2\lrangle{2}| \leq r) \implies (x_1\lrangle{1} = x_2\lrangle{2}) & \text{[zCDP-G]}\\
\lefteqn{\Gamma \vdash x_1 \xleftarrow{\$} \mathtt{Gauss}(e_1,\sigma^2) \sim^{\mathtt{tCDP}}_{0, r^2/2\sigma^2} x_2 \xleftarrow{\$} \mathtt{Gauss}(e_2,\sigma^2)}\\
&\qquad\qquad\qquad\qquad\qquad (|e_1\lrangle{1} - e_2\lrangle{2}| \leq r) \implies (x_1\lrangle{1} = x_2\lrangle{2}) & \text{[tCDP-G]}\\
&\AxiomC{
	$
		\begin{array}{l@{}}
		\exists{c >\frac{1+\sqrt{3}}{2}}.~
(2\log(0.66/\delta) \leq c^2)\wedge (\frac{c r}{\varepsilon} \leq \sigma)
		\end{array}
	$}
	\UnaryInfC{$
	\begin{array}{l@{}}
		\Gamma \vdash x_1 \xleftarrow{\$} \mathtt{Gauss}(e_1,\sigma^2) \sim^{\mathtt{DP}}_{\varepsilon,~\delta} x_2 \xleftarrow{\$} \mathtt{Gauss}(e_2,\sigma^2) \colon\\
		\qquad\qquad\qquad\qquad\qquad\qquad
		 (|e_1\lrangle{1} - e_2\lrangle{2}| \leq r) \implies (x_1\lrangle{1} = x_2\lrangle{2})
	\end{array}
	$}
	\DisplayProof
 & \text{[DP-G]}\\
	&\AxiomC{$
	1 < 1/\sqrt{\rho} \leq A/\delta
	$}
	\UnaryInfC{$
	\begin{array}{l@{}}
		\Gamma \vdash x_1 \xleftarrow{\$} e_1 + A \cdot\mathtt{arsinh}\left(\frac{1}{A}\mathtt{Gauss}(0,\delta^2/2\rho)\right)\\ \qquad\qquad\qquad\qquad\sim^{\mathtt{tCDP}}_{16\rho,A/8\delta} x_2 \xleftarrow{\$} e_2 + A \mathtt{arsinh}\left(\frac{1}{A}\mathtt{Gauss}(0,\delta^2/2\rho)\right) \colon\\
		\qquad\qquad\qquad\qquad\qquad\qquad
		 (|e_1\lrangle{1} - e_2\lrangle{2}| \leq \delta) \implies (x_1\lrangle{1} = x_2\lrangle{2})
	\end{array}
	$}
	\DisplayProof
 & \text{[tCDP-SinhG]}
\end{align*}
\caption{Rules for basic mechanisms for DP, RDP, zCDP, and tCDP in span-apRHL.}
\label{figure:mechanism:rules}
\end{figure*}
\subsection{Denotational Semantics of pWHILE}
To prove the soundness of span-apRHL we  interpret pWHILE in $\Meas$
using the sub-Giry monad $\mathcal{G}$. First, we interpret
the value types $\mathtt{bool}$, $\mathtt{int}$, and $\mathtt{real}$ as the
finite discrete space $\mathbb{B} = 1 + 1 = \{\mathtt{true},\mathtt{false}\}$,
the countable discrete space $\mathbb{Z} = \{0,1,\ldots \}$, and the Lebesgue
measurable space $\mathbb{R}$ respectively.  We interpret $\tau^d$ as the
product $\interpret{\tau}^d$ and we interpret a typing context $\Gamma
= \{ x_1\colon\tau_1,x_2\colon\tau_2,\ldots,x_n\colon\tau_n \}$ as a product
$\interpret{\tau_1} \times \interpret{\tau_2} \times \cdots \times
\interpret{\tau_n}$.

To give a semantics to expressions, distribution expressions, and commands, we
interpret their associated typing/well-formedness judgments in a context
$\Gamma$.  We interpret an expression judgment $\Gamma \vdash^t e \colon \tau$
as a measurable function $\interpret{\Gamma \vdash^t e \colon \tau} \colon
\interpret{\Gamma} \to \interpret{\tau}$; for instance, the variable case
$\Gamma \vdash^t x\colon \tau$ is interpreted as the projection $\pi_x \colon
\interpret{\Gamma} \to \interpret{\tau}$.

We interpret a reference $\interpret{\Gamma \vdash^t e_1[e_2]\colon
\tau}$ of an element by $\mathrm{ref}\lrangle{\tau,n}
(\interpret{\Gamma \vdash^t e_1},\interpret{\Gamma \vdash^t e_2})$ where
$\mathrm{ref}\lrangle{\tau,n} \colon \interpret{\tau}^n \times \mathbb{Z} \to
\interpret{\tau}$ is defined by
$\mathrm{ref}\lrangle{\tau,n}((x_0,\ldots,x_{n-1}),k) =
x_{\min(\max(k,0),n)}$.\footnote{%
We can describe it categorically by using products and coproducts in $\Meas$.}

All operators ${\oplus}$ and comparisons ${\bowtie}$ are interpreted
as measurable functions ${\oplus} \colon \interpret{\tau} \times
\interpret{\tau} \to \interpret{\tau}$ and ${\bowtie} \colon \interpret{\tau}
\times \interpret{\tau} \to \interpret{\mathtt{bool}}$ respectively.
Likewise, we interpret a distribution expression judgment $\Gamma \vdash^p \nu
\colon \tau$ as a measurable function $\interpret{\Gamma \vdash^p \nu \colon
\tau} \colon \interpret{\Gamma} \to \mathcal{G}\interpret{\tau}$ as follows:
\begin{align*}
\interpret{\Gamma \vdash^p \mathtt{Dirac}(e) \colon \tau}
& = \eta_{\interpret{\tau}} \circ \interpret{\Gamma \vdash^t e \colon \tau},\\
\interpret{\Gamma \vdash^p \mathtt{Bern}(e) \colon \mathtt{bool}}
& = \mathtt{Bern}(\interpret{\Gamma \vdash^t e \colon \mathtt{real}}),\\
\interpret{\Gamma \vdash^p \mathtt{Lap}(e_1,e_2) \colon \mathtt{real}}
& = \mathtt{Lap}(\interpret{\Gamma \vdash^t e_1 \colon \mathtt{real}},\interpret{\Gamma \vdash^t e_2 \colon \mathtt{real}}),\\
\interpret{\Gamma \vdash^p \mathtt{Gauss}(e_1,e_2) \colon \mathtt{real}}
& = \mathcal{N}(\interpret{\Gamma \vdash^t e_1 \colon \mathtt{real}},\interpret{\Gamma \vdash^t e_2 \colon \mathtt{real}}).
\end{align*}
Finally, we interpret a command judgment $\Gamma \vdash c$ inductively as a measurable
function $\interpret{\Gamma \vdash c } \colon \interpret{\Gamma} \to
\mathcal{G}\interpret{\Gamma}$ by
\begin{align*}
\interpret{\Gamma\vdash x \xleftarrow{\$} \nu}
	&= \mathcal{G}(\mathrm{rw}\lrangle{\Gamma \mid x\colon\tau}) \circ\mathrm{st}_{\interpret{\Gamma},\interpret{\tau}} \circ\lrangle{\mathrm{id}_{\interpret{\Gamma}},\interpret{\nu}},\\
\interpret{\Gamma\vdash c_1 ; c_2}
	&= {\interpret{\Gamma\vdash c_2}}^\sharp \circ \interpret{\Gamma\vdash c_1},\\
\interpret{\Gamma\vdash\mathtt{skip}}
	& = \eta_{\interpret{\Gamma}}\\
\interpret{\Gamma \vdash \mathtt{if}~b~\mathtt{then}~c_1~\mathtt{else}~c_2}
	&
	=\left[\interpret{\Gamma\vdash c_1},\interpret{\Gamma\vdash c_2} \right]
	\circ \mathrm{br}\lrangle{\Gamma}
	\circ\lrangle{\interpret{\Gamma\vdash b} ,\mathrm{id}_{\interpret{\Gamma}}}
\end{align*}
Here, $\mathrm{rw}\lrangle{\Gamma \mid x\colon\tau} \colon \interpret{\Gamma}
\times \interpret{x\colon\tau} \to \interpret{\Gamma}$ (${x\colon\tau} \in
\Gamma $) is an overwriting operation of memories mapping $((a_1,\ldots,a_k,\ldots,
a_n),b_k) \mapsto (a_1,\ldots,b_k,\ldots, a_n)$; this is given by the
Cartesian products in $\Meas$.  The function $\mathrm{br}\lrangle{\Gamma} \colon
2 \times \interpret{\Gamma} \to \interpret{\Gamma} + \interpret{\Gamma}$ comes
from the canonical isomorphism $2 \times \interpret{\Gamma} \cong
\interpret{\Gamma} + \interpret{\Gamma}$ from the distributivity of $\Meas$.

To interpret loops, we introduce the dummy ``abort'' command $\Gamma \vdash
\mathtt{null}$ that is interpreted by the null/zero measure $\interpret{\Gamma
\vdash \mathtt{null}} = 0$, and the following commands corresponding to the
finite unrollings of the loop:
\[
[\mathtt{while}~b~\mathtt{do}~c]_n
=
\begin{cases} 
\mathtt{if}~b~\mathtt{then}~\mathtt{null}~\mathtt{else}~\mathtt{skip}, & \text{ if } n=0
\\
\mathtt{if}~b~\mathtt{then}~c;[\mathtt{while}~b~\mathtt{do}~c]_{k}, & \text{ if } n=k+1
\end{cases}
\]
We then interpret loops as the supremum of interpretations of finite
executions:\footnote{%
  This is well-defined, since the family $\{ \interpret{\Gamma\vdash
  [\mathtt{while}~e~\mathtt{do}~c]_n} \}_{n \in \mathbb{N}}$ is an
  $\omega$-chain with respect to the $\omega\mathbf{CPO}_\bot$-enrichment
$\sqsubseteq$ of $\Meas_{\mathcal{G}}$.}
\[
\interpret{\Gamma \vdash \mathtt{while}~b~\mathtt{do}~c}
	= \sup_{ n\in\mathbb{N}} \interpret{\Gamma\vdash [\mathtt{while}~e~\mathtt{do}~c]_n}.
\]

\subsection{Proof of Soundness of the Program Logic}

\begin{lemma}
The [assn] rule is sound.
\end{lemma}
\begin{proof}
We may assume $x_1 \neq x_2$ without loss of generality. Let
\[
  ((\phi^1,a_1^1,a_2^1),(\phi^2,a_1^2,a_2^2))
  \in \Rinterpret{\Gamma \vdash^R \Phi\{e_1\lrangle{1},e_2\lrangle{2}/x_1\lrangle{1},x_2\lrangle{2}\}}
\]
where $a^i_j$ is a value of variable $x_j$ ($i = 1,2$).
%
%
Since $x_1$ and $x_2$ are not free variables in $e_1$ and $e_2$ respectively, we have 
\[
((\phi_1,\interpret{\Gamma \vdash^t e_1 \colon \tau}(\phi^1,a_1^1,a_2^1),a_2^1),(\phi^2,a^2_1,\interpret{\Gamma \vdash^t e_2 \colon \tau}(\phi^2,a_1^2,a_2^2)) \in \Rinterpret{\Gamma \vdash^R \Phi}.
\]
Therefore, 
\[
(f_1(\phi^1,a_1^1,a_2^1), f_2(\phi^2,a_1^2,a_2^2)) \in \Rinterpret{\Gamma \vdash^R \Phi}
\]
where
$f_i = \mathrm{rw}\lrangle{\Gamma \mid x_i \colon\tau} \circ \lrangle{\mathrm{id}_\interpret{\Gamma},\interpret{\Gamma \vdash^t e_i \colon \tau}}$ ($i = 1,2$).
Therefore, we obtain the following morphism of spans
(note that both $\interpret{\Gamma \vdash^R
\Phi\{e_1\lrangle{1},e_2\lrangle{2}/x_1\lrangle{1},x_2\lrangle{2}\}}$ and
$\interpret{\Gamma \vdash^R \Phi}$ are binary relation converted to spans)
,
\begin{align*}
\lefteqn{(f_1,f_2,(f_1\times f_2)|_\Rinterpret{\Gamma \vdash^R \Phi\{e_1\lrangle{1},e_2\lrangle{2}/x_1\lrangle{1},x_2\lrangle{2}\}})\colon}\\
&\qquad\qquad\qquad\qquad\interpret{\Gamma \vdash^R \Phi\{e_1\lrangle{1},e_2\lrangle{2}/x_1\lrangle{1},x_2\lrangle{2}\}}
\to \interpret{\Gamma \vdash^R \Phi}.
\end{align*}
Letting
$g_i = \eta_{\interpret{\Gamma}} \circ f_i = \interpret{\Gamma \vdash x_i \leftarrow e_i}$,
we conclude
\begin{align*}
\lefteqn{(g_1,g_2,\lrangle{\eta_\Phi,\eta_\Phi} \circ (g_1\times g_2)|_\Rinterpret{\Gamma \vdash^R \Phi\{e_1\lrangle{1},e_2\lrangle{2}/x_1\lrangle{1},x_2\lrangle{2}\}})\colon}\\
&\qquad\qquad\qquad\qquad\interpret{\Gamma \vdash^R \Phi\{e_1\lrangle{1},e_2\lrangle{2}/x_1\lrangle{1},x_2\lrangle{2}\}}\to \interpret{\Gamma \vdash^R \Phi}^{\sharp(\mathbf{\Delta},1_A,0)}.
\end{align*}
\end{proof}

\begin{lemma}
The [seq] rule is sound.
\end{lemma}
\begin{proof}
Since the judgments $\Gamma \vdash c_1\sim^{\mathbf{\Delta}}_{\alpha,\delta}
c'_1 \colon \Phi \implies \Phi'$ and $\Gamma \vdash c_2
\sim^{\mathbf{\Delta}}_{\beta,\gamma} c'_2 \colon \Phi' \implies
\Psi$ are valid,
we obtain the following two morphisms in $\Span(\Meas)$ for witness functions $l_1$ and $l_2$:
\begin{align*}
(\interpret{\Gamma \vdash c_1}, \interpret{\Gamma \vdash c_1'}, l_1) &\colon \interpret{\Gamma \vdash^R \Phi} \to \interpret{\Gamma \vdash^R \Phi'}^{\sharp(\mathbf{\Delta},\alpha,\delta)}\\
(\interpret{\Gamma \vdash c_2}, \interpret{\Gamma \vdash c_2'}, l_2) &\colon \interpret{\Gamma \vdash^R \Phi'} \to \interpret{\Gamma \vdash^R \Psi}^{\sharp(\mathbf{\Delta},\beta,\gamma)}
\end{align*}
By taking the graded Kleisli lifting of the second morphism $(\interpret{\Gamma
\vdash c_2}, \interpret{\Gamma \vdash c_2'}, l_2)$, for some witness function
$l_3$, we have a $\Span(\Meas)$-morphism
\[
(\interpret{\Gamma \vdash c_2}^\sharp, \interpret{\Gamma \vdash c_2'}^\sharp, l_3) \colon \interpret{\Gamma \vdash^R \Phi'}^{\sharp(\mathbf{\Delta},\alpha,\delta)} \to \interpret{\Gamma \vdash^R \Psi}^{\sharp(\mathbf{\Delta},\alpha\beta,\delta+\gamma)}.
\]
Composing, we have a span-morphism giving validity of $\Gamma \vdash c_1;c_2
\sim^{\mathbf{\Delta}}_{\alpha\beta,\delta+\gamma} c'_1;c'_2 \colon \Phi
\implies \Psi$ :
\[
(\interpret{\Gamma \vdash c_2}^\sharp \circ \interpret{\Gamma \vdash c_1}, \interpret{\Gamma \vdash c_2'}^\sharp \circ \interpret{\Gamma \vdash c_1'}, l_3 \circ l_1) \colon \interpret{\Gamma \vdash^R \Phi} \to \interpret{\Gamma \vdash^R \Psi}^{\sharp(\mathbf{\Delta},\alpha\beta,\delta+\gamma)}.
\]
\end{proof}

\begin{lemma}
The [weak] rule is sound
\end{lemma}
\begin{proof}
Since the judgment $\Gamma \vdash c_1\sim^{\mathbf{\Delta}}_{\alpha,\delta} c_2
\colon \Phi \implies \Psi$ is valid, we have a witness function $l \colon
\Rinterpret{\Gamma\vdash^R\Phi} \to
W(\interpret{\Gamma\vdash^R\Psi},\mathbf{\Delta},\alpha,\delta)$ such that
\[
(\interpret{\Gamma\vdash c_1},\interpret{\Gamma\vdash c_2},l) \colon \interpret{\Gamma\vdash^R\Phi} \to \interpret{\Gamma\vdash^R\Psi}^{\sharp(\mathbf{\Delta},\alpha,\delta)}
\]
From the inclusions $\Gamma \vdash^I \Phi' \implies \Phi$  and $\Gamma \vdash^I
\Psi \implies \Psi'$ of relations, we have
%
%
\begin{align*}
(\mathrm{id}_\interpret{\Gamma},\mathrm{id}_\interpret{\Gamma},(\mathrm{id}_\interpret{\Gamma}\times\mathrm{id}_\interpret{\Gamma})|_{\Rinterpret{\Gamma\vdash^R\Phi'}}) &\colon \interpret{\Gamma\vdash^R\Phi'} \to \interpret{\Gamma\vdash^R\Phi}\\
(\mathrm{id}_\interpret{\Gamma},\mathrm{id}_\interpret{\Gamma},(\mathrm{id}_\interpret{\Gamma}\times\mathrm{id}_\interpret{\Gamma})|_{\Rinterpret{\Gamma\vdash^R\Psi}}) &\colon \interpret{\Gamma\vdash^R\Psi} \to \interpret{\Gamma\vdash^R\Psi'}.
\end{align*}
Thanks to the inclusion structure of the span-lifting $({-})^{\sharp(\mathbf{\Delta})}$,
we obtain
\[
(\mathcal{G}\mathrm{id}_\interpret{\Gamma},\mathcal{G}\mathrm{id}_\interpret{\Gamma},(\mathcal{G}\mathrm{id}_\interpret{\Gamma}\times \mathcal{G}\mathrm{id}_\interpret{\Gamma})|_{W(\interpret{\Gamma\vdash^R\Phi},\mathbf{\Delta},\alpha,\delta)} ) \colon
\interpret{\Gamma\vdash^R\Psi'}^{\sharp(\mathbf{\Delta},\alpha,\delta)} \to \interpret{\Gamma\vdash^R\Psi'}^{\sharp(\mathbf{\Delta},\beta,\gamma)}
\]
Therefore, we conclude
\[
(\interpret{\Gamma\vdash c_1},\interpret{\Gamma\vdash c_2},l|_{\Rinterpret{\Gamma\vdash^R \Phi'}})
\colon \interpret{\Gamma\vdash^R\Phi'} \to \interpret{\Gamma\vdash^R\Psi'}^{\sharp(\mathbf{\Delta},\beta,\gamma)}
\]
\end{proof}

\begin{lemma}
The [cond] rule is sound.
\end{lemma}
\begin{proof}
Since the judgments
$\Gamma \vdash c_1\sim^{\mathbf{\Delta}}_{\alpha,\delta} c_2 \colon \Phi \wedge b\lrangle{1} \implies \Psi$ and
$\Gamma \vdash c_1'\sim^{\mathbf{\Delta}}_{\alpha,\delta} c_2' \colon \Phi \wedge \neg b\lrangle{1} \implies \Psi$ are valid,
we have two witness functions $l_T \colon \Rinterpret{\Gamma\vdash^R \Phi \wedge b\lrangle{1}} \to W(\interpret{\Gamma\vdash^R \Psi},\mathbf{\Delta},\alpha,\delta)$ and $l_F \colon \Rinterpret{\Gamma\vdash^R \Phi \wedge \neg b\lrangle{1}} \to W(\interpret{\Gamma\vdash^R \Psi},\mathbf{\Delta},\alpha,\delta)$ that make the following morphisms in $\Span(\Meas)$:
\begin{align*}
(\interpret{\Gamma\vdash c_1},\interpret{\Gamma\vdash c_2},l_T) &\colon \interpret{\Gamma\vdash^R \Phi \wedge b\lrangle{1}} \to  \interpret{\Gamma\vdash^R \Psi}^{\sharp(\mathbf{\Delta},\alpha,\delta)}\\
(\interpret{\Gamma\vdash c_1'},\interpret{\Gamma\vdash c_2'},l_F) &\colon \interpret{\Gamma\vdash^R \Phi \wedge \neg b\lrangle{1}} \to  \interpret{\Gamma\vdash^R \Psi}^{\sharp(\mathbf{\Delta},\alpha,\delta)}.
\end{align*}
By the coproduct structure of $\Span(\Meas)$, we have the following span-morphism:
\begin{align*}
\lefteqn{(
[\interpret{\Gamma\vdash c_1},\interpret{\Gamma\vdash c_1'}],
[\interpret{\Gamma\vdash c_2},\interpret{\Gamma\vdash c_2'}],
[l_T,l_F]
)\colon}\\
&\qquad\qquad\qquad\interpret{\Gamma\vdash^R \Phi \wedge b\lrangle{1}} \mathbin{\dot{+}} \interpret{\Gamma\vdash^R \Phi \wedge \neg b\lrangle{1}} \to  \interpret{\Gamma\vdash^R \Psi}^{\sharp(\mathbf{\Delta},\alpha,\delta)}.
\end{align*}

We write
$g_1 = \mathrm{br}\lrangle{\Gamma}\circ \lrangle{\interpret{\Gamma \vdash^t b},\mathrm{id}_\interpret{\Gamma}}$ and 
$g_2 = \mathrm{br}\lrangle{\Gamma}\circ \lrangle{\interpret{\Gamma \vdash^t \neg b},\mathrm{id}_\interpret{\Gamma}}$.
%
%
We construct the following morphism by using $\Gamma \vdash^I \Phi \implies b\lrangle{1} = b'\lrangle{2}$
%
%
\begin{equation}\label{eq:techinicalpointofcase}
(g_1,g_2,H \circ (g_1 \times g_2)|_{\Rinterpret{\Gamma\vdash^R \Phi}})\colon
\interpret{\Gamma\vdash^R \Phi} \to \interpret{\Gamma\vdash^R \Phi \wedge b\lrangle{1}} \mathbin{\dot{+}} \interpret{\Gamma\vdash^R \Phi \wedge \neg b\lrangle{1}},
\end{equation}
where $H$ is the composition $H_3 \circ H_2 \circ H_1$ of 
\begin{itemize}
\item $H_1 \colon (\interpret{\Gamma}+\interpret{\Gamma}) \times
(\interpret{\Gamma}+\interpret{\Gamma}) \cong 4 \times (\interpret{\Gamma}
\times\interpret{\Gamma})$

defined by $(\iota_i(\phi_1),\iota_j(\phi_2)) \mapsto ((i,j),(\phi_1,\phi_2))$ where $i,j \in 2$,
\item $H_2 \colon  4 \times (\interpret{\Gamma} \times\interpret{\Gamma}) \to 2 \times
(\interpret{\Gamma} \times\interpret{\Gamma})$

defined by
$((b_1,b_2),\phi^1,\phi^2) \mapsto (b_1,\phi^1,\phi^2)$,
\item $H_3 \colon  2\times (\interpret{\Gamma} \times \interpret{\Gamma}) \cong (\interpret{\Gamma}\times\interpret{\Gamma}) + (\interpret{\Gamma}\times\interpret{\Gamma})$

defined by
$(b,(\phi^1,\phi^2)) \mapsto \iota_b (\phi^1,\phi^2)$.
\end{itemize}
Here, $\iota_i$ are coprojections $\iota_1 \colon A \to A + B$ and $ \iota_2 \colon B \to A + B$.
The bijections $H_1$ and $H_3$ are given from the distributivity of products and coproducts in $\Meas$,
and $H_2$ is given by a projection.

%

Now, let $(\phi^1,\phi^2) \in \Rinterpret{\Gamma\vdash^R \Phi}$.
Since we suppose $\Gamma \vdash^I \Phi \implies b\lrangle{1} = b'\lrangle{2}$, we have
\[
(g_1(\phi^1), g_2(\phi^2))=
\begin{cases}
((1,\phi^1),(1,\phi^2)) & (\phi^1,\phi^2) \in \Rinterpret{\Gamma\vdash^R \Phi \wedge b\lrangle{1}} = \Rinterpret{\Gamma\vdash^R \Phi \wedge b'\lrangle{2}}\\
((2,\phi^1),(2,\phi^2)) & (\phi^1,\phi^2) \in \Rinterpret{\Gamma\vdash^R \Phi \wedge \neg b\lrangle{1}} = \Rinterpret{\Gamma\vdash^R \Phi \wedge \neg b'\lrangle{2}}.
\end{cases}
\]
We observe the role of $H$ in the first case ($(\phi^1,\phi^2) \in \Rinterpret{\Gamma\vdash^R \Phi \wedge b\lrangle{1}}$),
\begin{align*}
H(g_1(\phi^1), g_2(\phi^2))
&= H_3 \circ H_2 \circ H_1 ((1,\phi^1),(1,\phi^2))\\
&= H_3 \circ H_2 ((1,1), (\phi^1, \phi^2))
= H_3 (1,(\phi^1, \phi^2))
= \iota_1 (\phi^1, \phi^2).
\end{align*}
In the same way, we have $H(g_1(\phi^1), g_2(\phi^2)) = \iota_2 (\phi^1, \phi^2)$ in the second case.
Therefore, the measurable function $H \circ (g_1 \times
g_2)|_{\Rinterpret{\Gamma\vdash^R \Phi}}$ forms a function from $\Rinterpret{\Gamma\vdash^R \Phi}$ to
$\Rinterpret{\Gamma\vdash^R
\Phi \wedge b\lrangle{1}} + \Rinterpret{\Gamma\vdash^R \Phi
\wedge \neg b\lrangle{1}}$ satisfying (\ref{eq:techinicalpointofcase}).

Since $\interpret{\Gamma \vdash \mathtt{if}~b~\mathtt{then}~c~\mathtt{else}~c'} = [\interpret{\Gamma\vdash c},\interpret{\Gamma\vdash c'}]\circ \mathrm{br}\lrangle{\Gamma}\circ \lrangle{\interpret{\Gamma \vdash^t b},\mathrm{id}_\interpret{\Gamma}}$, we conclude the soundness.
\end{proof}

\begin{remark}
  Similarly, we have soundness of [case].
\end{remark}

\begin{remark}
The soundness of the [while] rule is a consequence of the soundness of [seq], [weak], and the [case] rule since the [while] in our logic deal only with finite-loops.
\end{remark}

\begin{lemma}
The rule [RDP-G] is sound.
\end{lemma}
\begin{proof}
We assume $x_1 \neq x_2$.
First, it can be directly checked that the function 
$f = {\mathcal{N}({-},\sigma^2)} \colon \mathbb{R} \to \mathcal{G}\mathbb{R}$
is measurable.  From \citet[Proposition 3]{Mironov17}, the function $f$ satisfies
$D^\alpha(f(x)||f(y)) \leq \alpha r^2/2\sigma^2$ whenever $|x - y| \leq r$.
Hence, $(f,f,(f \times f)|_\Phi)$ is a span-morphism
$\Phi \to \mathrm{Eq}_{\mathbb{R}}^{\sharp(D^\alpha,\ast,\alpha r^2/2\sigma^2)}$ 
where $\Phi = \Set{(x,y) \in \mathbb{R} \times \mathbb{R}|{|x - y|} \leq r}$ is regarded as a span.

We next construct a span-morphism
$(h_1,h_2, (h_1 \times h_2)|_\Theta)$ mapping $\Theta \to \mathrm{Eq}_{\mathbb{R}}^{\sharp(D^\alpha,\ast,\alpha r^2/2\sigma^2)}$
where $\Theta = \interpret{\Gamma \vdash^R |e_1\lrangle{1} - e_2\lrangle{2}| \leq r}$ and 
$h_i = \interpret{\Gamma \vdash^p \mathtt{Gauss}(e_i,\sigma^2) \colon \mathtt{real}}$ ($i = 1, 2$).
We write
$g_i = \interpret{\Gamma \vdash^t e_i \colon \mathtt{real}}$ ($i = 1, 2$).
It is clear that $(g_1,g_2,(g_1 \times g_2)|_\Theta)$ is a span-morphism $\Theta \to \Phi$.
Since $h_i = f_i \circ g_i$ ($i = 1,2$), the triple $(h_1, h_2, (h_1 \times
h_2)|_\Theta)$ is a span-morphism $\Theta \to
\mathrm{Eq}_{\mathbb{R}}^{\sharp(D^\alpha,\ast,\alpha r^2/2\sigma^2)}$.

Now, the triple $(\mathrm{id}_\interpret{\Gamma} \times h_1,
\mathrm{id}_\interpret{\Gamma} \times h_2, (\mathrm{id}_\interpret{\Gamma}
\times \mathrm{id}_\interpret{\Gamma}) \times (h_1 \times h_2)|_\Theta)$ is a
morphism of spans $\top_\interpret{\Gamma} \mathbin{\dot\times} \Theta \to
\top_\interpret{\Gamma} \mathbin{\dot\times}
(\mathrm{Eq}_{\mathbb{R}})^{\sharp(D^\alpha,\ast,\alpha r^2/2\sigma^2)}$ where
$\top_\interpret{\Gamma} =
(\interpret{\Gamma},\interpret{\Gamma},\interpret{\Gamma}\times\interpret{\Gamma},\pi_1,\pi_2)$.
Thanks to the unit and the double strength of the span-lifting
$\{(-)^{\sharp(D^\alpha,\ast,\rho)}\}_\rho$, the triple 
$(\mathrm{st}_{\interpret{\Gamma},\mathbb{R}}, \mathrm{st}_{\interpret{\Gamma},\mathbb{R}},
\langle \mathrm{st}_{\interpret{\Gamma},\mathbb{R}} \circ (\pi_1 \times \pi_1),
\mathrm{st}_{\interpret{\Gamma},\mathbb{R}} \circ (\pi_2 \times \pi_2) \rangle
|_{(\interpret{\Gamma} \times \interpret{\Gamma})\times W(\mathrm{Eq}_{\mathbb{R}},D^\alpha,\ast,\alpha r^2/2\sigma^2)})$
is a morphism of spans $\top_\interpret{\Gamma} \mathbin{\dot\times}
(\mathrm{Eq}_{\mathbb{R}})^{\sharp(D^\alpha,\ast,\alpha r^2/2\sigma^2)} \to
(\top_\interpret{\Gamma} \mathbin{\dot\times}
\mathrm{Eq}_{\mathbb{R}})^{\sharp(D^\alpha,\ast,\alpha r^2/2\sigma^2)}$.
%
%

We write $k_i = \mathrm{rw}\lrangle{\Gamma \mid x_i\colon\mathtt{real}}$ ($i = 1,2$).
For any $(((\phi^1,a_1^1,a_2^1),r),((\phi^2,a_1^2,a_2^2),r)) \in
\top_\interpret{\Gamma} \mathbin{\dot\times} \mathrm{Eq}_{\mathbb{R}}$ where
$a^i_j$ is a value of variable $x_j$ ($i = 1,2$), 
we have 
\begin{align*}
&(\mathrm{rw}\lrangle{\Gamma \mid x_1\colon\mathtt{real}}((\phi^1,a_1^1,a_2^1),r^1),
\mathrm{rw}\lrangle{\Gamma \mid x_2\colon\mathtt{real}}((\phi^2,a_1^2,a_2^2),r^2)
)\\
&=
((\phi^1,r,a_2^1),
(\phi^2,a^2_1,r))
\in
\Rinterpret{\Gamma \vdash^R x_1\lrangle{1} = x_2\lrangle{2}}
\end{align*}
Hence, the triple $(k_1,k_2,(k_1 \times k_2)|_{(\interpret{\Gamma} \times
\interpret{\Gamma}) \times \mathrm{Eq}_{\mathbb{R}}})$ forms a morphism of spans
$(\top_\interpret{\Gamma} \mathbin{\dot\times} \mathrm{Eq}_{\mathbb{R}}) \to
\interpret{\Gamma \vdash^R x_1\lrangle{1} = x_2\lrangle{2}}$. (Note that
$(\top_\interpret{\Gamma} \mathbin{\dot\times} \mathrm{Eq}_{\mathbb{R}})$ and
$\interpret{\Gamma \vdash^R x_1\lrangle{1} = x_2\lrangle{2}}$ are binary relations
converted to spans.)

By the functoriality of the span-lifting $\{(-)^{\sharp(D^\alpha,\ast,\rho)}\}_\rho$, we obtain in $\Span(\Meas)$,
\begin{align*}
\lefteqn{(k_1,k_2,(k_1 \times k_2)|_{(\interpret{\Gamma} \times \interpret{\Gamma}) \times \mathrm{Eq}_{\mathbb{R}}})^{\sharp(D^\alpha,\ast,\alpha r^2/2\sigma^2)}\colon}\\
&\qquad \qquad \qquad (\top_\interpret{\Gamma} \mathbin{\dot\times} \mathrm{Eq}_{\mathbb{R}})^{\sharp(D^\alpha,\ast,\alpha r^2/2\sigma^2)}
\to \interpret{\Gamma \vdash^R x_1\lrangle{1} = x_2\lrangle{2}}^{\sharp(D^\alpha,\ast,\alpha r^2/2\sigma^2)}.
\end{align*}

Since $\interpret{\Gamma \vdash x_i \xleftarrow{\$} \mathtt{Gauss}(e_i,\sigma^2)} = \mathcal{G}k_i \circ \mathrm{st}_{\interpret{\Gamma},\mathbb{R}} \circ \lrangle{\mathrm{id}_\interpret{\Gamma}, h_i}$ ($i = 1,2$),
we conclude the soundness of [RDP-G]:
\begin{align*}
\lefteqn{(\interpret{\Gamma \vdash x_1 \xleftarrow{\$} \mathtt{Gauss}(e_1,\sigma^2)},\interpret{\Gamma \vdash x_2 \xleftarrow{\$} \mathtt{Gauss}(e_2,\sigma^2)},l)}\\
&= 
(k_1,k_2,(k_1 \times k_2)|_{(\interpret{\Gamma} \times \interpret{\Gamma}) \times \mathrm{Eq}_{\mathbb{R}}})^{\sharp(D^\alpha,\ast,\alpha r^2/2\sigma^2)}\\
& \qquad \circ
(\mathrm{st}_{\interpret{\Gamma},\mathbb{R}}, \mathrm{st}_{\interpret{\Gamma},\mathbb{R}},
\langle
\mathrm{st}_{\interpret{\Gamma},\mathbb{R}} \circ (\pi_1 \times \pi_1),
\mathrm{st}_{\interpret{\Gamma},\mathbb{R}} \circ (\pi_2 \times \pi_2)
\rangle
|_{(\interpret{\Gamma} \times \interpret{\Gamma})\times W(\mathrm{Eq}_{\mathbb{R}},D^\alpha,\ast,\alpha r^2/2\sigma^2)})\\
& \qquad \circ
(\mathrm{id}_\interpret{\Gamma} \times h_1, \mathrm{id}_\interpret{\Gamma} \times h_2, (\mathrm{id}_\interpret{\Gamma} \times \mathrm{id}_\interpret{\Gamma}) \times (h_1 \times h_2)|_\Theta)\\
& \qquad \circ
(\lrangle{\mathrm{id}_\interpret{\Gamma},\mathrm{id}_\interpret{\Gamma}},\lrangle{\mathrm{id}_\interpret{\Gamma},\mathrm{id}_\interpret{\Gamma}},\lrangle{\mathrm{id}_{\interpret{\Gamma}\times\interpret{\Gamma}}|_{\Theta},\mathrm{id}_\Theta})\colon\\
&\Theta
\to 
\interpret{\Gamma \vdash^R x_1\lrangle{1} = x_2\lrangle{2}}^{\sharp(D^\alpha,\ast,\alpha r^2/2\sigma^2)}.
\end{align*}
\end{proof}

Soundness of other mechanism rules follows similarly using \citet[Propositions
5, 6, 7]{Mironov17}, \citet[Proposition 1]{DworkMcSherryNissimSmith2006},
\citet[Lemma 4.2]{Sato2016MFPS} (a refinement of \citet[Theorem
3.22]{DworkRothTCS-042}), the soundness of the transitivity rules are proved by
\citet[Lemma 4.2(iii)]{olmedo2014approximate}, \citet[Proposition 27]{BunS16} and
Lemma \ref{lem:Renyi:weak-triangle}, and the soundness of the conversion rules
follows by \citet[Proposition 4]{BunS16}, \citet[Proposition 3]{Mironov17}, and
\citet[Lemmas 3.2, 3.5]{BunS16}.
\fi

\section{Omitted proofs}
\begin{theorem}[Theorem \ref{thm:div:cont+comp->mon}]
  An $A$-graded family $\mathbf{\Delta}$ is additive if it is continuous and composable.
\end{theorem}
\begin{proof}
From the continuity of $\mathbf{\Delta}$,
\[
\Delta^{\alpha\beta}_{X\times Y}(\mu_1 \otimes \mu_3,\mu_2 \otimes \mu_4)
= \sup\Set{\Delta^{\alpha\beta}_{I}(\mathcal{G}k(\mu_1 \otimes \mu_3),\mathcal{G}k(\mu_2 \otimes \mu_4))| k \colon X\times Y \to I}.
\]

We fix $k \colon X \times Y \to I$.
For any $\mu \in \mathcal{G}Y$, we define $K_{\mu} \colon X \to \mathcal{G}I$ by 
$K_\mu = \mathcal{G}k \circ \mathrm{st}_{X,Y} \circ (\mathrm{id}_X \times \overline {\mu}) \circ \inverse{\rho}_X$
where $\overline{\mu}$ is the generalized element $1 \to \mathcal{G}Y$ assigning $\mu$, and $\rho_X$ is a canonical isomorphism $X \cong X \times 1$.
We then obtain for any $\mu \in \mathcal{G}X$, 
\begin{align*}
K_{\mu}^\sharp(\mu')
& = \mathcal{G}k \circ \mu_{X \times Y} \circ \mathcal{G}\mathrm{st}_{X,Y} \circ \mathcal{G} (\mathrm{id}_X \times \overline {\mu}) \circ \mathcal{G} \inverse{\rho}_X (\mu')\\
& = \mathcal{G}k \circ \mu_{X \times Y} \circ \mathcal{G}\mathrm{st}_{X,Y} \circ \mathcal{G} (\mathrm{id}_X \times \overline {\mu}) \circ \mathrm{st}'_{X,1} \circ \inverse{\rho}_{\mathcal{G}X} (\mu')\\
& = \mathcal{G}k \circ \mu_{X \times Y} \circ \mathcal{G}\mathrm{st}_{X,Y}\circ \mathrm{st}'_{X,\mathcal{G}Y}  \circ (\mathrm{id}_{\mathcal{G}X} \times \overline {\mu}) \circ \inverse{\rho}_{\mathcal{G}X} (\mu')\\
& = \mathcal{G}k \circ \mathrm{dst}_{X,Y} (\mu',\mu) = \mathcal{G}k(\mu' \otimes \mu).
\end{align*}
We also obtain $K_\mu(x) = \mathcal{G}k(\mathbf{d}_{x} \otimes \mu)$ for any $x\in X$.
This implies $K_\mu(x) = \mathcal{G}k(x,-)(\mu)$ where $k(x,-) Y \to I$ is measurable because $(\mathbf{d}_{x} \otimes \mu)(\inverse{k}(A)) = \mu((\inverse{k}(A))|_x) = \mu(\inverse{k(x,-)}(A))$ for any $A \subseteq I$.
From the composability and continuity of $\mathbf{\Delta}$, we have
\begin{align*}
\Delta^{\alpha\beta}_{I}(\mathcal{G}k(\mu_1 \otimes \mu_3),\mathcal{G}k(\mu_2 \otimes \mu_4))
&= \Delta^{\alpha\beta}_{I}(K_{\mu_3}^\sharp(\mu_1),K_{\mu_4}^\sharp(\mu_2))\\
&\leq \Delta^{\alpha}_{X}(\mu_1,\mu_2) + \sup_{x \in X}\Delta^{\beta}_{I}(K_{\mu_3}(x),K_{\mu_4}(x))\\
&= \Delta^{\alpha}_{X}(\mu_1,\mu_2) + \sup_{x \in X}\Delta^{\beta}_{I}( \mathcal{G}k(x,-)(\mu_3),\mathcal{G}k(x,-)(\mu_4))\\
&\leq \Delta^{\alpha}_{X}(\mu_1,\mu_2) +\Delta^{\beta}_{Y}(\mu_3,\mu_4).
\end{align*}

Since $k \colon X \times Y \to I$ is arbitrary, we conclude the additivity of $\mathbf{\Delta}$.
\end{proof}

\begin{theorem}[Theorem \ref{thm:div:fin-comp->comp}]
A continuous approximable $A$-graded family $\mathbf{\Delta}$ is composable if
finite-composable.
\end{theorem}
\begin{proof}
Let $\mu_1,\mu_2 \in \mathcal{G}X$ and $f,g \colon X \to \mathcal{G}Y$.
Since $\mathbf{\Delta}$ is continuous, approximable, and finite composable, we obtain,
\begin{align*}
\lefteqn{\Delta^{\alpha\beta}_{Y}(f^\sharp (\mu_1), g^\sharp (\mu_2))}\\
&\leq\sup\Set{
			\Delta^{\alpha\beta}_I (\mathcal{G}k(f^\sharp(\mu_1)),\mathcal{G}k(f^\sharp(\mu_2)))
		\mid 
			I \in \FinSet, k \colon X \to I
		}
	\\
&\leq\sup\Set{
			\lim_{n \to \infty}\Delta^{\alpha\beta}_I
			(
				(\mathcal{G}k \circ f \circ m_n \circ m^\ast_n)^\sharp(\mu_1),
				(\mathcal{G}k \circ g \circ m_n \circ m^\ast_n)^\sharp(\mu_2)
			)
		\mid 
			I \in \FinSet, k \colon X \to I.
		}
	\\
& \leq 
	\sup\Set{
			\lim_{n \to \infty}\Delta^{\alpha}_{J_n}
			(
				\mathcal{G} m^\ast_n (\mu_1),
				\mathcal{G} m^\ast_n (\mu_2)
			)
		\mid 
			I \in \FinSet, k \colon X \to I
		}
	\\
& \quad +
	\sup\Set{
			\lim_{n \to \infty}\sup_{j \in J_n}\Delta^\beta_{I}
			(
				\mathcal{G}k \circ f \circ m_n(j),
				\mathcal{G}k \circ g \circ m_n(j)
			)
		\mid 
			I \in \FinSet, k \colon X \to I
		}
	\\
\end{align*}
Regarding the first term of the last inequality, 
since $m^\ast_n \colon X \to J_n$ where $J_n \in \FinSet$, and $\Delta^\alpha$ is continuous, we have
\[
\Delta^{\alpha}_{J_n}(\mathcal{G} m^\ast_n (\mu_1),\mathcal{G} m^\ast_n (\mu_2))
\leq \Delta^{\alpha}_X(\mu_1,\mu_2)
.\]
Concerning the second term, since $m_n(j) \in X$ for any $n$ and $j \in J_n$, and $k \colon I \to X$ and $\Delta^\beta$ is continuous, we obtain 
\[
\sup_{j \in J_n}\Delta^\beta_{I}
(
\mathcal{G}k \circ f \circ m_n(j),
\mathcal{G}k \circ g \circ m_n(j)
)
\leq \sup_{x \in X}\Delta^\beta_{I}
(
\mathcal{G}k \circ f(x),
\mathcal{G}k \circ g(x)
)
\leq 
\sup_{x \in X}\Delta^\beta_{Y}(f(x),g(x))
.\]
This completes the proof.
\end{proof}

\begin{theorem}[Theorem \ref{f-divergence:approximability}]
The $f$-divergence $\Delta^f$ is approximable for any weight function $f$. 
\end{theorem}
\begin{proof}
Consider $h,k \colon X \to \mathcal{G}I$.
Let $|I| = N$.
We may regard $\mathcal{G}I \subseteq [0,1]^N$.
We define a partition $\{C^n_{j_1 \ldots j_{2N}}\}_{j_1,\ldots,j_{2N} \in \{0,1,\ldots, 2^n -1\}}$ of $X$ by
\begin{align*}
C^n_{j_1 \ldots j_{2N}} &= \inverse{h}(B^n_{j_1 \ldots j_N}) \cap \inverse{k}(B^n_{j_{N+1} \ldots j_{2N}})\\
& \qquad \text{ where } B^n_{j_1 \ldots j_N} = A_{j_1} \times \cdots \times A_{j_N}, \quad A^n_0 = \{0\} \text{ and } A^n_{l+1} = \left({l}/{2^n}, {(l+1)}/{2^n} \right].
\end{align*}
We define $J_n = \Set{(j_1,\ldots,j_{2N})\mid j_1,\ldots,j_{2N} \in \{0,1,\ldots, 2^n -1\}, C^n_{j_1 \ldots j_{2N}} \neq \emptyset}$.
We next define $m^\ast_n \colon X \to J_n$ and $m_n \colon J_n \to X$ as follows:
$m^\ast_n(x)$ is the unique element $(j_1,\ldots,j_{2N}) \in J_n$ satisfying $x \in C^n_{j_1,\ldots,j_{2N}}$, and 
$m_n(j_1,\ldots,j_{2N})$ is an element of $C^n_{j_1,\ldots,j_{2N}}$.

From the construction of $\{C^n_{j_1 \ldots j_{2N}}\}_{j_1,\ldots,j_{2N} \in \{0,1,\ldots, 2^n -1\}}$,
for any $n \in \mathbb{N}$, $x \in X$, and $i \in I$, 
\[
|h(x)(i) - (h \circ m_n \circ m^\ast_n)(x)(i)| \leq 2/2^n,\quad
|k(x)(i) - (k \circ m_n \circ m^\ast_n)(x)(i)| \leq 2/2^n
\]
holds.
In particular, for any $i \in I$, the sequences of functions $\{ (h \circ m_n \circ m^\ast_n)(-)(i)\}_{n \in \mathbb{N}}$ and $\{(k \circ m_n \circ m^\ast_n)(-)(i)\}_{n \in \mathbb{N}}$ converge \emph{uniformly} to $h(-)(i)$ and $k(-)(i)$ respectively.
Hence, for any  $\mu_1, \mu_2 \in \mathcal{G}X$, we have
\begin{align*}
h^\sharp(\mu_1)(i) = \int_X h(-)(i)~d\mu_1
&= \lim_{n \to \infty}\int_X (h \circ m_n \circ m^\ast_n)(-)(i)~d\mu_1
 = \lim_{n \to \infty} (h \circ m_n \circ m^\ast_n)^\sharp(\mu_1)(i),\\
g^\sharp(\mu_2)(i) = \int_X k(-)(i)~d\mu_2
&= \lim_{n \to \infty}\int_X (k \circ m_n \circ m^\ast_n)(-)(i)~d\mu_2
=\lim_{n \to \infty} (k \circ m_n \circ m^\ast_n)^\sharp(\mu_2)(i).
\end{align*}
Therefore,
\begin{align*}
\Delta^f_I (h^\sharp(\mu_1),k^\sharp(\mu_2))
& = \sum_{i \in I} g^\sharp(\mu_2)(i) f\left(\frac{h^\sharp(\mu_1)(i)}{g^\sharp(\mu_2)(i)}\right)\\
& = \sum_{i \in I} (\lim_{n \to \infty} (k \circ m_n \circ m^\ast_n)^\sharp(\mu_2)(i)) f\left( \frac{\lim_{n \to \infty} (h \circ m_n \circ m^\ast_n)^\sharp(\mu_1)(i)}{\lim_{n \to \infty} (k \circ m_n \circ m^\ast_n)^\sharp(\mu_2)(i)}\right)\\
&= \lim_{n \to \infty}  \sum_{i \in I} (k \circ m_n \circ m^\ast_n)^\sharp(\mu_2)(i) f\left( \frac{(h \circ m_n \circ m^\ast_n)^\sharp(\mu_1)(i)}{(k \circ m_n \circ m^\ast_n)^\sharp(\mu_2)(i)}\right)\\
&= \lim_{n \to \infty} \Delta^f_I ((h \circ m_n \circ m^\ast_n)^\sharp(\mu_1),(k \circ m_n \circ m^\ast_n)^\sharp(\mu_2)) 
\end{align*}
Remark that the third equality in the above calculation is obtained from the continuity of the weight function $f$.
We then conclude that $\Delta^f$ is approximable.
\end{proof}

\begin{theorem}[Theorem \ref{properties:Renyi}]
  For any $\alpha > 1$, the R\'enyi divergence $D^\alpha$ of order $\alpha$ is reflexive,
  continuous, approximable, composable, and additive (as a singleton-graded
  family).  
\end{theorem}
\begin{proof}
By Theorems \ref{f-divergence:continuity} and \ref{f-divergence:approximability},
the $f$-divergence $\Delta^{\mathrm{R}(\alpha)}$ of the weight function $f(t) = t^\alpha$ is
continuous and approximable. Since the function $g \colon \mathbb{R}_{\leq 0} \to \ol\RR$ defined by
$g(t) =
\frac{1}{\alpha-1}\log(t)$ is monotone and continuous, $D^\alpha = \frac{1}{\alpha-1}\log \Delta^{\mathrm{R}(\alpha)}$ is also continuous and approximable.  Thus, it suffices to show the
reflexivity and \emph{finite-composability} of $D^\alpha$.  The reflexivity is
obvious: $D^\alpha_X(\mu||\mu)_X = \frac{1}{\alpha - 1} \log \mu(X) \leq 0$.  We
show the finite-composability.  Let $I,J\in \FinSet$, $d_1, d_2 \in
\mathcal{G}J$, and $h,k \colon J \to \mathcal{G}I$.  We calculate by Jensen's
inequality:
\begin{align*}
\Delta_{I}^{\mathrm{R}(\alpha)}(h^\sharp d_1,k^\sharp d_2)
&= \sum_{i \in I} \left(\sum_{j \in J}d_2(j)\cdot k(j)(i) \right) \left(\frac{\sum_{j \in J}d_1(j)\cdot h(j)(i)}{\sum_{j \in J}d_2(j)\cdot k(j)(i)} \right)^\alpha\\
&\leq \sum_{j \in J} d_2(j)\left(\frac{d_1(j)}{d_2(j)} \right)^\alpha \sum_{i \in I} k(j)(i)\left(\frac{h(j)(i)}{k(j)(i)} \right)^\alpha\\
&\leq \sum_{j \in J} d_2(j)\left(\frac{d_1(j)}{d_2(j)} \right)^\alpha \cdot \Delta_{\alpha}^{\mathrm{R}(\alpha)}(h(j), k(j))\\
&\leq \Delta_{J}^{\mathrm{R}(\alpha)} (d_1, d_2)\cdot\sup_{j \in J}\Delta_{I}^{\mathrm{R}(\alpha)} (h(j), k(j)).
\end{align*}
This implies
$
D^\alpha_I (h^\sharp d_1||k^\sharp d_2)
\leq D^\alpha_J (d_1||d_2) + \sup_{j \in J}D^\alpha_I(h(j)||k(j))
$.
\end{proof}

\begin{proposition}[Proposition \ref{lem:Renyi:monotonicity}]
  If $1 < \alpha \leq \beta$ then
\[
D^\alpha_X(\mu_1||\mu_2) \leq D^\beta_X(\mu_1||\mu_2).
\]
\end{proposition}
\begin{proof}
The proof is almost the same as \citet[Theorem 3]{6832827}.
Since $D^\alpha$ and $D^\beta$ are continuous, it suffices to prove in finite discrete case.
We denote by $|p|$ the sum $\sum_{i \in I} p_i$.
We may assume $|p| > 0$ since if $|p| = 0$ then $D^\alpha_I(p||q) = D^\beta_I(p||q) = -\infty$.
We remark that the function $t \mapsto t^{\frac{\alpha-1}{\beta-1}}$ is concave.
We have $\left(\frac{1}{|p|}\right)^{\frac{\alpha-1}{\beta-1}} \leq \frac{1}{|p|}$ since $1 \leq \frac{1}{|p|}$.
Therefore,
\begin{align*}
\frac{1}{\alpha - 1} \log \sum_{i \in I} p_i^{\alpha}q_i^{1 - \alpha}
&= \frac{1}{\alpha - 1} \log \left( |p| \sum_{i \in I} \frac{p_i}{|p|} \left(\left(\frac{p_i}{q_i}\right)^{1-\beta}\right)^{\frac{\alpha-1}{\beta-1}}\right)\\
&\leq \frac{1}{\alpha - 1} \log  \left( |p|  \left(\sum_{i \in I} \frac{p_i}{|p|} \left(\frac{p_i}{q_i}\right)^{1-\beta}\right)^{\frac{\alpha-1}{\beta-1}}\right)\\
&\leq \frac{1}{\alpha - 1} \log \left(|p| \cdot \sum_{i \in I} \frac{p_i}{|p|} \left(\frac{p_i}{q_i}\right)^{1-\beta}\right)^{\frac{\alpha-1}{\beta-1}}\\
&= \lefteqn{\frac{1}{\beta - 1} \log \sum_{i \in I} p_i^{\beta}q_i^{1 - \beta}}
\end{align*}
This completes the proof.
\end{proof}

\begin{proposition}[Proposition \ref{lem:Renyi:weak-triangle}]
For any $\alpha > 1$, $\mu_1,\mu_2, \mu_3 \in \mathcal{G}X$, and  $p,q > 1$ satisfying $\frac{1}{p}+\frac{1}{q} = 1$, we have
\[
D^\alpha_X(\mu_1||\mu_3) \leq \frac{p\alpha -1 }{p(\alpha - 1)}D^{p\alpha}_X(\mu_1||\mu_2) + D^{\frac{q}{p}(p\alpha - 1)}_X(\mu_1||\mu_2)
.\]
\end{proposition}
\begin{proof}
Recall that if $\mu_1 \not\ll \mu_2$ then $D^\alpha_X(\mu_1||\mu_2) = \infty$.
Hence, we may assume $\mu_1 \ll \mu_2 \ll \mu_3$ without loss of generality 
(if not so, the right-hand side should be infinity).
By chain rule of Radon-Nikodym derivative and H\"{o}lder's inequality, 
\begin{align*}
\Delta^{\mathrm{R}(\alpha)}_X(\mu_1,\mu_3)
&= \int_X \left(\frac{d\mu_1}{d\mu_3}\right)^\alpha d\mu_3\\
&= \int_X \left(\frac{d\mu_1}{d\mu_2} \cdot \frac{d\mu_2}{d\mu_3}\right)^\alpha d\mu_3\\
&= \int_X \left(\frac{d\mu_1/d\mu_3}{d\mu_2/d\mu_3}\right)^\alpha \cdot \left(\frac{d\mu_2}{d\mu_3}\right)^{\frac{1}{p}} \cdot \left(\frac{d\mu_2}{d\mu_3}\right)^{\alpha-\frac{1}{p}} d\mu_3\\
&\leq \left(
	\int_X \left( \frac{d\mu_1/d\mu_3}{d\mu_2/d\mu_3} \right)^{p\alpha}
		\cdot \left( \frac{d\mu_2}{d\mu_3} \right)
	d\mu_3
	\right)^{\frac{1}{p}}
\cdot \left(\int_X
		\left(\frac{d\mu_2}{d\mu_3}\right)^{q(\alpha-\frac{1}{p})} d\mu_3
	\right)^{\frac{1}{q}} \\
& =
\Delta^{\mathrm{R}(p\alpha)}_X (\mu_1 || \mu_2)^{\frac{1}{p}} \cdot
\Delta^{\mathrm{R}(q\alpha-\frac{q}{p})}_X (\mu_2 || \mu_3)^{\frac{1}{q}}
\end{align*}
We then conclude 
$D^\alpha_X(\mu_1||\mu_3) \leq \frac{p\alpha -1 }{p(\alpha - 1)}D^{p\alpha}_X(\mu_1||\mu_2) + D^{\frac{q}{p}(p\alpha - 1)}_X(\mu_1||\mu_2)$.
\end{proof}

\begin{theorem}[Theorem \ref{properties:zCDP}]
The $\mathbb{R}_{\geq 0}$-graded family $\mathbf{\Delta^{\mathtt{zCDP}}}=\{\Delta^{\mathtt{zCDP}(\xi)}\}_{0 \leq \xi}$ 
is reflexive, continuous, composable, and additive.
\end{theorem}
\begin{proof}
Consider any $\alpha > 1$. We consider a $(\mathbb{R}_{\geq 0},+,0,\leq)$-graded family $\mathbf{\Delta}^{\mathtt{zCDP+}(\alpha)} = \{\Delta^{\mathtt{zCDP+}(\xi,\alpha)}\}_{\xi \in \mathbb{R}_{\geq 0}}$ of the following divergences:
\[
\Delta^{\mathtt{zCDP+}(\xi,\alpha)}_X(\mu_1,\mu_2)
=\frac{1}{\alpha} \left(D^\alpha(\mu_1||\mu_2)) - \xi \right).
\]
By the previous theorem \ref{properties:Renyi}, this family is reflexive and continuous for any $\alpha > 1$.
The composability of the  family $\mathbf{\Delta}^{\mathtt{zCDP+}(\alpha)} = \{\Delta^{\mathtt{zCDP+}(\xi,\alpha)}\}_{\xi \in \mathbb{R}_{\geq 0}}$
is the direct consequence of the composability of $\alpha$-R\'enyi divergence:
for any $\mu_1,\mu_2 \in \mathcal{G}X$, and $f,g \colon X \to \mathcal{G}Y$,
\[
\frac{1}{\alpha}(D^\alpha_X(f^\sharp(\mu_1)||g^\sharp(\mu_2)) - (\xi_1+\xi_2))
\leq 
\frac{1}{\alpha}(D^\alpha_X(\mu_1||\mu_2) - \xi_1)
+
\sup_{x\in X} \frac{1}{\alpha}(D^\alpha_Y(f(x)||g(y)) - \xi_2).
\]
Since $\Delta^{\mathtt{zCDP}(\xi)} = \sup_{\alpha > 1}
\Delta^{\mathtt{zCDP+}(\xi,\alpha)}$, the graded family
$\mathbf{\Delta^{\mathtt{zCDP}}}=\{\Delta^{\mathtt{zCDP}(\xi)}\}_{0 \leq \xi}$
is reflexive, continuous, and composable.\footnote{%
  We obtain these properties from commutativity $\sup_{y \in Y} \sup_{x \in X}
  f(x,y) = \sup_{x \in X} \sup_{y \in Y}f(x,y)$ of supremums.  We drop the
  approximability, which is not given by a supremum but rather by a limit.}
The additivity is obtained from Theorem \ref{thm:div:cont+comp->mon}.
\end{proof}

\subsection{Detailed Construction and Proof of Well-definedness of Approximate Span-lifting}

\begin{definition}[Functors]
If the family $\mathbf{\Delta}$ is \emph{functorial} then
the structure of \emph{endofunctor} on $\Span(\Meas)$ of the approximate span-lifting $(-)^{\sharp(\mathbf{\Delta},\alpha,\delta)}$ is given as follows:
for all $\alpha \in A$, $\delta \in \ol\RR$, and 
$(h,k,l) \colon (X,Y,\Phi,\rho_1,\rho_2) \to (X',Y',\Psi,\rho'_1,\rho'_2)$ in $\Span(\Meas)$,
\begin{equation}\label{eq:span-lifting:functor_part}
(\mathcal{G}h,\mathcal{G}k,(\mathcal{G}l \times \mathcal{G}l)|_{W(\Phi,\mathbf{\Delta},\alpha,\delta)})
\colon (X,Y,\Phi,\rho_1,\rho_2)^{\sharp(\mathbf{\Delta},\alpha,\delta)}
\to (X',Y',\Psi,\rho'_1,\rho'_2)^{\sharp(\mathbf{\Delta},\alpha,\delta)}.
\end{equation}
\end{definition}
\begin{theorem}[Well-definedness]
If $\mathbf{\Delta}$ is functorial then the above structure $(-)^{\sharp(\mathbf{\Delta},\alpha,\delta)}$ forms indeed an endofunctor on $\Span(\Meas)$.
\end{theorem}
\begin{proof}
We first show the well-definedness of (\ref{eq:span-lifting:functor_part}).
We fix $(h,k,l) \colon (X,Y,\Phi,\rho_1,\rho_2) \to (X',Y',\Psi,\rho'_1,\rho'_2)$ in $\Span(\Meas)$ and parameters 
$\alpha \in A$ and $\delta \in \ol\RR$.
Let $(\nu_1,\nu_2)\in W(\Phi,\mathbf{\Delta},\alpha,\delta)$. 
The pair satisfies $\Delta^\alpha_\Phi(\nu_1,\nu_2) \leq \delta$.
Since the divergence $\Delta^\alpha$ is functorial, we have 
$\Delta^\alpha_\Psi(\mathcal{G}(l)(\nu_1),\mathcal{G}(l)(\nu_2)) \leq \delta$.
Thus, $(\mathcal{G}l \times \mathcal{G}l)|_{W(\Phi,\mathbf{\Delta},\alpha,\delta)}$ is a measurable function from 
$W(\Phi,\mathbf{\Delta},\alpha,\delta)$ to
$W(\Psi,\mathbf{\Delta},\alpha,\delta)$.\footnote{%
  Strictly speaking, we consider the function $W(\Phi,\mathbf{\Delta},\alpha,\delta)
  \xrightarrow{(\mathcal{G}l \times
  \mathcal{G}l)|_{W(\Phi,\mathbf{\Delta},\alpha,\delta)}} (\text{Image})
  \xrightarrow{\text{inclusion}} W(\Phi,\mathbf{\Delta},\alpha,\delta)$ through
  the image $(\text{Image})$.  Functoriality shows the existence of the
inclusion.}
Since $\mathcal{G}$ is a functor on $\Meas$, we obtain,
\begin{align*}
\mathcal{G}\rho'_1\circ \pi_1 \circ (\mathcal{G}l \times \mathcal{G}l)|_{W(\Phi,\mathbf{\Delta},\alpha,\delta)}
= \mathcal{G}\rho'_1\circ \mathcal{G}l \circ \pi_1|_{W(\Phi,\mathbf{\Delta},\alpha,\delta)}
&= \mathcal{G}h \circ \mathcal{G}\rho_1 \circ \pi_1|_{W(\Phi,\mathbf{\Delta},\alpha,\delta)},\\
\mathcal{G}\rho'_2 \circ \pi_2\circ (\mathcal{G}l \times \mathcal{G}l)|_{W(\Phi,\mathbf{\Delta},\alpha,\delta)}
= \mathcal{G}\rho'_2\circ \mathcal{G}l \circ \pi_2|_{W(\Phi,\mathbf{\Delta},\alpha,\delta)}
&= \mathcal{G}k \circ \mathcal{G}\rho_2 \circ \pi_2|_{W(\Phi,\mathbf{\Delta},\alpha,\delta)}.
\end{align*}
Thus, the construction (\ref{eq:span-lifting:functor_part}) is a mapping on $\Span(\Meas)$-morphisms.

The functoriality is obvious by definition.
\end{proof}

\begin{definition}[Graded monad structures]
If the family $\mathbf{\Delta}$ is \emph{reflexive and composable} then 
the structure of $A \times (\ol\RR,+,0,\leq)$-\emph{graded monad} on
$\Span(\Meas)$ is given as follows.
\begin{description}
  \item[Unit:] for any span $(X,Y,\Phi,\rho_1,\rho_2)$, we define
\begin{equation}\label{eq:span-lifting:unit}
(\eta_X,\eta_Y,\langle \eta_\Phi,\eta_\Phi \rangle) \colon
(X,Y,\Phi,\rho_1,\rho_2) \to (X,Y,\Phi,\rho_1,\rho_2)^{\sharp(\mathbf{\Delta},1_A,0)}.
\end{equation}
\item[Kleisli extensions:] for any morphism $(h,k,l) \colon(X,Y,\Phi,\rho_1,\rho_2) \to(X',Y',\Psi,\rho'_1,\rho'_2)^{\sharp(\mathbf{\Delta},\alpha,\delta)}$ in $\Span(\Meas)$, 
  we define
\begin{align}\label{eq:span-lifting:Kleisli}
\lefteqn{
(h^\sharp, k^\sharp, ((\pi_1|_{W(\Psi,\mathbf{\Delta},\alpha,\delta)} \circ l)^\sharp \times (\pi_2|_{W(\Psi,\mathbf{\Delta},\alpha,\delta)} \circ l)^\sharp)|_{W(\Phi,\mathbf{\Delta},\beta,\gamma)})\colon}\notag\\
&\qquad\qquad\qquad\qquad (X,Y,\Phi,\rho_1,\rho_2)^{\sharp(\mathbf{\Delta},\beta,\gamma)} \to (X',Y',\Psi,\rho'_1,\rho'_2)^{\sharp(\mathbf{\Delta},\alpha\beta,\delta+\gamma)}
\end{align}
\item[Inclusions:] for any $\alpha \preceq \beta$, $\delta \leq \gamma$, and
  $(X,Y,\Phi,\rho_1,\rho_2)$ in $\Span(\Meas)$, we define
\begin{equation}\label{eq:span-lifting:inclusion}
(\mathrm{id}_{\mathcal{G}X},\mathrm{id}_{\mathcal{G}Y}, (\mathrm{id}_{\mathcal{G}\Phi} \times \mathrm{id}_{\mathcal{G}\Phi})|_{W(\Phi,\mathbf{\Delta},\alpha,\delta)})
\colon
(X,Y,\Phi,\rho_1,\rho_2)^{\sharp(\mathbf{\Delta},\alpha,\delta)}\hspace{-1em} \to (X,Y,\Phi,\rho_1,\rho_2)^{\sharp(\mathbf{\Delta},\beta,\gamma)}.
\end{equation}
  \end{description}
\end{definition}
We remark here that each $(-)^{\sharp(\mathbf{\Delta},\alpha,\delta)}$ is also an endofunctor because $\mathbf{\Delta}$ is the functorial since it is both reflexive and composable.
\begin{theorem}[Well-definedness]
If $\mathbf{\Delta}$ is reflexive and composable then the above structures $(-)^{\sharp(\mathbf{\Delta},\alpha,\delta)}$ form indeed an $A \times \ol\RR$-graded monad on $\Span(\Meas)$.
\end{theorem}
\begin{proof}
We first prove that the components are well-defined.
\begin{description}
\item[Unit:]
We show the well-definedness of (\ref{eq:span-lifting:unit}).
We fix $(X,Y,\Phi,\rho_1,\rho_2)$ in $\Span(\Meas)$.
For any 
$\phi \in \Phi$, we have 
$\langle \eta_\Phi,\eta_\Phi \rangle(\phi) = (\mathbf{d}_\phi,\mathbf{d}_\phi)$.
Since $\mathbf{\Delta}$ is reflexive, we have $\Delta^{1_A}(\mathbf{d}_\phi,\mathbf{d}_\phi) \leq 0$.
Thus, $\langle \eta_\Phi,\eta_\Phi \rangle$ is indeed a measurable function from $(X,Y,\Phi,\rho_1,\rho_2)$ to $W(\Phi,\mathbf{\Delta},1_A,0)$.
Since $\eta$ is a unit of the sub-Giry monad $\mathcal{G}$, 
we obtain
\begin{align*}
\mathcal{G}\rho_1 \circ \pi_1|_{W(\Phi,\mathbf{\Delta},1_A,0)} \circ \langle\eta_\Phi,\eta_\Phi \rangle
= \mathcal{G}\rho_1 \circ \eta_\Phi
&= \eta_X \circ \rho_1,\\
\mathcal{G}\rho_2 \circ \pi_2|_{W(\Phi,\mathbf{\Delta},1_A,0)} \circ \langle\eta_\Phi,\eta_\Phi \rangle
= \mathcal{G}\rho_2 \circ \eta_\Phi
&= \eta_Y \circ \rho_2.
\end{align*}
Thus (\ref{eq:span-lifting:unit}) is well-defined.
\item[Kleisli extensions:]
We show the well-definedness of (\ref{eq:span-lifting:Kleisli}).
We fix a $\Span(\Meas)$-morphism
\[
  (h,k,l) \colon(X,Y,\Phi,\rho_1,\rho_2) \to(X',Y',\Psi,\rho'_1,\rho'_2)^{\sharp(\mathbf{\Delta},\alpha,\delta)}
\]
and parameters $\beta \in A$ and $\gamma \in \ol\RR$.  
For any $\phi \in \Phi$, we have
$\Delta_{\Psi}^\alpha (\pi_1|_{W(\Psi,\mathbf{\Delta},\alpha,\delta)} \circ l(\phi),\pi_2|_{W(\Psi,\mathbf{\Delta},\alpha,\delta)} \circ l(\phi)) \leq \delta$.
Since $\mathbf{\Delta}$ is composable, we have
for any $(\nu_1,\nu_2) \in W(\Phi,\mathbf{\Delta},\delta,\gamma)$, 
\[
\Delta_{\Psi}^{\alpha\beta} ((\pi_1|_{W(\Psi,\mathbf{\Delta},\alpha,\delta)} \circ l)^\sharp(\nu_1),(\pi_2|_{W(\Psi,\mathbf{\Delta},\alpha,\delta)} \circ l)^\sharp(\nu_2)) \leq \delta+\gamma
\]
This implies that $ ((\pi_1|_{W(\Psi,\mathbf{\Delta},\alpha,\delta)} \circ l)^\sharp \times (\pi_2|_{W(\Psi,\mathbf{\Delta},\alpha,\delta)} \circ l)^\sharp)|_{W(\Phi,\mathbf{\Delta},\beta,\gamma)}$ is indeed a measurable function from $W(\Phi,\mathbf{\Delta},\beta,\gamma)$ to $W(\Psi,\mathbf{\Delta},\alpha\beta,\delta+\gamma)$.
Since $(-)^\sharp$ is the Kleisli lifting of the sub-Giry monad, we obtain
\begin{align*}
\lefteqn{\mathcal{G}\rho'_1 \circ \pi_1|_{W(\Psi,\mathbf{\Delta},\alpha\beta,\delta+\gamma)} \circ  ((\pi_1|_{W(\Psi,\mathbf{\Delta},\alpha,\delta)} \circ l)^\sharp \times (\pi_2|_{W(\Psi,\mathbf{\Delta},\alpha,\delta)} \circ l)^\sharp)|_{W(\Phi,\mathbf{\Delta},\beta,\gamma)}}\\
& = \mathcal{G}\rho'_1 \circ (\pi_1|_{W(\Psi,\mathbf{\Delta},\alpha,\delta)} \circ l)^\sharp \circ \pi_1|_{W(\Phi,\mathbf{\Delta},\beta,\gamma)}
= (\mathcal{G}\rho'_1 \circ \pi_1|_{W(\Psi,\mathbf{\Delta},\alpha,\delta)} \circ l)^\sharp \circ \pi_1|_{W(\Phi,\mathbf{\Delta},\beta,\gamma)}\\
&= h^\sharp \circ \mathcal{G}\rho_1\circ \pi_1|_{W(\Phi,\mathbf{\Delta},\beta,\gamma)}\\
\lefteqn{\mathcal{G}\rho'_2 \circ \pi_2|_{W(\Psi,\mathbf{\Delta},\alpha\beta,\delta+\gamma)} \circ  ((\pi_1|_{W(\Psi,\mathbf{\Delta},\alpha,\delta)} \circ l)^\sharp \times (\pi_2|_{W(\Psi,\mathbf{\Delta},\alpha,\delta)} \circ l)^\sharp)|_{W(\Phi,\mathbf{\Delta},\beta,\gamma)}}\\
&= k^\sharp \circ \mathcal{G}\rho_2\circ \pi_2|_{W(\Phi,\mathbf{\Delta},\beta,\gamma)}
\end{align*}
Thus (\ref{eq:span-lifting:Kleisli}) is well-defined.
\item[Inclusions:]
We show the well-definedness of (\ref{eq:span-lifting:inclusion}).
We fix $(X,Y,\Phi,\rho_1,\rho_2)$ in $\Span(\Meas)$ and parameters $\alpha \preceq \beta$ and $\delta \leq \gamma$.
Since $\mathbf{\Delta}$ is an $A$-graded family of divergences, we have $\Delta^\beta \leq \Delta^\alpha$.
This implies that there is the inclusion function $W(\Phi,\mathbf{\Delta},\alpha,\delta) \hookrightarrow W(\Phi,\mathbf{\Delta},\beta,\gamma)$ in $\Meas$.
Hence, by treating the restrictions of functions, we obtain
\begin{align*}
\mathrm{id}_{\mathcal{G}X} \circ \mathcal{G}\rho_1 \circ \pi_1|_{W(\Phi,\mathbf{\Delta},\alpha,\delta)} &= \mathcal{G}\rho_1 \circ \pi_1 \circ (\mathrm{id}_{\mathcal{G}\Phi} \times \mathrm{id}_{\mathcal{G}\Phi})|_{W(\Phi,\mathbf{\Delta},\alpha,\delta)}\\
\mathrm{id}_{\mathcal{G}Y} \circ \mathcal{G}\rho_2 \circ \pi_2|_{W(\Phi,\mathbf{\Delta},\alpha,\delta)} &= \mathcal{G}\rho_2 \circ \pi_2 \circ (\mathrm{id}_{\mathcal{G}\Phi} \times \mathrm{id}_{\mathcal{G}\Phi})|_{W(\Phi,\mathbf{\Delta},\alpha,\delta)}.
\end{align*}
Therefore (\ref{eq:span-lifting:inclusion}) is well defined.
\end{description}
Therefore, the components of graded monad structures are well-defined.
It is easy to check the axioms of graded monad in \citet[Definition 2.3]{Katsumata2014PEM} by using monad structure of the sub-Giry monad $\mathcal{G}$ since the graded monad structure of the approximate span-lifting is given by using the monad structure of $\mathcal{G}$ and restrictions.
\end{proof}

\begin{definition}[Double strength]
If the family $\mathbf{\Delta}$ is reflexive, composable, and \emph{additive}
then a \emph{double strength} of the graded monad $(-)^{\sharp(\mathbf{\Delta},\alpha,\delta)}$ is given as follows:
for any pair $(X,Y,\Phi,\rho_1,\rho_2)$ and $(X',Y',\Psi,\rho'_1,\rho'_2)$ of spans,
\begin{align}\label{eq:span-lifting:monoidality}
\lefteqn{(\mathrm{dst}_{X,X'},\mathrm{dst}_{Y,Y'},\langle \mathrm{dst}_{\Phi,\Psi}\circ(\pi_1 \times \pi_1), \mathrm{dst}_{\Phi,\Psi}\circ(\pi_2 \times \pi_2)\rangle|_{W(\Phi,\mathbf{\Delta},\alpha,\delta) \times W(\Psi,\mathbf{\Delta},\beta,\gamma)})\colon}\notag\\
&\qquad\qquad
(X,Y,\Phi,\rho_1,\rho_2)^{\sharp(\mathbf{\Delta},\alpha,\delta)} \mathbin{\dot\times} (X',Y',\Psi,\rho'_1,\rho'_2)^{\sharp(\mathbf{\Delta},\beta,\gamma)} \to (\Phi \mathbin{\dot\times} \Psi)^{\sharp(\mathbf{\Delta},\alpha\beta,\delta+\gamma)}.
\end{align}
\end{definition}

\begin{theorem}[Well-definedness (Theorem \ref{thm:span-lifting:conclusion})]
If $\mathbf{\Delta}$ is reflexive, composable, and additive then the above structure forms indeed a double strength of the graded monad $(-)^{\sharp(\mathbf{\Delta},\alpha,\delta)}$ on $\Span(\Meas)$.
\end{theorem}
\begin{proof}
Since $\mathbf{\Delta}$ is reflexive and composable,  $(-)^{\sharp(\mathbf{\Delta},\alpha,\delta)}$ forms an $A \times \ol\RR$-graded monad on $\Span(\Meas)$.
We show the well-definedness of (\ref{eq:span-lifting:monoidality}).
We fix spans $(X,Y,\Phi,\rho_1,\rho_2)$ and $(X',Y',\Psi,\rho'_1,\rho'_2)$ and parameters $\alpha,\beta \in A$ and 
$\gamma,\delta \in \ol\RR$.
Since $\mathbf{\Delta}$ is additive,
$\langle \mathrm{dst}_{\Phi,\Psi}\circ (\pi_1 \times \pi_1), \mathrm{dst}_{\Phi,\Phi'}\circ(\pi_2 \times \pi_2) \rangle|_{W(\Phi,\mathbf{\Delta},\alpha,\delta) \times W(\Psi,\mathbf{\Delta},\beta,\gamma)}$ is indeed a measurable function from $W(\Phi,\mathbf{\Delta},\alpha,\delta) \times W(\Psi,\mathbf{\Delta},\beta,\gamma)$ to $W(\Phi\mathbin{\dot\times}\Psi,\mathbf{\Delta},\alpha\beta,\delta+\gamma)$.
From the binaturality of the double strength $\mathrm{dst}$ of the sub-Giry monad $\mathcal{G}$, we have
\begin{align*}
& 
\mathcal{G}(\rho_1\times \rho'_1) \circ \pi_1|_{W(\Phi\mathbin{\dot\times}\Psi,\mathbf{\Delta},\alpha\beta,\delta+\gamma)} \circ
\langle \mathrm{dst}_{\Phi,\Psi}\circ(\pi_1 \times \pi_1), \mathrm{dst}_{\Phi,\Psi}\circ(\pi_2 \times \pi_2) \rangle|_{W(\Phi,\mathbf{\Delta},\alpha,\delta) \times W(\Psi,\mathbf{\Delta},\beta,\gamma)}
\\
&=
\mathcal{G}(\rho_1\times \rho'_1) \circ \mathrm{dst}_{\Phi,\Psi}\circ(\pi_1 \times \pi_1)|_{W(\Phi,\mathbf{\Delta},\alpha,\delta) \times W(\Psi,\mathbf{\Delta},\beta,\gamma)}\\
&=
\mathcal{G}(\rho_1\times \rho'_1) \circ \mathrm{dst}_{\Phi,\Psi}\circ(\pi_1|_{W(\Phi,\mathbf{\Delta},\alpha,\delta)} \times \pi_1|_{W(\Psi,\mathbf{\Delta},\beta,\gamma)})\\
&=
\mathrm{dst}_{X,X'}\circ
( (\mathcal{G}\rho_1 \circ \pi_1|_{W(\Phi,\mathbf{\Delta},\alpha,\delta)}) \times 
(\mathcal{G}\rho'_1 \circ\pi_1|_{ W(\Psi,\mathbf{\Delta},\beta,\gamma)})),\\
&
\mathcal{G}(\rho_2\times \rho'_2) \circ \pi_1|_{W(\Phi\mathbin{\dot\times}\Psi,\mathbf{\Delta},\alpha\beta,\delta+\gamma)} \circ
\langle \mathrm{dst}_{\Phi,\Psi}\circ(\pi_1 \times \pi_1), \mathrm{dst}_{\Phi,\Psi}\circ(\pi_2 \times \pi_2) \rangle|_{W(\Phi,\mathbf{\Delta},\alpha,\delta) \times W(\Psi,\mathbf{\Delta},\beta,\gamma)}\\
&= \mathrm{dst}_{Y,Y'}\circ
( (\mathcal{G}\rho_2 \circ \pi_1|_{W(\Phi,\mathbf{\Delta},\alpha,\delta)}) \times 
(\mathcal{G}\rho'_2 \circ\pi_1|_{ W(\Psi,\mathbf{\Delta},\beta,\gamma)})).\\
\end{align*}
Hence, (\ref{eq:span-lifting:monoidality}) is well-defined.
It is easy to check the axioms of double strength (modulo grading)
by using double strength of the sub-Giry monad $\mathcal{G}$.
\end{proof}
\end{document}